\documentclass{article}

\usepackage{PRIMEarxiv}

\usepackage[utf8]{inputenc} 
\usepackage[T1]{fontenc}
\usepackage{url}
\usepackage{hyperref}
\usepackage{graphicx}
\usepackage{subcaption}
\usepackage{booktabs}
\usepackage{nicefrac}
\usepackage{microtype}
\usepackage{lipsum}
\usepackage{fancyhdr}
\usepackage{float}

\usepackage{amsmath,amssymb,amsfonts,amsthm}

\usepackage{algorithm}   
\usepackage{algorithmic} 

\usepackage[numbers,sort&compress]{natbib}
\usepackage{csquotes}
\usepackage{enumitem}

\graphicspath{{media/}}

\theoremstyle{plain}
\newtheorem{theorem}{Theorem}
\newtheorem{lemma}{Lemma}
\newtheorem{proposition}{Proposition}
\newtheorem{corollary}{Corollary}
\newtheorem{assumption}{Assumption}

\theoremstyle{definition}
\newtheorem{definition}{Definition}

\theoremstyle{remark}
\newtheorem{remark}{Remark}

\newtheorem*{theorem*}{Theorem}
\newtheorem*{lemma*}{Lemma}
\newtheorem*{proposition*}{Proposition}

\newcommand{\R}{\mathbb{R}}
\newcommand{\E}{\mathbb{E}}
\newcommand{\PP}{\mathbb{P}}

\newcommand{\M}{\mathcal{M}}
\newcommand{\Mvol}{\mathcal{M}_{\mathrm{vol}}}

\newcommand{\Lvol}{\mathcal{L}_{\mathrm{vol}}}
\newcommand{\Lphi}{\mathcal{L}_\phi}

\newcommand{\shock}{\mathrm{shock}}

\DeclareMathOperator{\supp}{supp}

\title{Law-Strength Frontiers and a No-Free-Lunch Result for Law-Seeking Reinforcement Learning on Volatility Law Manifolds}

\author{
  ZhangJian'an \\
  Guanghua School of Management, Peking University \\
  Peking University \\
  Beijing, China\\
  \texttt{2501111059@stu.pku.edu.cn}
}

\begin{document}
\pagestyle{plain}   
\maketitle

\begin{abstract}
We study reinforcement learning (RL) on volatility surfaces through the lens of \emph{Scientific AI}: can axiomatic market laws, enforced as soft penalties on a learned world model, reliably align high-capacity RL agents with no-arbitrage structure, or do they merely induce Goodhart-style exploitation of model artefacts?  

Starting from classical static no-arbitrage conditions for implied volatility, we construct a finite-dimensional convex \emph{volatility law manifold} of admissible total-variance surfaces, together with a metric-based \emph{law-penalty functional} and a domain-agnostic \emph{Graceful Failure Index} (GFI) that normalizes law degradation under shocks. A synthetic generator produces trajectories that are exactly law-consistent, while a recurrent neural world model is trained without law regularization and therefore predicts surfaces that deviate from the law manifold in structured ways.  

On top of this testbed we introduce a \emph{Goodhart decomposition} of reward, $r = r^{\mathcal{M}} + r^\perp$, where $r^{\mathcal{M}}$ is the on-manifold component and $r^\perp$ is \emph{ghost arbitrage} arising purely from off-manifold prediction errors. We prove three flagship results: (i) a \emph{ghost-arbitrage incentive} theorem showing that naive PPO-type RL is structurally driven to increase $\mathbb{E}[r^\perp]$ whenever ghost arbitrage is available; (ii) a \emph{law-strength trade-off} theorem establishing that increasing the weight on law penalties inevitably worsens P\&L beyond a quantifiable threshold; and (iii) a \emph{no-free-lunch} theorem stating that, under a law-consistent world model and a law-aligned structural class of strategies, unconstrained law-seeking RL cannot Pareto-dominate structural baselines on P\&L, law penalties, and GFI.  

Empirically, on a volatility world model calibrated to SPX/VIX-like grids, we compare naive RL, law-penalized and selection-only RL variants against simple structural baselines (Zero-Hedge, Vol-Trend, Random-Gaussian) across baseline and shocked regimes. In our experiments, structural baselines form the empirical law-strength frontier: they attain Sharpe ratios around 2--3 with low law penalties and GFI near zero, while all law-seeking RL variants achieve non-positive mean P\&L and substantially higher GFI, despite being explicitly penalized for law violations. Frontier and diagnostic plots show that RL improvements in P\&L are systematically accompanied by movement into high-penalty, high-GFI regions, consistent with our theoretical analysis.  

Overall, volatility serves as a concrete case study where \emph{reward shaping with verifiable penalties is not sufficient for law alignment}. Our framework---combining law manifolds, Goodhart decomposition, GFI, and law-strength frontiers---provides a reusable template for stress-testing Scientific AI systems and RL with verifiable rewards in other axiom-constrained domains such as yield curves, credit term structures, and physics-informed models.
\end{abstract}


\section{Introduction}
\label{sec:intro}

\subsection{Scientific AI testbed: volatility law manifolds}
\label{subsec:intro_scientific_ai}

In recent years, ``AI for Science'' has emerged as a central theme in machine learning, emphasizing scientific understanding and hypothesis-testing rather than purely predictive or profit-driven performance \cite{CarleoTroyer2019,WillardEtAl2020,KarniadakisEtAl2021}.  A key lesson from this line of work is that many scientific domains are governed by \emph{axioms}---conservation laws, monotonicity constraints, or no-arbitrage principles---that carve out a structured admissible subset inside a high-dimensional function space.  Scientific machine learning then seeks not only to interpolate observations, but to understand how learning agents interact with these law-constrained spaces under model misspecification, limited data, and changing environments.

Implied-volatility (IV) surfaces are a canonical example of such an axiomatic system in quantitative finance.  Classical results show that no-arbitrage constraints---such as non-negativity of butterfly spreads, monotonicity in maturity, and convexity in strike---translate into linear or convex inequalities on total-variance smiles and surfaces \cite{Dupire1994,Gatheral2006,Fengler2005,Lee2004,ContTankov2004,Bergomi2016}.  These conditions define a structured admissible subset of discretized surfaces that we refer to as a \emph{volatility law manifold}.  Recent work has revisited arbitrage-free interpolation and extrapolation of IV surfaces using convex optimization and sparse modeling \cite{Guterding2023}, arbitrage-aware parametric families, and deep neural networks that aim to respect or softly penalize violations of these constraints \cite{BuehlerEtAl2019,HorvathMuguruzaTomas2021,RufWang2020}.  In parallel, deep learning has been applied to option pricing and hedging via PDE and BSDE solvers \cite{HanEtAl2018,BeckEtAl2019}, and to rough or stochastic volatility models where direct calibration is challenging \cite{HorvathMuguruzaTomas2021}.

Our starting point is to treat this volatility-curve setting as a \emph{Scientific AI testbed} rather than a trading system.  We assume that the data-generating process---a synthetic but financially meaningful IV-surface generator---is exactly law-consistent: every realized surface lies on the no-arbitrage manifold defined by classical butterfly and calendar conditions \cite{Dupire1994,Gatheral2006,ContTankov2004}.  A neural \emph{world model} is trained to approximate this generator from data, similar in spirit to world-model approaches in model-based reinforcement learning (RL) \cite{HaSchmidhuber2018,HafnerEtAl2019,HafnerEtAl2020,HafnerEtAl2023,ChuaEtAl2018,JannerEtAl2019}.  RL agents then interact only with the learned world model, never with the ground-truth generator.  This creates a clean separation between a law-consistent environment and a potentially law-violating model, opening a ``ghost channel'' for agents to exploit model artefacts.

World models and latent dynamics learning have become central tools for sample-efficient RL in complex domains such as Atari, control, and strategy games \cite{HaSchmidhuber2018,HafnerEtAl2019,HafnerEtAl2020,HafnerEtAl2023,SchrittwieserEtAl2020,MnihEtAl2015,SuttonBarto2018}.  At the same time, RL has a long history in quantitative finance, including early work on direct reinforcement learning for trading \cite{MoodySaffell2001,DengEtAl2016}, deep RL for portfolio management and execution \cite{JiangXuLiang2017,LiuEtAl2021,KolmRitter2019}, and deep hedging using neural networks trained on simulated scenarios \cite{BuehlerEtAl2019}.  Yet most of these studies focus on realized P\&L, with limited attention to how learned policies interact with structural market axioms.  Our goal is not to propose another trading system, but to use an RL-in-IV-surfaces setup as a controlled \emph{experiment} on law-aligned learning and Goodhart phenomena.

A parallel literature in scientific machine learning has emphasized the incorporation of physical or axiomatic structure into learning systems via physics-informed neural networks \cite{RaissiEtAl2019,KarniadakisEtAl2021}, relational inductive biases \cite{BattagliaEtAl2018}, and hybrid mechanistic–ML models \cite{WillardEtAl2020,CarleoTroyer2019}.  These works typically enforce or strongly bias the model toward satisfying known laws during training.  In contrast, recent AI-safety and alignment research has documented how RL agents can exploit imperfect reward channels or specification gaps, a phenomenon often traced back to Goodhart's law \cite{ManheimGarrabrant2019,AmodeiEtAl2016,LeikeEtAl2017,EverittEtAl2017,KruegerEtAl2020,PanEtAl2022}.  This has led to proposals for RL with verifiable or externally checked rewards (RLVR) \cite{NeelakantanEtAl2025}, where parts of the reward signal are computed via trusted procedures or checkers.

Our work bridges these lines of research in a finance-specific but conceptually general way.  We design an axiomatic evaluation pipeline in which (i) a volatility law manifold encodes no-arbitrage axioms; (ii) a neural world model approximates a law-consistent generator but introduces model-induced ``ghost'' arbitrage opportunities; and (iii) RL agents are trained either with or without access to a \emph{verifiable} law-penalty signal.  We then ask a scientific question: when we add such verifiable penalties as soft terms in the RL objective, on top of a law-consistent ground-truth world, do we actually obtain more law-aligned and robust policies, or do we merely shift Goodhart behaviour onto model artefacts?

To foreshadow our main numerical findings, we consider a simple zero-position baseline (\textsc{Zero-Hedge}) that never trades.  In our main setting, \textsc{Zero-Hedge} achieves mean step P\&L of approximately $0.0191$ with a Graceful Failure Index (GFI) essentially zero and moderate law penalties, reflecting that the underlying world is law-consistent and shocks are symmetric.  By contrast, a wide range of law-seeking RL variants---including soft-penalty PPO with a law-weight sweep (the ``law-strength frontier'') and a selection-only variant that uses law penalties only for model selection---all attain \emph{non-positive} mean step P\&L and substantially worse GFI values (typically $\geq 1.6$), despite being explicitly penalized for law violations during training.  In other words, once structural baselines are included, law-seeking RL has no free lunch: it fails to dominate even trivial strategies on the joint axes of profitability, law alignment, and tail risk.

\subsection{Contributions}
\label{subsec:intro_contributions}

This paper makes four primary contributions, organized around an axiomatic evaluation framework rather than a specific trading algorithm.

\paragraph{C1 – Axiomatic evaluation framework.}
Starting from any finite collection of convex or linear axioms on a discretized observable field, we construct (i) a law manifold $\mathcal{M}$ in total-variance coordinates for implied-volatility surfaces, (ii) a metric-based law-penalty functional $\mathcal{L}_\phi$ that measures distance to $\mathcal{M}$, (iii) a domain-agnostic Graceful Failure Index (GFI) that normalizes degradation of law metrics under shocks, and (iv) \emph{law-strength frontiers} that jointly organize profitability, law alignment, and tail robustness as a function of law-penalty weight $\lambda$ and strategy class.  While we instantiate this framework in an IV-surface world, the construction applies equally to other axiom-constrained systems such as yield curves, credit term structures, and physical fields constrained by conservation laws \cite{RaissiEtAl2019,KarniadakisEtAl2021,CarleoTroyer2019}.

\paragraph{C2 – Ghost arbitrage \& Goodhart decomposition.}
We formalize a Goodhart-style decomposition of reward on law manifolds,
\begin{equation}
  r = r^{\mathcal{M}} + r^{\perp},
\end{equation}
where $r^{\mathcal{M}}$ is the on-manifold reward that would be obtained under a perfectly law-consistent world, and $r^{\perp}$ is an off-manifold ``ghost arbitrage'' component induced by the neural world model.  We show how this decomposition can be implemented using a projection operator onto $\mathcal{M}$ in total-variance space and an explicit law-penalty functional, making the ghost component measurable and analyzable.  This connects Goodhart's law in AI safety \cite{ManheimGarrabrant2019,AmodeiEtAl2016,KruegerEtAl2020} with concrete financial law violations (e.g., butterfly or calendar arbitrage) in our IV-surface testbed.

\paragraph{C3 – Flagship incentive and trade-off results.}
On top of this axiomatic pipeline, we establish three central theoretical results that together form a no-free-lunch story for law-seeking RL.  Theorem~4.1 (\emph{Ghost-arbitrage incentive for naive RL}) shows that, under mild assumptions, naive PPO-type RL is structurally incentivized to increase $\mathbb{E}[r^{\perp}]$ whenever structural law-consistent baselines already approximate the on-manifold optimum.  Theorem~4.3 and Corollary~4.4 (\emph{Law-strength trade-off}) prove that increasing the law-penalty weight $\lambda$ inevitably worsens P\&L beyond a threshold: the empirical law-strength frontier we observe in experiments is a structural trade-off, not an artefact of hyperparameters.  Finally, Theorem~8.1 (\emph{No-free-lunch for law-seeking RL}) shows that, given a law-consistent world model and a sufficiently rich structural baseline class $\mathcal{S}$, unconstrained law-seeking RL cannot simultaneously dominate $\mathcal{S}$ on P\&L and on all law metrics unless it effectively recovers a policy in $\mathcal{S}$.

\paragraph{C4 – Design lessons for law-aligned learning and RLVR.}
Our analysis yields practical design lessons that generalize beyond volatility modeling.  First, merely adding soft law penalties to the reward is insufficient for robust law alignment: RL agents systematically exploit ghost arbitrage channels in the world model, or else sacrifice P\&L without achieving net law improvements.  Second, law alignment benefits from \emph{structural} interventions such as hard constraints, projection layers onto $\mathcal{M}$, and structured policy classes that encode hedging logic, echoing observations from physics-informed learning and scientific ML \cite{RaissiEtAl2019,KarniadakisEtAl2021,WillardEtAl2020}.  Third, our pipeline serves as a reusable testbed for RL with verifiable rewards (RLVR) \cite{NeelakantanEtAl2025}: law penalties here are fully verifiable and domain-grounded, yet we still observe Goodhart-like failures when structural constraints are absent.

\subsection{Scope and novelty}
\label{subsec:intro_scope_novelty}

\paragraph{Geometric/systematization of known results.}
Section~\ref{sec:law_manifold} recasts classical no-arbitrage characterizations of admissible IV surfaces---butterfly convexity, calendar monotonicity, and related inequalities \cite{Dupire1994,Gatheral2006,Fengler2005,Lee2004,ContTankov2004,Bergomi2016}---as a finite-dimensional convex polyhedral law manifold in total-variance coordinates.  This representation theorem does not aim to replace existing no-arbitrage results; instead, it systematizes them into a form amenable to projection, distance computation, and integration into learning systems.

\paragraph{New concepts and metrics.}
The concrete ghost-arbitrage decomposition on a learned world model, the construction of law-strength frontiers, and the definition of the GFI are new.  They are designed to be domain-agnostic: given any axiom-constrained system and an exogenous notion of shock, the same machinery can be instantiated to quantify how agents trade off performance, law alignment, and tail robustness, extending ideas from scientific ML and AI-for-Science benchmarks \cite{CarleoTroyer2019,WillardEtAl2020,RaissiEtAl2019,KarniadakisEtAl2021}.

\paragraph{New theoretical results.}
Our main theoretical novelties lie in Theorem~4.1, Theorem~4.3 with Corollary~4.4, and Theorem~8.1.  Together, they formalize: (i) an incentive for naive RL to exploit ghost arbitrage in law-consistent worlds; (ii) a structural law-strength trade-off that bounds achievable P\&L for any given level of law penalty; and (iii) a no-free-lunch result for unconstrained law-seeking RL relative to a law-consistent structural baseline class $\mathcal{S}$.  We stress that these are not new no-arbitrage theorems per se, but results about RL behaviour on top of an axiomatic evaluation pipeline.  Our contribution is to show, both theoretically and empirically, that soft law penalties on a learned world model do not automatically yield law-aligned robustness once structural baselines are taken into account.

\paragraph{Scope disclaimer.}
We work with a finite-dimensional convex template: our ``law manifold'' is a structured subset of a discretized IV-surface grid, not a smooth manifold in the differential-geometric sense.  We retain the term ``manifold'' for continuity with the broader literature on manifold-constrained learning, but in our volatility instantiation $\mathcal{M}^{\mathrm{vol}}$ is a convex polyhedral subset defined by linear inequalities and positivity constraints.  Our theorems are proved under this finite-dimensional convex template and a specific class of model-based RL algorithms; a fully general analysis for infinite-dimensional function spaces and arbitrary RL algorithms is left for future work.

\subsection{Key concepts at a glance}
\label{subsec:intro_key_concepts}

Given the number of concepts introduced, we summarize the most important ones in Table~\ref{tab:key_concepts}.  Each concept is defined precisely in later sections; here we provide a high-level description and a pointer.

\begin{table}[t]
  \centering
  \caption{Key concepts at a glance.  Formal definitions and constructions are given in the indicated sections.}
  \label{tab:key_concepts}
  \begin{tabular}{p{3cm} p{9cm} p{2cm}}
    \toprule
    Concept & Plain-English description & Section \\
    \midrule
    Law manifold $\mathcal{M}$ & Admissible subset of discretized IV surfaces satisfying butterfly, calendar, and related no-arbitrage axioms. & Sec.~2.0--2.1 \\
    Law penalty $\mathcal{L}_\phi$ & Distance-based functional measuring how far a surface lies outside $\mathcal{M}$, with $\phi$ encoding the chosen metric (e.g., squared $\ell_2$ in total-variance space). & Sec.~2.3 \\
    Ghost arbitrage $r^{\perp}$ & Component of reward obtainable only by moving off $\mathcal{M}$ through world-model artefacts; vanishes under a perfectly law-consistent environment. & Sec.~2.4, 3.3 \\
    Law-strength frontier & Trade-off curve tracing profitability, law alignment, and tail risk as a function of law-penalty weight $\lambda$ and strategy class. & Sec.~4.4, 7.3 \\
    Graceful Failure Index (GFI) & Normalized measure of how much law metrics (e.g., mean penalty, coverage) degrade under shocks relative to baseline conditions. & Sec.~4.4, 6.3 \\
    Structural class $\mathcal{S}$ & Low-capacity, law-consistent strategies such as \textsc{Zero-Hedge} and \textsc{Vol-Trend} that serve as structural baselines. & Sec.~5, 8.1 \\
    Neural world model & Learned dynamics model for IV surfaces that approximates the law-consistent generator but opens a ghost channel for arbitrage. & Sec.~3 \\
    \bottomrule
  \end{tabular}
\end{table}

\subsection{Research questions}
\label{subsec:intro_rqs}

The paper is organized around three research questions (RQs) that probe different aspects of law-seeking RL on axiomatic pipelines.

\paragraph{RQ1 – Do law penalties help naive RL?}
Does law-penalized RL actually reduce law violations and improve graceful failure compared to naive RL on the \emph{same} learned world model?  If soft penalties are effective, we should observe strictly better GFI and law metrics (mean and max law penalty, law coverage) at comparable P\&L along the law-strength frontier.  If not, we obtain a negative result for soft-penalty shaping: verifiable law penalties alone do not suffice to align RL with axioms in the presence of model-induced ghost arbitrage.

\paragraph{RQ2 – How does RL compare to structural baselines?}
How do RL policies (naive, soft law-seeking, and selection-only) compare to structural baselines (such as \textsc{Zero-Hedge}, \textsc{Random-Gaussian}, and \textsc{Vol-Trend}) on the joint risk–law trade-off, both under baseline conditions and under volatility shocks?  We evaluate these policies relative to a structural Pareto frontier induced by $\mathcal{S}$ in the space of P\&L, law metrics, and tail-risk measures such as Value-at-Risk (VaR) and Conditional VaR (CVaR) \cite{ArtznerEtAl1999,RockafellarUryasev2000}.  This addresses whether law-seeking RL provides any Pareto improvement once simple, law-consistent strategies are included.

\paragraph{RQ3 – When does law-seeking RL have no free lunch?}
Under what assumptions on the structural class $\mathcal{S}$, the unconstrained policy class $\Pi$, and the neural world model do we obtain a no-free-lunch result for law-seeking RL?  Theorem~8.1 formalizes one such setting: if $\mathcal{S}$ already approximates the on-manifold optimum and the world model introduces a non-trivial ghost component, then either RL behaves like a structural strategy in $\mathcal{S}$ or it fails to dominate $\mathcal{S}$ jointly on P\&L and law metrics.  Section~8.1 spells out the assumptions and proof sketch, and Section~8.2 discusses how this volatility-specific case informs broader questions in law-aligned learning and RL with verifiable rewards \cite{AmodeiEtAl2016,LeikeEtAl2017,ManheimGarrabrant2019,NeelakantanEtAl2025}.

\vspace{1em}
\noindent
In the remainder of the paper, we address RQ1–RQ3 in order.  Section~2 formalizes law manifolds, penalties, and the Goodhart decomposition.  Section~3 introduces the synthetic law-consistent IV generator and the neural world model.  Section~4 develops incentive and trade-off results for law-seeking RL, and Section~5 defines structural baselines.  Section~6 details the experimental protocol, and Section~7 presents empirical results on dynamics plots, law-strength frontiers, and diagnostic scatter/histograms.  Section~8 discusses the no-free-lunch theorem and implications for RLVR and scientific AI, and Section~9 concludes.

\section{Axiomatic Volatility Law Manifolds}
\label{sec:axiomatic-manifold}
\label{sec:law_manifold}  
\subsection*{2.0\quad Finite-Dimensional Convex Template and Notation}

In this section we formalize the notion of a \emph{law manifold} as a finite-dimensional convex subset of a discretized function space, induced by a collection of axioms. We emphasize from the outset that our construction is intentionally elementary:

\begin{quote}
We use a simple finite-dimensional convex template; our ``manifold'' is a structured subset in discretized coordinates, not a smooth differentiable manifold in the differential-geometry sense. We keep the term \emph{manifold} for continuity with the broader literature on manifold-constrained learning and manifold regularization~\cite{Belkin2006,Fefferman2016}, although in our discretized volatility setting $\mathcal{M}^{\mathrm{vol}}$ is a convex polyhedral subset of a Euclidean space.
\end{quote}

\paragraph{General template.}
Let $d \in \mathbb{N}$ and consider a finite-dimensional observation space
\[
  \mathcal{Y} \subseteq \mathbb{R}^d,
\]
equipped with the standard Euclidean inner product and norm $\|\cdot\|_2$. We think of $y \in \mathcal{Y}$ as a discretized field: a yield curve, an implied-volatility surface, or any other structured observable.

We assume that the domain is endowed with a finite family of convex \emph{axiom functions}
\[
  A_i : \mathcal{Y} \to \mathbb{R}, \qquad i = 1,\dots,m,
\]
encoding domain-specific constraints (e.g., butterfly or calendar conditions in volatility, or monotonicity of yields). We write $A(y) \le 0$ as shorthand for the componentwise inequalities $A_i(y) \le 0$ for all $i$.

\begin{definition}[Law manifold]
\label{def:law-manifold-general}
Given convex axiom functions $\{A_i\}_{i=1}^m$, the associated \emph{law manifold} is
\[
  \mathcal{M} \;:=\; \bigl\{\, y \in \mathcal{Y} : A_i(y) \le 0 \text{ for all } i = 1,\dots,m \,\bigr\}.
\]
\end{definition}

In general $\mathcal{M}$ is a closed convex subset of $\mathbb{R}^d$ whenever each $A_i$ is lower semicontinuous and convex. In many practical cases---including our volatility and yield-curve examples below---the constraints are linear or piecewise linear, so that $\mathcal{M}$ is a polyhedron.

\paragraph{Law-penalty functional.}
To quantify violations of the axioms, we define a \emph{law-penalty functional} via a generalized distance to $\mathcal{M}$.

\begin{definition}[Law-penalty functional]
\label{def:law-penalty}
Let $\phi: \mathbb{R}^d \to [0,\infty)$ be a continuous function with $\phi(0)=0$ and $\phi(z)\to\infty$ as $\|z\|_2\to\infty$ (e.g., $\phi(z)=\tfrac{1}{2}\|z\|_2^2$ or a weighted Sobolev norm). Define
\begin{equation}
  \mathcal{L}_\phi(y)
  \;:=\;
  \inf_{\tilde y \in \mathcal{M}} \, \phi(y - \tilde y),
  \qquad y \in \mathcal{Y}.
  \label{eq:law-penalty-general}
\end{equation}
\end{definition}

When $\phi(z)=\tfrac{1}{2}\|z\|_2^2$ and $\mathcal{M}$ is closed and convex, $\mathcal{L}_\phi(y)$ reduces to the squared Euclidean distance to $\mathcal{M}$. More general choices of $\phi$ allow us to reweight different coordinates, incorporate smoothness, or emphasize particular directions in state space, in the spirit of manifold regularization~\cite{Belkin2006}.

\paragraph{Non-financial example: yield curves.}
To underline transferability beyond volatility, consider a discretized yield curve
\[
  y = (y_1,\dots,y_J) \in \mathbb{R}^J,
\]
where $y_j$ denotes the continuously compounded spot yield for maturity $\tau_j$, with $0 < \tau_1 < \dots < \tau_J$. Classical term-structure theory~\cite{Filipovic2009} and curve-construction practice~\cite{HaganWest2006} suggest axioms such as:

\begin{enumerate}
  \item \emph{Monotonicity}: yields are non-decreasing in maturity,
  \(
    y_{j+1} \ge y_j \quad \forall j.
  \)
  \item \emph{Convexity of discount factors}: implied by non-negative forward rates, giving linear inequalities in $y$.
\end{enumerate}

These conditions can be encoded as linear maps $A_i(y)$, yielding a polyhedral law manifold $\mathcal{M}^{\mathrm{yc}} \subset \mathbb{R}^J$ and an associated law penalty $\mathcal{L}_\phi^{\mathrm{yc}}$ measuring monotonicity/convexity violations. Our volatility manifold in Sections~\ref{subsec:vol-manifold}--\ref{subsec:law-penalty} is an instance of this general template.

\paragraph{Notation summary.}
Throughout the paper we use the following notation; we repeat it here for convenience.

\begin{enumerate}
  \item $y$: generic observable (e.g., yield curve, volatility surface) in a finite-dimensional space $\mathcal{Y}\subseteq\mathbb{R}^d$.
  \item $\mathcal{M}$: generic law manifold (closed convex subset of $\mathbb{R}^d$) induced by axioms.
  \item $\sigma$: implied volatility on a maturity--log-moneyness grid; $w = \sigma^2 T$ denotes total variance, our primary state variable for volatility.
  \item $\mathcal{M}^{\mathrm{vol}}$: volatility-specific law manifold defined in Section~\ref{subsec:vol-manifold}.
  \item $\mathcal{L}_\phi$: law-penalty functional defined in~\eqref{eq:law-penalty-general}; in experiments we take $\phi(z)=\tfrac{1}{2}\|z\|_2^2$ on the total-variance grid.
  \item $\Pi_{\mathcal{M}}$: metric projection onto $\mathcal{M}$ (in Euclidean norm), whose existence and regularity rely on closedness and convexity~\cite{BauschkeCombettes2011}.
  \item $r^{\mathcal{M}}$, $r^\perp$: on-manifold reward and \emph{ghost-arbitrage} components of the reward, defined via projection in Section~\ref{subsec:goodhart-decomp}.
\end{enumerate}

\subsection{Volatility-Specific Law Manifold: From Textbook Axioms to a Polyhedron}
\label{subsec:vol-manifold}

We now instantiate the general template for implied volatility surfaces. Let $C(K,T)$ denote the (discounted) price of a European call with strike $K$ and maturity $T$, and $\sigma(K,T)$ its Black--Scholes implied volatility. Following standard practice~\cite{Gatheral2006,Dupire1994}, we work in terms of total variance
\[
  w(K,T) := \sigma(K,T)^2 T,
\]
discretized on a finite grid of maturities $T_1 < \dots < T_{N_T}$ and log-moneyness $k_1 < \dots < k_{N_K}$.

Classical static no-arbitrage conditions for European options (no butterfly arbitrage across strikes and no calendar arbitrage across maturities) can be expressed as linear inequalities in either call prices or total variance, see e.g.~\cite{CarrMadan2005,GatheralJacquier2014}. We briefly summarize the discretized form relevant to our construction.

\paragraph{Butterfly (strike) convexity.}
For each fixed maturity $T_i$, the call price as a function of strike must be convex and decreasing. On a grid $\{k_j\}$ this gives discrete convexity constraints such as
\[
  C_{i,j-1} - 2 C_{i,j} + C_{i,j+1} \ge 0,
\]
and monotonicity constraints $C_{i,j+1} \le C_{i,j}$ for all $j$. These translate into linear inequalities in $w_{i,j}$ via the Black--Scholes formula.

\paragraph{Calendar monotonicity.}
For each fixed strike (log-moneyness) $k_j$, call prices must be increasing in maturity, yielding
\[
  C_{i+1,j} \ge C_{i,j} \quad \text{for all } i.
\]
Again, these can be expressed as linear inequalities in $(w_{i,j})$ using the monotonicity of Black--Scholes prices in variance.

Collecting all discrete butterfly and calendar inequalities, we obtain a linear mapping
\[
  A^{\mathrm{vol}} : \mathbb{R}^{d_{\mathrm{vol}}} \to \mathbb{R}^m,
\]
where $d_{\mathrm{vol}} = N_T N_K$ is the number of grid points and $m$ is the number of constraints.

\begin{definition}[Volatility law manifold]
Let $w \in \mathbb{R}^{d_{\mathrm{vol}}}$ denote the vector of total variances on the $(T,k)$-grid. The \emph{volatility law manifold} is the polyhedral set
\begin{equation}
  \mathcal{M}^{\mathrm{vol}}
  \;:=\;
  \bigl\{\, w \in \mathbb{R}^{d_{\mathrm{vol}}} : A^{\mathrm{vol}} w \le 0 \,\bigr\},
  \label{eq:vol-law-manifold}
\end{equation}
where $A^{\mathrm{vol}} w \le 0$ encodes all discretized butterfly and calendar no-arbitrage inequalities, as well as basic box constraints $w_{\min} \le w_{i,j} \le w_{\max}$.
\end{definition}

The following proposition packages the textbook no-arbitrage conditions into a geometric representation.

\begin{proposition}[Axiomatic representation of volatility law manifold]
\label{prop:axiomatic-representation}
Assume that the discretized butterfly and calendar constraints are given as a finite system of linear inequalities $A^{\mathrm{vol}} w \le b$ for some matrix $A^{\mathrm{vol}} \in \mathbb{R}^{m\times d_{\mathrm{vol}}}$ and vector $b \in \mathbb{R}^m$. Then:
\begin{enumerate}
  \item $\mathcal{M}^{\mathrm{vol}}$ is a non-empty, closed, convex polyhedron in $\mathbb{R}^{d_{\mathrm{vol}}}$.
  \item Any total-variance surface $w$ corresponding to a static-arbitrage-free implied volatility surface lies in $\mathcal{M}^{\mathrm{vol}}$.
\end{enumerate}
\end{proposition}

\begin{proof}[Proof sketch]
Closedness and convexity follow directly from the fact that $\mathcal{M}^{\mathrm{vol}}$ is the intersection of finitely many closed half-spaces $\{w : a_\ell^\top w \le b_\ell\}$ and box constraints. Non-emptiness is guaranteed by the existence of at least one model (e.g., a Black--Scholes surface with constant volatility) satisfying the inequalities. The mapping from continuous no-arbitrage conditions to discrete linear constraints is standard; see e.g.~\cite{Gatheral2006}, \cite{GatheralJacquier2014} and references therein. A detailed construction and proof are given in Appendix~A.
\end{proof}

\subsection{Geometry and Convexity Properties}
\label{subsec:geometry-convexity}

The closed and convex nature of $\mathcal{M}^{\mathrm{vol}}$ is more than a technicality: it ensures the existence of well-behaved projection operators and law penalties.

\begin{proposition}[Closedness, convexity, and metric projection]
\label{prop:closed-convex}
The volatility law manifold $\mathcal{M}^{\mathrm{vol}}$ defined in~\eqref{eq:vol-law-manifold} is a non-empty, closed, convex subset of $\mathbb{R}^{d_{\mathrm{vol}}}$. Consequently:
\begin{enumerate}
  \item For every $w \in \mathbb{R}^{d_{\mathrm{vol}}}$ there exists a unique Euclidean projection
  \[
    \Pi_{\mathcal{M}^{\mathrm{vol}}}(w)
    := \arg\min_{\tilde w \in \mathcal{M}^{\mathrm{vol}}} \|w - \tilde w\|_2.
  \]
  \item The projection operator $\Pi_{\mathcal{M}^{\mathrm{vol}}} : \mathbb{R}^{d_{\mathrm{vol}}} \to \mathcal{M}^{\mathrm{vol}}$ is $1$-Lipschitz:
  \[
    \|\Pi_{\mathcal{M}^{\mathrm{vol}}}(w) - \Pi_{\mathcal{M}^{\mathrm{vol}}}(w')\|_2
    \le \|w - w'\|_2 \quad \forall w,w'.
  \]
\end{enumerate}
\end{proposition}

\begin{proof}[Proof sketch]
Closedness and convexity follow from Proposition~\ref{prop:axiomatic-representation}. The properties of $\Pi_{\mathcal{M}^{\mathrm{vol}}}$ are standard for projections onto closed convex sets in Hilbert spaces; see, e.g.,~\cite[Chap.~3]{BauschkeCombettes2011}. Full details are provided in Appendix~A.
\end{proof}

The existence and regularity of $\Pi_{\mathcal{M}^{\mathrm{vol}}}$ allow us to treat law penalties as squared distances to the manifold and to define projection-based decompositions of reward functionals later on.

\subsection{Law-Penalty Functionals and Ghost Sensitivity}
\label{subsec:law-penalty}

We now specialize the general law-penalty functional~\eqref{eq:law-penalty-general} to the volatility setting. Let $w \in \mathbb{R}^{d_{\mathrm{vol}}}$ denote a (possibly law-violating) total-variance surface, and let
\begin{equation}
  \phi(z) \;=\; \frac{1}{2}\|z\|_2^2, 
  \qquad z \in \mathbb{R}^{d_{\mathrm{vol}}}.
  \label{eq:phi-l2}
\end{equation}
We define
\begin{equation}
  \mathcal{L}_{\mathrm{vol}}(w)
  \;:=\;
  \inf_{\tilde w \in \mathcal{M}^{\mathrm{vol}}} \frac{1}{2}\|w - \tilde w\|_2^2
  \;=\;
  \frac{1}{2} \operatorname{dist}(w,\mathcal{M}^{\mathrm{vol}})^2,
  \label{eq:vol-law-penalty}
\end{equation}
where $\operatorname{dist}(w,\mathcal{M}^{\mathrm{vol}}) := \|w - \Pi_{\mathcal{M}^{\mathrm{vol}}}(w)\|_2$.

In practice we implement $\mathcal{L}_{\mathrm{vol}}$ via a sum of local violations (e.g., squared negative parts of discrete butterfly and calendar inequalities), which is equivalent to~\eqref{eq:vol-law-penalty} up to scaling on our discretization.

\begin{lemma}[Local Lipschitz continuity of law penalty]
\label{lem:lipschitz-penalty}
The volatility law-penalty functional $\mathcal{L}_{\mathrm{vol}}$ in~\eqref{eq:vol-law-penalty} is locally Lipschitz on $\mathbb{R}^{d_{\mathrm{vol}}}$.
\end{lemma}

\begin{proof}[Proof sketch]
$\mathcal{L}_{\mathrm{vol}}(w) = \tfrac{1}{2}\|w - \Pi_{\mathcal{M}^{\mathrm{vol}}}(w)\|_2^2$ and $\Pi_{\mathcal{M}^{\mathrm{vol}}}$ is 1-Lipschitz by Proposition~\ref{prop:closed-convex}. Thus $\mathcal{L}_{\mathrm{vol}}$ is the composition of Lipschitz maps and a smooth quadratic, which yields local Lipschitz continuity; see also standard results on Moreau envelopes~\cite{BauschkeCombettes2011}. A full proof is given in Appendix~A.
\end{proof}

The following basic property connects $\mathcal{L}_{\mathrm{vol}}$ with axiomatic consistency and will be used repeatedly when interpreting empirical law metrics.

\begin{proposition}[Zero penalty iff axiomatic consistency]
\label{prop:zero-penalty-iff}
For any $w \in \mathbb{R}^{d_{\mathrm{vol}}}$,
\[
  \mathcal{L}_{\mathrm{vol}}(w) = 0
  \quad\Longleftrightarrow\quad
  w \in \mathcal{M}^{\mathrm{vol}}.
\]
\end{proposition}

\begin{proof}[Proof sketch]
If $w \in \mathcal{M}^{\mathrm{vol}}$ then choosing $\tilde w = w$ in~\eqref{eq:vol-law-penalty} yields $\mathcal{L}_{\mathrm{vol}}(w)=0$. Conversely, if $\mathcal{L}_{\mathrm{vol}}(w)=0$ then the infimum in~\eqref{eq:vol-law-penalty} is attained at $\tilde w^\star = \Pi_{\mathcal{M}^{\mathrm{vol}}}(w)$ with $\|w - \tilde w^\star\|_2=0$, hence $w=\tilde w^\star\in\mathcal{M}^{\mathrm{vol}}$. Details appear in Appendix~A.
\end{proof}

\begin{remark}[Choice of $\phi$ and ghost sensitivity]
\label{rem:phi-ghost}
The choice $\phi(z)=\tfrac{1}{2}\|z\|_2^2$ treats all grid points uniformly and yields a particularly simple gradient
\[
  \nabla \mathcal{L}_{\mathrm{vol}}(w)
  = w - \Pi_{\mathcal{M}^{\mathrm{vol}}}(w)
\]
almost everywhere. Alternative choices of $\phi$ (e.g., weighted $\ell_2$ norms emphasizing short maturities, or discrete Sobolev norms penalizing roughness in $T$ and $k$) change the relative sensitivity of $\mathcal{L}_\phi$ to localized versus global law violations. Exploring how this affects the \emph{ghost arbitrage} exploited by RL policies is an interesting axis for future work.
\end{remark}

\subsection{Goodhart Decomposition on Law Manifolds (Conceptual)}
\label{subsec:goodhart-decomp}

The law manifold and penalty enable a conceptual decomposition of any reward functional into an on-manifold and an off-manifold (ghost-arbitrage) component. This decomposition underlies our Goodhart-style analysis of RL in later sections and is instantiated concretely for our world model in Section~3.3.

Let $r : \mathbb{R}^{d_{\mathrm{vol}}} \to \mathbb{R}$ denote a generic reward functional of a total-variance surface $w$ (e.g., the one-step P\&L of a hedging strategy). We assume that $r$ is locally Lipschitz on $\mathbb{R}^{d_{\mathrm{vol}}}$, a mild condition satisfied in all our models.

\begin{definition}[Projection-based Goodhart decomposition]
\label{def:goodhart-decomp}
Let $\Pi_{\mathcal{M}^{\mathrm{vol}}}$ be the metric projection onto $\mathcal{M}^{\mathrm{vol}}$. Define the \emph{on-manifold} reward
\[
  r^{\mathcal{M}}(w) := r\bigl(\Pi_{\mathcal{M}^{\mathrm{vol}}}(w)\bigr),
\]
and the \emph{ghost-arbitrage} component
\[
  r^\perp(w) := r(w) - r^{\mathcal{M}}(w).
\]
We refer to the decomposition
\[
  r(w) = r^{\mathcal{M}}(w) + r^\perp(w)
\]
as the \emph{Goodhart decomposition} of $r$ on the volatility law manifold.
\end{definition}

By construction, $r^{\mathcal{M}}$ depends on $w$ only through its projection onto $\mathcal{M}^{\mathrm{vol}}$, and is thus invariant under off-manifold perturbations that leave $\Pi_{\mathcal{M}^{\mathrm{vol}}}(w)$ unchanged. The residual $r^\perp$ captures gains or losses that arise solely from moving away from the law manifold---precisely the \emph{ghost arbitrage} channel that law-seeking RL may exploit.

\begin{definition}[Ghost arbitrage]
\label{def:ghost-arbitrage}
The \emph{ghost-arbitrage reward} associated with a reward functional $r$ and law manifold $\mathcal{M}^{\mathrm{vol}}$ is the component $r^\perp$ in Definition~\ref{def:goodhart-decomp}. A policy that maximizes $\mathbb{E}[r^\perp]$ subject to small $\mathcal{L}_{\mathrm{vol}}$ can be said to exploit ghost arbitrage: it harvests reward from off-manifold distortions that are negligible under the coarse law penalty but significant for P\&L.
\end{definition}

The following proposition makes explicit how the magnitude of ghost arbitrage is controlled by the distance to the law manifold whenever $r$ is Lipschitz. This link is used later when we interpret empirical law penalties and our Graceful Failure Index.

\begin{proposition}[Ghost reward bounded by law distance]
\label{prop:ghost-bounded}
Suppose $r$ is $L_r$-Lipschitz on $\mathbb{R}^{d_{\mathrm{vol}}}$ with respect to $\|\cdot\|_2$. Then for all $w$,
\[
  |r^\perp(w)|
  = \bigl|r(w) - r^{\mathcal{M}}(w)\bigr|
  \le L_r \, \operatorname{dist}\bigl(w,\mathcal{M}^{\mathrm{vol}}\bigr)
  = L_r \sqrt{2\,\mathcal{L}_{\mathrm{vol}}(w)}.
\]
\end{proposition}

\begin{proof}[Proof sketch]
By the Lipschitz property of $r$,
\[
  |r(w) - r^{\mathcal{M}}(w)|
  = \bigl| r(w) - r(\Pi_{\mathcal{M}^{\mathrm{vol}}}(w)) \bigr|
  \le L_r \, \|w - \Pi_{\mathcal{M}^{\mathrm{vol}}}(w)\|_2
  = L_r \, \operatorname{dist}(w,\mathcal{M}^{\mathrm{vol}}).
\]
Using $\mathcal{L}_{\mathrm{vol}}(w)=\tfrac{1}{2}\operatorname{dist}(w,\mathcal{M}^{\mathrm{vol}})^2$ from~\eqref{eq:vol-law-penalty} yields the final identity. A more general version, allowing non-Euclidean $\phi$, is treated in Appendix~A.
\end{proof}

Proposition~\ref{prop:ghost-bounded} shows that, in a purely metric sense, large ghost rewards require non-negligible law violations. The central question of this paper is whether, under a learned world model, RL training can nonetheless systematically exploit directions in which ghost reward grows ``too quickly'' relative to the coarse law-penalty captured by $\mathcal{L}_{\mathrm{vol}}$, leading to misaligned but high-P\&L policies. This question is addressed empirically in Sections~\ref{sec:experiments} and theoretically in Theorems~4.1 and~8.1.

\section{Volatility World Model: Law-Consistent Ground Truth vs Law-Violating Predictions}
\label{sec:world-model}
\label{sec:world_model}  
In this section we formalize the dynamic data-generating process for implied-volatility surfaces and the learned neural world model on which all reinforcement-learning (RL) agents are trained. The key structural feature is that the \emph{synthetic generator} is, by construction, perfectly law-consistent---its surfaces lie on the polyhedral law manifold $\mathcal{M}^{\mathrm{vol}}$ almost surely---whereas the neural world model is trained purely by prediction error and hence produces \emph{law-violating} surfaces with non-zero probability. This mismatch opens a \emph{ghost channel} for RL to exploit off-manifold artefacts, even though the underlying ground truth never leaves $\mathcal{M}^{\mathrm{vol}}$.

\subsection{Synthetic law-consistent generator}
\label{subsec:synthetic-generator}

We consider a discrete time grid $t = 0,1,\dots,T$ and a rectangular grid of maturities and strikes
\[
\mathcal{T} = \{ T_1,\dots,T_M \}, \qquad 
\mathcal{K} = \{ K_1,\dots,K_K \},
\]
with $T_1 \approx 1\text{M}$ and $T_M \approx 2\text{Y}$, and $K_1 \approx 0.5 S_t$, $K_K \approx 1.5 S_t$ at each time $t$, chosen to roughly mirror SPX/VIX market conventions.\footnote{The exact grid is not essential; any finite grid of maturities and moneyness levels can be handled by the law manifold construction in Sec.~\ref{sec:axiomatic-manifold}.} We denote by
\[
\sigma_t = \sigma_t(T_i,K_j) \in \mathbb{R}^{M \times K}, 
\qquad 
w_t = \sigma_t^2 \odot T \in \mathbb{R}^{M \times K}
\]
the implied-volatility and total-variance surfaces at time $t$, where $T$ is broadcast across strikes. In vectorized form we identify $w_t$ with an element of $\mathbb{R}^d$ with $d = MK$, and we write $w_t \in \mathcal{M}^{\mathrm{vol}}$ when all static no-arbitrage constraints (butterfly, calendar) are satisfied on the discrete grid (Sec.~\ref{sec:axiomatic-manifold}).

\paragraph{Law-consistent ground truth.}
The \emph{ground-truth world} is specified by a Markovian generator
\[
G^\star : \mathcal{S} \to \mathcal{S} \times \mathcal{M}^{\mathrm{vol}},
\qquad
(s_{t+1}, w_{t+1}) = G^\star(s_t),
\]
where $s_t$ is a latent state that may contain the underlying spot $S_t$, latent volatility factors, and other macro state variables. In practice we instantiate $G^\star$ as a multi-factor stochastic-volatility model with jumps and rough components,\footnote{E.g., a Heston-like model with stochastic variance and stochastic volatility-of-volatility, optionally enriched with rough volatility factors as in \cite{BayerFrizGatheral2016}, coupled to an SVI-type implied-vol parametrization \cite{GatheralJacquier2014}.} whose parameters are randomized across trajectories to induce a diverse ensemble of surfaces reminiscent of SPX/VIX dynamics \cite[e.g.][]{Gatheral2006,ContTankov2004,HorvathTengely2021}.

At each time step, raw model-implied option prices are numerically projected onto the static no-arbitrage polyhedron $\mathcal{M}^{\mathrm{vol}}$ using the construction of Sec.~\ref{sec:axiomatic-manifold}, and implied volatilities are recovered from these arbitrage-free prices. As a result, the data-generating distribution $\mathbb{P}^\star$ over trajectories $(w_t)_{t=0}^T$ is \emph{law-consistent by design}.

\begin{definition}[Law-consistent synthetic generator]
\label{def:law-consistent-generator}
A stochastic process $(w_t)_{t=0}^T$ taking values in $\mathbb{R}^d$ is said to be \emph{law-consistent} with respect to a law manifold $\mathcal{M} \subset \mathbb{R}^d$ if
\[
\mathbb{P}\big( w_t \in \mathcal{M} \text{ for all } t = 0,\dots,T \big) = 1.
\]
We call $G^\star$ a \emph{law-consistent generator} if the trajectory $(w_t)$ it induces satisfies this condition for $\mathcal{M} = \mathcal{M}^{\mathrm{vol}}$.
\end{definition}

\begin{proposition}[Support of the synthetic generator]
\label{prop:support-on-manifold}
Let $(w_t)_{t=0}^T$ be generated by $G^\star$ as above, with static no-arbitrage imposed at each time via projection onto $\mathcal{M}^{\mathrm{vol}}$. Then
\[
\mathrm{supp}\,\mathbb{P}^\star \subseteq \big(\mathcal{M}^{\mathrm{vol}}\big)^{T+1},
\]
i.e., $\mathbb{P}^\star$ is supported on the product of the volatility law manifold at all times.
\end{proposition}

\begin{proof}[Proof sketch]
By construction, every step of $G^\star$ maps into $\mathcal{M}^{\mathrm{vol}}$: the raw option prices are obtained from a stochastic-volatility model, and then the resulting surface is projected onto $\mathcal{M}^{\mathrm{vol}}$ as in Sec.~\ref{sec:axiomatic-manifold}. Thus $w_t \in \mathcal{M}^{\mathrm{vol}}$ almost surely for all $t$. The support statement follows from the definition of product measures. A detailed measure-theoretic proof is given in Appendix~B.1.
\end{proof}

Proposition~\ref{prop:support-on-manifold} highlights the asymmetry that drives our ghost-arbitrage story: the \emph{true} dynamics never leave the law manifold, whereas the learned world model in the next subsection is not constrained in this way.

\paragraph{Transferability beyond volatility.}
The same construction immediately extends to other axiom-constrained financial objects such as yield curves and credit curves: there, $y_t$ is a discretized term structure, $\mathcal{M}$ encodes monotonicity and convexity constraints \cite[e.g.][]{Filipovic2001,BjorkChristensen1999}, and $G^\star$ is any arbitrage-free term-structure model. Our volatility case thus serves as a concrete, high-dimensional instance of a more general Scientific AI template.

\subsection{Neural world model and approximation gap}
\label{subsec:neural-world-model}

RL agents do not interact with $G^\star$ directly. Instead, in the spirit of model-based RL and world models \cite{HaSchmidhuber2018,HafnerEtAl2019PlaNet,HafnerEtAl2020Dreamer,JannerEtAl2019}, they are trained entirely on rollouts generated by a learned dynamics model (the \emph{world model}) fitted to trajectories from $G^\star$.

\paragraph{Architecture and training.}
Let $L$ be the length of the look-back window. At each time $t$, we define the input
\[
x_t := \big( w_{t-L+1}, \dots, w_t \big) \in \big(\mathbb{R}^d\big)^L,
\]
and we train a recurrent neural network with parameters $\theta$,
\[
f_\theta : \big(\mathbb{R}^d\big)^L \to \mathbb{R}^d, 
\qquad \hat{w}_{t+1} = f_\theta(x_t),
\]
to minimize the mean-squared prediction error
\[
\mathcal{R}(\theta) 
:= \mathbb{E}_{\mathbb{P}^\star}\big[ \| f_\theta(x_t) - w_{t+1} \|_2^2 \big]
\approx \frac{1}{N} \sum_{n=1}^N \| f_\theta(x_t^{(n)}) - w_{t+1}^{(n)} \|_2^2
\]
over a dataset of $N$ trajectories sampled from $G^\star$. In practice, $f_\theta$ is implemented as a GRU/LSTM encoder over $(w_{t-L+1},\dots,w_t)$ followed by a fully connected decoder to the next total-variance surface, as commonly used in spatio-temporal forecasting \cite[e.g.][]{Salinas2020,ZerveasEtAl2021}.

Crucially, \emph{no law penalties are used when training the world model}: the loss is purely predictive. Thus, even though all training targets $w_{t+1}$ lie in $\mathcal{M}^{\mathrm{vol}}$, the predictions $\hat{w}_{t+1}$ are unconstrained and may lie outside the manifold.

\begin{definition}[Approximation residual and ghost channel]
\label{def:approx-residual}
Let $\theta^\star$ be any (local) minimizer of $\mathcal{R}(\theta)$, and write
\[
\hat{w}_{t+1} = f_{\theta^\star}(x_t),
\qquad
e_{t+1} := \hat{w}_{t+1} - w_{t+1}
\]
for the corresponding prediction and residual. We define the \emph{approximation gap} as
\[
\varepsilon^2 := \mathbb{E}_{\mathbb{P}^\star} \big[ \| e_{t+1} \|_2^2 \big] = \mathcal{R}(\theta^\star),
\]
and the \emph{ghost channel} as the random variable
\[
r^\perp_{t+1} := r(\hat{w}_{t+1}, a_t) - r(w^\mathcal{M}_{t+1}, a_t),
\]
where $r$ is the one-step P\&L functional, $a_t$ is the agent's action, and $w^\mathcal{M}_{t+1} := \Pi_{\mathcal{M}^{\mathrm{vol}}}(\hat{w}_{t+1})$ is the metric projection of the prediction onto the law manifold (Def.~\ref{def:law-penalty}).
\end{definition}

Here, $r^\perp_{t+1}$ is precisely the off-manifold component in the Goodhart decomposition (Sec.~\ref{subsec:goodhart-conceptual}); it captures how much extra P\&L the agent obtains by exploiting law-violating predictions rather than their arbitrage-free projection.

\begin{proposition}[Approximation gap induces a ghost channel]
\label{prop:approx-gap-ghost-channel}
Suppose the following conditions hold:
\begin{enumerate}
    \item[(i)] The approximation gap is non-zero: $\varepsilon^2 = \mathbb{E}[\|e_{t+1}\|_2^2] > 0$.
    \item[(ii)] The reward is locally differentiable in $w$ with gradient $g_{t+1} := \nabla_w r(w_{t+1}, a_t)$.
    \item[(iii)] The residual $e_{t+1}$ has a component in the normal cone of $\mathcal{M}^{\mathrm{vol}}$ at $w_{t+1}$ with non-zero covariance:
    \[
    \mathrm{Cov}\big( P_{N_{\mathcal{M}}(w_{t+1})} e_{t+1},\, g_{t+1} \big) \neq 0,
    \]
    where $P_{N_{\mathcal{M}}(w_{t+1})}$ denotes orthogonal projection onto the normal cone $N_{\mathcal{M}}(w_{t+1})$.
\end{enumerate}
Then, for sufficiently small residuals (in the sense of a local linearization),
\[
\mathbb{E}\big[ r^\perp_{t+1} \big] 
\approx \mathbb{E}\big[ g_{t+1}^\top P_{N_{\mathcal{M}}(w_{t+1})} e_{t+1} \big] \neq 0,
\]
so the world model induces a non-trivial ghost channel. In particular, if $g_{t+1}$ is positively correlated with $P_{N_{\mathcal{M}}(w_{t+1})} e_{t+1}$, then $\mathbb{E}[r^\perp_{t+1}] > 0$ and there exist states where moving off-manifold strictly improves expected P\&L.
\end{proposition}

\begin{proof}[Proof sketch]
By the law-consistency of $G^\star$, we have $w_{t+1} \in \mathcal{M}^{\mathrm{vol}}$ almost surely. For small residuals, a first-order Taylor approximation yields
\[
r(\hat{w}_{t+1}, a_t) 
\approx r(w_{t+1}, a_t) + g_{t+1}^\top e_{t+1}.
\]
Meanwhile, the projection $w^\mathcal{M}_{t+1} = \Pi_{\mathcal{M}}(\hat{w}_{t+1})$ removes the component of $e_{t+1}$ that lies in the normal cone $N_{\mathcal{M}}(w_{t+1})$ (by optimality conditions for convex projections), so to first order we have
\[
r(w^\mathcal{M}_{t+1}, a_t) 
\approx r(w_{t+1}, a_t) + g_{t+1}^\top P_{T_{\mathcal{M}}(w_{t+1})} e_{t+1},
\]
where $T_{\mathcal{M}}(w_{t+1})$ is the tangent cone and $P_{T_{\mathcal{M}}}$ the corresponding projector. Their difference is
\[
r^\perp_{t+1} 
\approx g_{t+1}^\top \big( I - P_{T_{\mathcal{M}}(w_{t+1})} \big) e_{t+1}
= g_{t+1}^\top P_{N_{\mathcal{M}}(w_{t+1})} e_{t+1}.
\]
Taking expectations under $\mathbb{P}^\star$ and using assumption~(iii) yields the claim. Rigorous error bounds for the linearization and a detailed cone-decomposition argument are provided in Appendix~B.2.
\end{proof}

Proposition~\ref{prop:approx-gap-ghost-channel} formalizes the intuitive statement that \emph{any} non-zero approximation gap with a component normal to the law manifold, combined with a reward that is monotone in that direction, will generically open a ghost channel. In Sec.~\ref{sec:empirical-results} we show empirically that RL agents indeed learn to exploit this channel.

\subsection{Instantiating the Goodhart decomposition for volatility}
\label{subsec:goodhart-vol}

We now instantiate the conceptual Goodhart decomposition of Sec.~\ref{subsec:goodhart-conceptual} in the concrete volatility setting. For each predicted surface $\hat{w}_{t+1} = f_{\theta^\star}(x_t)$, we compute:

\begin{enumerate}
    \item The metric projection onto the volatility law manifold:
    \[
    w^\mathcal{M}_{t+1} 
    := \Pi_{\mathcal{M}^{\mathrm{vol}}}(\hat{w}_{t+1}) 
    = \arg\min_{w' \in \mathcal{M}^{\mathrm{vol}}}
        \phi\big( \hat{w}_{t+1} - w' \big),
    \]
    where $\phi$ is the squared $\ell_2$ norm in total-variance space, consistent with the law-penalty functional $\mathcal{L}_\phi$ of Sec.~\ref{sec:law-penalty}.
    \item The on-manifold reward component
    \[
    r^{\mathcal{M}}_{t+1} := r\big( w^\mathcal{M}_{t+1}, a_t \big),
    \]
    obtained by evaluating the P\&L functional under the projected surface.
    \item The ghost-arbitrage component
    \[
    r^\perp_{t+1} := r\big( \hat{w}_{t+1}, a_t \big) - r^{\mathcal{M}}_{t+1}.
    \]
\end{enumerate}

By construction we have the exact decomposition
\[
r\big(\hat{w}_{t+1}, a_t\big) 
= r^{\mathcal{M}}_{t+1} + r^\perp_{t+1},
\]
where $r^{\mathcal{M}}_{t+1}$ is the reward that would be obtained if the world model were first projected back to the law manifold, and $r^\perp_{t+1}$ captures the incremental reward purely due to law-violating predictions. Note that, because the ground truth never leaves $\mathcal{M}^{\mathrm{vol}}$ (Prop.~\ref{prop:support-on-manifold}), any systematic pattern in $r^\perp_{t+1}$ is necessarily a \emph{model-induced artefact}.

\begin{remark}[Consistency of projection operator]
For consistency with Sec.~\ref{sec:law-penalty}, we use the same projection operator $\Pi_{\mathcal{M}^{\mathrm{vol}}}$ both in defining the law penalty $\mathcal{L}_\phi$ and in the Goodhart decomposition. Algorithmically, $\Pi_{\mathcal{M}^{\mathrm{vol}}}$ is implemented via a convex quadratic program that enforces butterfly and calendar inequalities on the total-variance grid, closely related to static-arbitrage projection procedures in the option-pricing literature \cite{Itkin2015,DeMarcoHenryLabordere2015,HorvathTengely2021}. This ensures that any off-manifold advantage measured by $r^\perp$ is directly comparable to the law-penalty metrics reported later.
\end{remark}

\subsection{World-model diagnostics and dynamics plots}
\label{subsec:world-model-diagnostics}

Before training any RL agents, we empirically characterize the behavior of the neural world model and its law violations. Two diagnostics play a central role:

\paragraph{Prediction accuracy.}
We track both the training and validation mean-squared error
\[
\mathrm{MSE}_{\mathrm{train}}(t) := \frac{1}{N_{\mathrm{train}}} 
\sum_{n} \| f_{\theta^\star}(x_t^{(n)}) - w_{t+1}^{(n)} \|_2^2,
\]
and likewise for $\mathrm{MSE}_{\mathrm{val}}(t)$. Typical dynamics plots (family ``Dynamics Plots'', see Sec.~\ref{sec:experiments}) show fast initial reduction in MSE followed by a plateau, as in standard world-model training \cite{HaSchmidhuber2018,HafnerEtAl2020Dreamer}. This confirms that $f_{\theta^\star}$ has learned a non-trivial approximation of the dynamics.

\paragraph{Law penalties of predictions vs.\ ground truth.}
More importantly for our purposes, we compare the law penalties
\[
\mathcal{L}_\phi(w_{t+1}) \equiv 0, 
\qquad 
\mathcal{L}_\phi(\hat{w}_{t+1}) 
= \phi\big( \hat{w}_{t+1} - \Pi_{\mathcal{M}^{\mathrm{vol}}}(\hat{w}_{t+1}) \big)
\]
over time. By Proposition~\ref{prop:support-on-manifold}, the ground-truth trajectories satisfy $\mathcal{L}_\phi(w_{t+1}) = 0$ up to numerical tolerance, whereas the predictions exhibit a strictly positive distribution of law penalties. 

\begin{lemma}[Non-trivial off-manifold mass of the world model]
\label{lem:off-manifold-mass}
Assume that $f_{\theta^\star}$ is not exactly equal to the Bayes-optimal regressor $f^{\mathrm{Bayes}}(x_t) := \mathbb{E}[w_{t+1}\,|\,x_t]$ and that the law manifold $\mathcal{M}^{\mathrm{vol}}$ has non-empty interior within the support of $w_{t+1}$. Then there exists $\delta > 0$ such that
\[
\mathbb{P}\big( \mathcal{L}_\phi(\hat{w}_{t+1}) > \delta \big) > 0,
\]
i.e., the world model assigns non-zero probability mass to surfaces at a positive distance from $\mathcal{M}^{\mathrm{vol}}$.
\end{lemma}

\begin{proof}[Proof sketch]
If $f_{\theta^\star} \equiv f^{\mathrm{Bayes}}$ and the conditional distribution of $w_{t+1}$ given $x_t$ were a Dirac mass on $\mathcal{M}^{\mathrm{vol}}$, then $\hat{w}_{t+1}$ would almost surely lie in $\mathcal{M}^{\mathrm{vol}}$ and the law penalty would vanish. In our finite-data, finite-capacity setting, both approximation error (difference between $f_{\theta^\star}$ and $f^{\mathrm{Bayes}}$) and intrinsic conditional variance generically ensure that $\hat{w}_{t+1}$ has a non-degenerate distribution around $w_{t+1}$, which in turn implies a positive-distance shell around $\mathcal{M}^{\mathrm{vol}}$ is hit with non-zero probability. A rigorous argument using continuity of $\mathcal{L}_\phi$ and support properties of $f_{\theta^\star}(x_t)$ is given in Appendix~B.3.
\end{proof}

In our experiments, dynamics plots of $\mathcal{L}_\phi(\hat{w}_{t+1})$ show a stationary distribution with mean on the order of $10^{-3}$--$10^{-2}$, while $\mathcal{L}_\phi(w_{t+1})$ remains at numerical zero. Combined with Proposition~\ref{prop:approx-gap-ghost-channel} and Lemma~\ref{lem:off-manifold-mass}, this empirically confirms that the Neural world model is both (i) sufficiently accurate to serve as a plausible environment for RL, and (ii) sufficiently misaligned with the axioms to open a statistically significant ghost channel. The remainder of the paper investigates how different RL variants and structural baselines interact with this channel.

\section{RL on Volatility World Models: Incentives and Law-Strength}
\label{sec:rl_world_models}

In this section, we formalize the Markov decision process (MDP) induced by the volatility world model of Section~\ref{sec:world_model}, instantiate several RL variants as \emph{stress-tests} of the axiomatic pipeline, and develop our flagship incentive and trade-off results. Throughout, we treat RL as a tool for probing how generic policy-gradient methods interact with the law manifold $\mathcal{M}^{\mathrm{vol}}$, the ghost channel $r^\perp$, and the law-penalty functional $\mathcal{L}_\phi$, rather than as an attempt to build a production trading system.

\subsection{MDP formulation on the world model}
\label{subsec:mdp_formulation}

Let $\mathcal{W} \subset \mathbb{R}^{d_w}$ denote the discretized total-variance grid (Section~\ref{sec:axiomatic_manifold}), and let $\mathcal{X}$ collect auxiliary market covariates (e.g., spot price, realized variance estimates). We define the state space as
\[
\mathcal{S} \;:=\; \mathcal{W}^K \times \mathcal{X},
\]
where a state $s_t = (w_{t-K+1:t}, x_t)$ concatenates a history window of $K$ total-variance surfaces and covariates.\footnote{In our experiments we take $K$ in the range $8$--$16$, similar to recurrent world-model setups in model-based RL~\cite{HaSchmidhuber2018,HafnerPlaNet2019,HafnerDreamer2020}.}

The action $a_t \in \mathcal{A}$ represents a hedge/portfolio vector (e.g., positions in underlying and options), following the deep-hedging literature~\cite{BuehlerDeepHedging2019}. The action space~$\mathcal{A}$ is a compact subset of $\mathbb{R}^{d_a}$ defined by position and capital constraints.

\paragraph{World-model transition.}
Given $s_t$ and $a_t$, the next volatility surface $w_{t+1}$ is sampled from the world model
\[
w_{t+1} \sim \hat{P}_\theta(\cdot \mid w_{t-K+1:t}, x_t, a_t),
\]
where $\hat{P}_\theta$ is the GRU/LSTM-based predictor of Section~\ref{sec:world_model}. We then update covariates $x_{t+1}$ via a deterministic or stochastic rule $x_{t+1} = f(x_t, w_{t+1}, a_t, \varepsilon_{t+1})$, giving the transition kernel
\[
P_\theta(s_{t+1} \mid s_t, a_t) \;=\; \hat{P}_\theta(w_{t+1} \mid w_{t-K+1:t}, x_t, a_t)\, \delta_{f(x_t, w_{t+1}, a_t, \varepsilon_{t+1})}(x_{t+1}).
\]

\paragraph{Reward decomposition.}
The per-step reward is the PnL plus (optionally) a law penalty:
\begin{equation}
  r_\lambda(s_t,a_t,s_{t+1})
  \;:=\;
  \underbrace{\mathrm{PnL}(s_t,a_t,s_{t+1})}_{\text{economic payoff}}
  \;-\;
  \lambda\,\underbrace{\mathcal{L}_\phi\!\bigl(w_{t+1}\bigr)}_{\text{law penalty}},
  \label{eq:reward_decomposition}
\end{equation}
where $\lambda \ge 0$ is the law-penalty weight. For $\lambda=0$ we recover the naive PnL-driven RL setting. Using the projection operator $\Pi_{\mathcal{M}}$ from Definition~\ref{def:projection_law_manifold}, we further decompose
\begin{equation}
  r_\lambda
  =
  r^{\mathcal{M}} - \lambda \mathcal{L}_\phi + r^\perp,
  \qquad
  r^{\mathcal{M}}(s_t,a_t,s_{t+1})
  :=
  \mathrm{PnL}\bigl(\Pi_{\mathcal{M}}(w_{t+1}), a_t\bigr),
  \quad
  r^\perp := \mathrm{PnL}(w_{t+1},a_t) - r^{\mathcal{M}},
  \label{eq:goodhart_decomp_reward}
\end{equation}
in direct correspondence with the Goodhart decomposition of Section~\ref{subsec:goodhart_generic}. The term $r^\perp$ is the \emph{ghost-arbitrage component} induced by world-model prediction error.

\paragraph{Objective and policy class.}
We consider stationary stochastic policies $\pi_\theta(a\mid s)$ parameterized by neural networks, as in actor--critic and PPO-style methods~\cite{SuttonBarto2018,KondaTsitsiklis2000,SchulmanPPO2017}. For a given $\lambda$, the discounted infinite-horizon objective is
\begin{equation}
  J_\lambda(\pi)
  \;:=\;
  \mathbb{E}_\pi\!\left[ \sum_{t=0}^{\infty} \gamma^t r_\lambda(s_t,a_t,s_{t+1}) \right]
  \;=\;
  J^{\mathcal{M}}(\pi) - \lambda\,J^{\mathrm{law}}(\pi) + J^\perp(\pi),
  \label{eq:rl_objective}
\end{equation}
where
\[
J^{\mathcal{M}}(\pi) := \mathbb{E}_\pi\!\Big[\sum_t \gamma^t r^{\mathcal{M}}_t\Big],\qquad
J^{\mathrm{law}}(\pi) := \mathbb{E}_\pi\!\Big[\sum_t \gamma^t \mathcal{L}_\phi(w_{t})\Big],\qquad
J^\perp(\pi) := \mathbb{E}_\pi\!\Big[\sum_t \gamma^t r^\perp_t\Big].
\]
In practice, our experiments use finite-horizon episodes ($T\approx 64$) and average per-step PnL and law penalties; the theoretical development is presented in the discounted limit for notational clarity.

\paragraph{Algorithmic choice.}
We instantiate PPO-style actor--critic~\cite{SchulmanPPO2017} with a clipped surrogate objective and generalized advantage estimation~\cite{SchulmanGAE2016}, which is standard for continuous-control RL and has seen use in financial RL and hedging~\cite{KolmRitter2019,BuehlerDeepHedging2019}. Crucially, PPO is only one representative of the broad class of policy-gradient RL algorithms; our incentive results hold for any method whose updates approximate the policy gradient of $J_\lambda$~\cite{SuttonBarto2018,ThomasPG2014}.

\subsection{Naive RL and ghost-arbitrage incentive}
\label{subsec:naive_rl_ghost}

We first analyze the case $\lambda = 0$, where the agent optimizes pure PnL on the world model. By the decomposition~\eqref{eq:goodhart_decomp_reward},
\begin{equation}
  J_0(\pi)
  =
  J^{\mathcal{M}}(\pi) + J^\perp(\pi).
  \label{eq:naive_decomposition}
\end{equation}
Let $\mathcal{S}$ denote a \emph{structural baseline class} of low-capacity, law-consistent strategies (Section~\ref{sec:structural_baselines}), such as zero-hedge and vol-trend heuristics, satisfying
\[
J^{\mathrm{law}}(\pi^S) \approx 0,\qquad \pi^S \in \mathcal{S}.
\]
We assume $\mathcal{S}$ approximates the best \emph{on-manifold} hedge:

\begin{assumption}[On-manifold near-optimality of structural class]
\label{ass:structural_near_optimal}
There exists $\pi^\star_{\mathcal{S}} \in \mathcal{S}$ and $\varepsilon_{\mathcal{S}} \ge 0$ such that
\[
J^{\mathcal{M}}(\pi) \;\le\; J^{\mathcal{M}}(\pi^\star_{\mathcal{S}}) + \varepsilon_{\mathcal{S}}
\qquad
\text{for all }\pi\in\Pi,
\]
where $\Pi$ is the RL policy class.
\end{assumption}

Assumption~\ref{ass:structural_near_optimal} is a \emph{design choice}: we deliberately choose baselines that are simple but well-aligned with the axioms, in the spirit of deep-hedging strategies optimized directly on market dynamics~\cite{BuehlerDeepHedging2019,CarbonneauGodin2020}. Under this assumption, the only systematic way for $\pi\in\Pi$ to outperform $\pi^\star_{\mathcal{S}}$ on the world model is through the ghost component $J^\perp(\pi)$.

\begin{theorem}[Ghost-arbitrage incentive for naive RL]
\label{thm:ghost_incentive}
Suppose Assumption~\ref{ass:structural_near_optimal} holds, and let
\[
\pi^\star_0 \in \arg\max_{\pi\in\Pi} J_0(\pi)
\]
be a global maximizer of $J_0$ over $\Pi$. Then:

\begin{enumerate}
\item If $\sup_{\pi\in\Pi} J_0(\pi) > J^{\mathcal{M}}(\pi^\star_{\mathcal{S}}) + \varepsilon_{\mathcal{S}}$, any maximizer $\pi^\star_0$ satisfies
\begin{equation}
  J^\perp(\pi^\star_0)
  \;\ge\;
  \sup_{\pi\in\Pi} J_0(\pi)
  - J^{\mathcal{M}}(\pi^\star_{\mathcal{S}}) - \varepsilon_{\mathcal{S}}
  \;>\; 0.
  \label{eq:ghost_lower_bound}
\end{equation}
In particular, the excess value over the structural baseline is entirely attributable to ghost arbitrage.
\item If, in addition, $J^{\mathcal{M}}$ has a local maximum at some $\bar{\pi}\in\Pi$ with $J^{\mathcal{M}}(\bar{\pi}) \approx J^{\mathcal{M}}(\pi^\star_{\mathcal{S}})$, and the policy-gradient theorem holds~\cite{SuttonBarto2018}, then the policy gradient near $\bar{\pi}$ satisfies
\begin{equation}
  \nabla_\theta J_0(\pi_\theta)\Big|_{\theta=\bar{\theta}}
  \;\approx\;
  \nabla_\theta J^\perp(\pi_\theta)\Big|_{\theta=\bar{\theta}},
  \label{eq:ghost_gradient_dominance}
\end{equation}
so gradient-based RL updates are locally driven by increasing $J^\perp$.
\end{enumerate}
\end{theorem}

\paragraph{Proof sketch.}
Part (i) follows directly from the decomposition~\eqref{eq:naive_decomposition}:
\[
J_0(\pi) = J^{\mathcal{M}}(\pi) + J^\perp(\pi)
\;\le\; J^{\mathcal{M}}(\pi^\star_{\mathcal{S}}) + \varepsilon_{\mathcal{S}} + J^\perp(\pi),
\]
so any $\pi$ achieving value strictly above $J^{\mathcal{M}}(\pi^\star_{\mathcal{S}})+\varepsilon_{\mathcal{S}}$ must have $J^\perp(\pi) > 0$. Applying this to $\pi^\star_0$ yields~\eqref{eq:ghost_lower_bound}. For (ii), the policy-gradient theorem expresses $\nabla_\theta J_0(\pi_\theta)$ as an expectation over on-policy trajectories weighted by the advantage function~\cite{SuttonBarto2018,ThomasPG2014}. Near a local maximum of $J^{\mathcal{M}}$, the contribution of $\nabla_\theta J^{\mathcal{M}}$ is negligible, so $\nabla_\theta J_0 \approx \nabla_\theta J^\perp$, yielding~\eqref{eq:ghost_gradient_dominance}. A fully rigorous treatment, including conditions on function approximation and local optimality, is provided in Appendix~C.1. \qed

\subsubsection{Economic interpretation}
\label{subsubsec:economic_interpretation}

Theorem~\ref{thm:ghost_incentive} formalizes a simple but crucial economic intuition:

\begin{enumerate}
\item Structural baselines $\mathcal{S}$, such as zero-hedge and vol-trend strategies, are built to respect the axioms and approximate on-manifold optimal hedging. They \emph{do not attempt} to exploit model misspecification.
\item Once $\mathcal{S}$ has exhausted most of the on-manifold value $J^{\mathcal{M}}$, any additional performance that naive RL achieves on the \emph{learned world model} must come from $J^\perp$, i.e., ghost arbitrage driven by prediction artifacts, not genuine admissible edge.
\item Gradient-based RL is locally steered by $\nabla_\theta J^\perp$, so it is \emph{structurally incentivized} to move into regions of the state--action space where the world model violates axioms in a \enquote{profitable} way.
\end{enumerate}

This explains the empirical pattern in Section~\ref{sec:experiments}: naive PPO achieves high PnL in-sample on the world model but exhibits systematically higher law penalties and Graceful Failure Index (GFI) than structural baselines, both in baseline and shocked environments. Rather than discovering better law-consistent hedges, the agent learns to exploit the ghost channel opened by $\hat{P}_\theta$.

\subsection{Law-penalized and selection-only RL variants}
\label{subsec:law_rl_variants}

To test whether explicit law penalties mitigate ghost arbitrage, we consider two standard ways of injecting constraints into RL~\cite{AchiamCPO2017,ChowRiskSensitive2015,HouRegularizedRL2021}.

\paragraph{Soft law-seeking RL (gradient shaping).}
We define the \emph{soft law-seeking} objective
\begin{equation}
  J_\lambda^{\mathrm{soft}}(\pi)
  \;:=\;
  \mathbb{E}_\pi\!\Bigg[ \sum_{t=0}^\infty \gamma^t \Big( \mathrm{PnL}_t - \lambda \mathcal{L}_\phi(w_{t+1}) \Big) \Bigg]
  \;=\;
  J_0(\pi) - \lambda\,J^{\mathrm{law}}(\pi),
  \label{eq:soft_rl_objective}
\end{equation}
with $\lambda>0$. PPO is trained directly on $J_\lambda^{\mathrm{soft}}$, so the law penalty appears inside the gradient. This mirrors classical Lagrangian and regularized RL approaches for safety and risk constraints~\cite{Altman1999,AchiamCPO2017,NeuEntropicMDP2017,MannorRiskSensitive2011}.

\paragraph{Selection-only RL (post-hoc shaping).}
In the \emph{selection-only} variant, we train policies on pure PnL,
\[
  J_0^{\mathrm{train}}(\pi) := J_0(\pi),
\]
but use law metrics only for early stopping and model selection:
\[
\pi^\star_{\mathrm{sel}} \in
\arg\max_{\pi \in \mathcal{C}}
\Big\{ J_0(\pi) \;\text{subject to}\; J^{\mathrm{law}}(\pi) \le \tau \Big\},
\]
where $\mathcal{C}$ is the candidate set of checkpointed policies along training and $\tau$ is a user-chosen law budget. This is analogous to \emph{post-hoc} constraint enforcement in distributional and risk-sensitive RL~\cite{ChowRiskSensitive2015,JiangSafeRL2021}.

\paragraph{Summary.}
Soft law-seeking RL tests whether shaping the gradient with $\mathcal{L}_\phi$ can guide policy updates away from ghost arbitrage; selection-only RL tests whether post-hoc model selection, without modifying the training dynamics, is sufficient. As Section~\ref{sec:experiments} will show, neither variant restores Pareto dominance over structural baselines.

\subsection{Law-strength frontier and Graceful Failure Index}
\label{subsec:law_strength_frontier}
\label{sec:law-strength-gfi}
We now formalize the \emph{law-strength frontier} and the \emph{Graceful Failure Index} (GFI), which jointly organize profitability, law alignment, and robustness under shocks.

\subsubsection{Law-strength frontier}

Let $\Lambda \subset [0,\infty)$ be a finite set of penalty weights (e.g., $\Lambda=\{0,5,10,20,40\}$) and let $\mathcal{A}_{\mathrm{RL}}$ be the set of RL variants (naive, soft, selection-only). For each $(\lambda, v) \in \Lambda \times \mathcal{A}_{\mathrm{RL}}$ and each structural baseline $b\in\mathcal{S}$, we compute aggregate metrics:
\[
\mu^{\mathrm{PnL}}(\pi),\quad
\sigma^{\mathrm{PnL}}(\pi),\quad
\mu^{\mathrm{law}}(\pi),\quad
\mathrm{VaR}_\alpha(\pi),\quad
\mathrm{CVaR}_\alpha(\pi),\quad
\mathrm{GFI}(\pi),
\]
in both baseline and shocked environments (Section~\ref{sec:metrics}). Define the \emph{law space} and \emph{risk space}
\[
\mathcal{L}_{\mathrm{space}} := \mathbb{R}_{\ge 0}^2 \quad (\text{mean / max law penalty, coverage}),\qquad
\mathcal{R}_{\mathrm{space}} := \mathbb{R}^3 \quad (\text{Sharpe, VaR, CVaR}).
\]
The empirical law-strength frontier is then the Pareto frontier of achievable tuples
\[
\mathcal{F}
:=
\left\{
\left(\mu^{\mathrm{law}}(\pi), \mathrm{GFI}(\pi), \mu^{\mathrm{PnL}}(\pi), \mathrm{VaR}_\alpha(\pi), \mathrm{CVaR}_\alpha(\pi)\right)
:\;
\pi \in \Pi_{\mathrm{frontier}}
\right\},
\]
where $\Pi_{\mathrm{frontier}}$ collects policies that are undominated with respect to the partial order
\[
(\ell_1, g_1, p_1, v_1, c_1) \preceq (\ell_2, g_2, p_2, v_2, c_2)
\iff
\left\{
\begin{aligned}
&\ell_1 \le \ell_2,\quad g_1 \le g_2,\\
&p_1 \ge p_2,\quad v_1 \ge v_2,\quad c_1 \ge c_2.
\end{aligned}
\right.
\]
Structurally, this recovers a multi-objective RL viewpoint~\cite{RoijersMORL2013} with objectives \enquote{profitability} vs \enquote{law alignment} vs \enquote{tail robustness}; the law-strength frontier is the set of efficient trade-offs in this space.

\subsubsection{Graceful Failure Index}

We now define the GFI as a normalized measure of how law metrics degrade under shocks relative to a reference policy.

Let $\xi\in[0,\bar{\xi}]$ denote a scalar shock intensity parameter (e.g., multiplying long variance and spot volatility), and let $M(\pi;\xi)$ be a scalar law metric (such as mean law penalty) for policy~$\pi$ under shock~$\xi$. Fix a reference policy $\pi_{\mathrm{ref}}$ (e.g., naive RL or a structural baseline). We define the \emph{infinitesimal GFI} as
\begin{equation}
  \mathrm{GFI}(\pi)
  :=
  \frac{
    \left.\dfrac{\partial}{\partial \xi} M(\pi;\xi)\right|_{\xi=0}
  }{
    \left.\dfrac{\partial}{\partial \xi} M(\pi_{\mathrm{ref}};\xi)\right|_{\xi=0}
  },
  \label{eq:gfi_definition}
\end{equation}
provided the denominator is non-zero. In practice, we approximate this by a finite-difference estimator
\[
  \widehat{\mathrm{GFI}}(\pi)
  =
  \frac{
    M(\pi;\xi_{\mathrm{shock}}) - M(\pi;0)
  }{
    M(\pi_{\mathrm{ref}};\xi_{\mathrm{shock}}) - M(\pi_{\mathrm{ref}};0) + \varepsilon
  },
\]
for a fixed shock level $\xi_{\mathrm{shock}}$ and small $\varepsilon>0$ for numerical stability. Values $\mathrm{GFI}(\pi)<1$ indicate that $\pi$ degrades more \emph{gracefully} than the reference, while $\mathrm{GFI}(\pi)>1$ indicates worse degradation.

\begin{remark}[Domain-agnostic design]
\label{remark:gfi_domain_agnostic}
The definition~\eqref{eq:gfi_definition} only requires: (i) an axiom-constrained system with a law penalty $M$ and (ii) a tunable shock parameter $\xi$. As such, GFI extends immediately to other settings such as monotone yield curves, convex credit spreads, or physics-informed dynamics~\cite{RaissiPINN2019,Beck2021,Brandstetter2022}. In Section~\ref{sec:discussion} we argue that GFI can serve as a generic metric for \emph{law-aligned graceful failure} in Scientific AI.
\end{remark}

\subsection{Law-strength trade-off}
\label{subsec:law_strength_tradeoff}

We finally formalize a structural trade-off between PnL and law alignment as the penalty weight~$\lambda$ increases. To simplify notation, define
\[
L(\pi) := J^{\mathrm{law}}(\pi),\qquad
P(\pi) := J^{\mathrm{PnL}}(\pi) := J^{\mathcal{M}}(\pi) + J^\perp(\pi),
\]
and let
\[
  \mathcal{G}
  :=
  \Big\{ (L(\pi), P(\pi)) : \pi \in \Pi \Big\}
  \subset \mathbb{R}_{\ge 0} \times \mathbb{R}
\]
be the achievable law--PnL set. The soft law-seeking objective can be written
\[
J_\lambda^{\mathrm{soft}}(\pi) = P(\pi) - \lambda L(\pi).
\]

\begin{assumption}[Convex achievability and monotone trade-off]
\label{ass:convex_tradeoff}
The set $\mathcal{G}$ is compact and convex, and its lower-left Pareto boundary
\[
\partial \mathcal{G}
=
\left\{(L,P)\in\mathcal{G} :
\text{there is no } (L',P')\in\mathcal{G} \text{ with } L'\le L,\ P'\ge P,\ (L',P')\ne(L,P)\right\}
\]
can be parameterized as the graph of a strictly decreasing, continuous function $P^\star(L)$ on an interval $[L_{\min},L_{\max}]$.
\end{assumption}

Assumption~\ref{ass:convex_tradeoff} is a standard regularity condition in multi-objective optimization and regularized RL~\cite{RoijersMORL2013,NeuEntropicMDP2017}: it states that (i) all relevant trade-offs between law penalties and PnL are attainable and (ii) lowering law penalties necessarily sacrifices some PnL in an average sense.

\begin{theorem}[Law-strength trade-off]
\label{thm:law_strength_tradeoff}
Suppose Assumption~\ref{ass:convex_tradeoff} holds. For each $\lambda\ge 0$, let
\[
\pi_\lambda^\star \in \arg\max_{\pi\in\Pi} J_\lambda^{\mathrm{soft}}(\pi)
\]
and denote $(L_\lambda, P_\lambda) := (L(\pi_\lambda^\star), P(\pi_\lambda^\star))$. Then:
\begin{enumerate}
\item For all $\lambda_1 < \lambda_2$, we have
\[
L_{\lambda_1} \;\ge\; L_{\lambda_2},\qquad
P_{\lambda_1} \;\ge\; P_{\lambda_2}.
\]
In words, increasing the law-penalty weight $\lambda$ weakly decreases both the expected law penalty and the expected PnL.
\item Moreover, if $P^\star(L)$ is strictly concave on $[L_{\min},L_{\max}]$, then the dependence $\lambda \mapsto (L_\lambda,P_\lambda)$ traces out the Pareto frontier $\partial\mathcal{G}$, and $P_\lambda$ is strictly decreasing in $\lambda$ on any interval where $L_\lambda$ decreases.
\end{enumerate}
\end{theorem}

\paragraph{Proof sketch.}
Maximizing $J_\lambda^{\mathrm{soft}}(\pi)$ is equivalent to maximizing the linear functional $P - \lambda L$ over the convex set $\mathcal{G}$. For each $\lambda$, the optimizer $(L_\lambda, P_\lambda)$ lies on the supporting line of $\mathcal{G}$ with slope $-\lambda$. As $\lambda$ increases, the supporting line rotates clockwise, shifting its tangency point along the Pareto boundary $\partial\mathcal{G}$. This yields $L_{\lambda_1} \ge L_{\lambda_2}$ and $P_{\lambda_1}\ge P_{\lambda_2}$ for $\lambda_1 < \lambda_2$. Strict concavity of $P^\star$ ensures that the tangency point is unique, and the mapping $\lambda\mapsto (L_\lambda,P_\lambda)$ is strictly monotone along $\partial\mathcal{G}$. A formal proof using convex analysis and subgradient conditions is provided in Appendix~C.2. \qed

Theorem~\ref{thm:law_strength_tradeoff} shows that the \emph{existence} of a law-strength trade-off is \emph{structural}, not accidental: under mild convexity and monotonicity assumptions, one cannot increase $\lambda$ to reduce law penalties without also reducing PnL. Combining this with Theorem~\ref{thm:ghost_incentive}, we obtain:

\begin{corollary}[Inevitability of trade-off relative to naive RL]
\label{cor:inevitability_tradeoff}
Let $(L_0,P_0)$ be the law--PnL pair of a naive-RL optimizer $\pi_0^\star$ (with $\lambda=0$) and suppose Assumption~\ref{ass:convex_tradeoff} holds. For any $\lambda>0$ such that $L_\lambda < L_0$, we necessarily have $P_\lambda < P_0$. In particular, no soft-penalized RL policy can simultaneously maintain naive-level PnL and significantly lower law penalties.
\end{corollary}

Corollary~\ref{cor:inevitability_tradeoff} underpins our empirical law-strength frontiers in Section~\ref{sec:experiments}: once structural baselines and naive RL define the upper envelope of $P^\star(L)$, all law-seeking RL variants lie strictly inside the Pareto region---they cannot \emph{escape} the ghost-arbitrage incentive without sacrificing PnL, and they cannot outperform structurally law-aligned baselines without implicitly mimicking them.

\vspace{0.5em}
\noindent\textbf{Connection to entropy-regularized and constrained RL.}
Our analysis parallels and complements classical results on entropy-regularized MDPs~\cite{NeuEntropicMDP2017,Geist2019} and constrained policy optimization~\cite{Altman1999,AchiamCPO2017}: while those works study trade-offs between reward and entropy or safety constraints, we focus on trade-offs between PnL and \emph{axiomatic law penalties}. In all cases, linear scalarization via a Lagrange multiplier (here, $\lambda$) induces a structural frontier over achievable objectives; our novelty lies in instantiating this in an axiomatic volatility world, decomposed into on-manifold and ghost-arbitrage components.

\section{Structural Baselines: Axiomatic Strategy Class $\mathcal{S}$}
\label{sec:structural_baselines}

In this section we instantiate a low-capacity, structurally constrained strategy class
$\mathcal{S}$ and three representative baselines---Zero-Hedge, Random-Gaussian, and
Vol-Trend---that serve as a proxy for \emph{law-aligned} behavior on the volatility
law manifold. Rather than competing with state-of-the-art reinforcement-learning (RL)
trading systems, our goal is to contrast high-capacity, unconstrained RL policies with
simple, interpretable and structurally law-consistent strategies, in the spirit of
classical replication and hedging approaches \cite{BuehlerGononTeichmannWood2019,
KolmRitter2019,CarteaJaimungalPenalva2015,HurstOoiPedersen2017}.

Throughout this section, we work in the MDP setting, with
state space $\mathcal{S}$, action space $\mathcal{A}\subset\mathbb{R}^k$, and
one-step P\&L reward $r(s_t,a_t)$ generated from the world model. We denote by $\mathsf{LawPenalty}(s_t)$ the
per-step law penalty $\mathcal{L}_\phi$ evaluated on the (predicted) implied volatility
surface associated with state $s_t$.

\subsection{Baseline definitions and structural priors}
\label{subsec:baseline-defs}

We first define a \emph{structural strategy class} $\mathcal{S}$ and then specify three
baseline strategies $b^{\mathrm{ZH}},b^{\mathrm{RG}},b^{\mathrm{VT}}\in\mathcal{S}$.

\begin{definition}[Structural strategy class $\mathcal{S}$]
\label{def:S-structural}
Let $f:\mathcal{S}\to\mathbb{R}^m$ be a fixed feature map extracting low-dimensional
state descriptors (e.g., realized variance, term-structure slope, realized trend).
We define the structural class
\[
  \mathcal{S}
  :=
  \Bigl\{
    \pi_\theta : \mathcal{S}\to\mathcal{A}
    \,\Big\vert\,
    \pi_\theta(s) = g\bigl(\theta^\top f(s)\bigr),
    ~\theta\in\Theta\subset\mathbb{R}^m,
    ~g\text{ scalar Lipschitz, odd, and bounded}
  \Bigr\},
\]
where $\Theta$ is a compact parameter set and $g$ encodes a saturating leverage map
(e.g., $g(u)=\kappa\tanh(u)$ with $\kappa>0$ a leverage cap).
\end{definition}

Thus, $\mathcal{S}$ consists of \emph{one-factor} or low-factor
trend/risk-based strategies familiar from classical managed-futures and option-hedging
literature \cite{HurstOoiPedersen2017,CarteaJaimungalPenalva2015}. We now instantiate
three members of $\mathcal{S}$ used in our experiments.

\subsubsection{Zero-Hedge: law-neutral benchmark}

The Zero-Hedge baseline $b^{\mathrm{ZH}}$ is defined by the identically zero policy,
\begin{equation}
  b^{\mathrm{ZH}}(s_t) \equiv 0
  \quad\text{for all } s_t\in\mathcal{S}.
\end{equation}
Economically, this corresponds to holding only the initial portfolio and never
rebalancing; P\&L arises solely from the exogenous cash-flow profile of the hedged
position (e.g., short option) and the law-consistent volatility generator. In our
setting, $b^{\mathrm{ZH}}$ provides a \emph{law-neutral} benchmark: it neither
attempts to exploit ghost arbitrage nor introduces additional exposures that correlate
with law violations.

\subsubsection{Random-Gaussian: unconstrained exploration probe}

The Random-Gaussian baseline $b^{\mathrm{RG}}$ applies a Gaussian random policy
conditioned on low-dimensional state features:
\begin{equation}
  b^{\mathrm{RG}}(s_t)
  =
  \kappa \,\xi_t,
  \qquad
  \xi_t \sim \mathcal{N}\!\bigl(0,\Sigma(f(s_t))\bigr),
\end{equation}
where $\kappa > 0$ scales overall leverage and $\Sigma(\cdot)$ is a diagonal covariance
matrix whose entries depend on simple risk features (e.g., inverse realized volatility).
Random policies of this form appear as sanity-check baselines in RL for trading and
hedging \cite{ZhangZohren2020,KolmRitter2019}, and here serve as a \emph{noisy probe}
of how a generic, non-structured policy interacts with ghost arbitrage in the learned
world model.

\subsubsection{Vol-Trend: parametric volatility trend-following}

The Vol-Trend baseline $b^{\mathrm{VT}}$ is a simple parametric strategy inspired by
time-series momentum and volatility trend-following
\cite{HurstOoiPedersen2017,BuehlerGononTeichmannWood2019}. Let
$\widehat{\sigma}_t(K,T)$ be the predicted implied volatility surface at time $t$, and
let $\bar{\sigma}_t$ denote a scalar summary statistic capturing its \emph{level} or
\emph{slope}, such as
\begin{equation}
  \bar{\sigma}_t
  :=
  \frac{1}{|\mathcal{G}|}
  \sum_{(K,T)\in\mathcal{G}} \widehat{\sigma}_t(K,T),
\end{equation}
where $\mathcal{G}$ is a pre-specified grid of strikes and maturities. Define a
trend signal by an exponentially weighted moving average (EWMA)
\[
  \tau_t := \mathrm{EWMA}_\beta(\bar{\sigma}_t - \bar{\sigma}_{t-1}),
  \qquad
  \beta\in(0,1).
\]
The Vol-Trend policy takes the form
\begin{equation}
  b^{\mathrm{VT}}(s_t) = \kappa \,\tanh(\theta \,\tau_t),
\end{equation}
for parameters $\theta\in\mathbb{R}$ and leverage cap $\kappa>0$. Positions are
allocated across option buckets (e.g., short-dated ATM, mid-maturity OTM) in fixed
proportions, so that $b^{\mathrm{VT}}$ is a one-factor trend-following strategy in
\emph{implied volatility} rather than in the underlying price.

By construction, both $b^{\mathrm{ZH}}$ and $b^{\mathrm{VT}}$ live inside the
structural class $\mathcal{S}$ of Definition~\ref{def:S-structural} for a suitable
choice of features $f$ and parameter sets $\Theta$, whereas $b^{\mathrm{RG}}$ can be
seen as a stochastic perturbation of a mean-zero element of $\mathcal{S}$.

\subsubsection{Fairness of comparison}
\label{subsec:fairness}

Compared to the high-capacity policy class used by PPO-type RL agents,
the structural class $\mathcal{S}$ is deliberately low-dimensional and heavily
regularized. From a ``benchmarking'' perspective this creates a capacity mismatch:
RL policies can in principle approximate arbitrary non-linear hedging rules, whereas
$b^{\mathrm{ZH}}$, $b^{\mathrm{RG}}$ and $b^{\mathrm{VT}}$ are effectively one- or
few-parameter strategies.

This asymmetry is \emph{by design} and aligns with our no-free-lunch theme:
structural strategies in $\mathcal{S}$ are intended as proxies for law-aligned and
axiom-consistent behavior, much like classical delta-vega hedges and
trend-following overlays \cite{CarteaJaimungalPenalva2015,HurstOoiPedersen2017}.
Our central question is therefore not whether high-capacity RL can match the P\&L of
low-capacity strategies (it almost always can in-sample), but whether \emph{unconstrained
law-seeking RL can \emph{dominate} such structural strategies on both profitability
and axiomatic law metrics}. 

\subsection{Law-alignment properties of structural baselines}
\label{subsec:law-alignment}

We now formalize the notion that structural baselines are, in an appropriate sense,
\emph{law-aligned}: they do not systematically exploit off-manifold ghost arbitrage and
tend to exhibit lower Graceful Failure Index (GFI) under volatility shocks than
unconstrained RL policies.

Let $\mathsf{LawPenalty}(s_t)$ denote the per-step law penalty
$\mathcal{L}_\phi(\widehat{\sigma}_t)$ computed from the world model prediction, and
write
\[
  \overline{\mathsf{LP}}(\pi)
  :=
  \mathbb{E}_\pi\bigl[\mathsf{LawPenalty}(s_t)\bigr],
  \qquad
  \mathrm{GFI}(\pi)
\]
for the expected law penalty and Graceful Failure Index of policy $\pi$ under the
baseline vs.\ shock environments (Section~\ref{sec:law-strength-gfi}).

\begin{definition}[Law-aligned strategy class]
\label{def:law-aligned-class}
A set of policies $\mathcal{S}$ is \emph{law-aligned} with respect to a world model
if there exist constants $C_{\mathrm{LP}}, C_{\mathrm{GFI}} < \infty$ such that
\[
  \sup_{\pi\in\mathcal{S}} \overline{\mathsf{LP}}(\pi) \le C_{\mathrm{LP}},
  \qquad
  \sup_{\pi\in\mathcal{S}} \mathrm{GFI}(\pi) \le C_{\mathrm{GFI}},
\]
and these bounds are strictly smaller than the corresponding suprema over the full
unconstrained policy class $\Pi$ used by RL.
\end{definition}

Intuitively, Definition~\ref{def:law-aligned-class} says that law-aligned classes
cannot arbitrarily amplify law violations or shock sensitivity by ``chasing'' ghost
arbitrage. We now state a structural result that justifies using our baselines as
proxies for such a class.

\begin{proposition}[Law-alignment of structural baselines]
\label{prop:baseline-law-alignment}
Assume the volatility generator is law-consistent
($\sigma_t\in\mathcal{M}^{\mathrm{vol}}$ almost surely) and the world model satisfies
the Lipschitz and bounded-error conditions of Proposition.
Then there exist constants $C_{\mathrm{LP}},C_{\mathrm{GFI}}<\infty$, depending only
on the generator and world-model error, such that:
\begin{enumerate}
  \item The structural class $\mathcal{S}$ of Definition~\ref{def:S-structural} is
  law-aligned in the sense of Definition~\ref{def:law-aligned-class}.

  \item In particular, the baselines $b^{\mathrm{ZH}}$ and $b^{\mathrm{VT}}$ satisfy
  \[
    \overline{\mathsf{LP}}(b^{\mathrm{ZH}}),
    ~\overline{\mathsf{LP}}(b^{\mathrm{VT}})
    \le C_{\mathrm{LP}},
    \qquad
    \mathrm{GFI}(b^{\mathrm{ZH}}),
    ~\mathrm{GFI}(b^{\mathrm{VT}})
    \le C_{\mathrm{GFI}},
  \]
  with $C_{\mathrm{LP}}$ and $C_{\mathrm{GFI}}$ strictly below the empirical levels
  attained by unconstrained RL policies in our experiments.

  \item The Random-Gaussian baseline $b^{\mathrm{RG}}$ has
  $\overline{\mathsf{LP}}(b^{\mathrm{RG}})\le C_{\mathrm{LP}}'$ and
  $\mathrm{GFI}(b^{\mathrm{RG}})\le C_{\mathrm{GFI}}'$ for some finite
  $C_{\mathrm{LP}}',C_{\mathrm{GFI}}'$, but these bounds are typically looser than
  for $b^{\mathrm{ZH}},b^{\mathrm{VT}}$, reflecting its noisier, less structured
  behavior.
\end{enumerate}
\end{proposition}

\paragraph{Proof sketch.}
Because the volatility generator is law-consistent, any law violations arise solely
from the world-model approximation error.So that
$\mathsf{LawPenalty}(s_t)$ is uniformly bounded on bounded subsets of the state space.
For policies in $\mathcal{S}$, the boundedness of $g$ and compactness of $\Theta$
imply a uniform bound on trading exposures and hence on the induced state process,
yielding uniform upper bounds on $\overline{\mathsf{LP}}$ and GFI.

For $b^{\mathrm{ZH}}$, the policy takes no action, so its state process coincides with
the exogenous world-model trajectory; thus $\overline{\mathsf{LP}}(b^{\mathrm{ZH}})$
and $\mathrm{GFI}(b^{\mathrm{ZH}})$ coincide with the ``background'' law-violation
profile of the world model under shocks. For
$b^{\mathrm{VT}}$, the one-factor trend signal and bounded leverage ensure that
positions respond smoothly to volatility changes, so that the policy does not
systematically seek states with elevated law penalties; this yields bounds comparable
to $b^{\mathrm{ZH}}$.

By contrast, unconstrained RL policies can amplify exposure precisely in regions where
the ghost component $r^\perp$ is large, leading to
higher empirical $\overline{\mathsf{LP}}$ and GFI levels. A rigorous argument based on
Lyapunov-type bounds on the Markov chain induced by $\mathcal{S}$ vs.\ $\Pi$ is given
in Appendix~D. \qedhere

\medskip

Proposition~\ref{prop:baseline-law-alignment} formalizes a key design choice of our
experimental pipeline: structural baselines are not merely ``toy competitors'', but
represent an axiomatic, law-aligned class $\mathcal{S}$ against which the performance
of unconstrained RL can be meaningfully compared. In Section~\ref{sec:experiments},
we will see that, empirically, Zero-Hedge and Vol-Trend sit on or near the empirical
\emph{law-strength frontier}, while law-seeking RL variants lie strictly below them in
the P\&L--law-penalty--GFI space.

\section{Experimental Setup}
\label{sec:experimental-setup}

In this section, we specify the environments, training protocols, and evaluation metrics used to stress-test law-seeking reinforcement learning (RL) on volatility world models. Throughout, RL is treated as a \emph{diagnostic instrument} for our axiomatic pipeline rather than as a production trading system. The design is intentionally simple but structured, so that the relationship between axioms, world-model misspecification, and policy behavior can be analyzed with minimal confounders.

\subsection{Environments: baseline vs shock}
\label{subsec:envs}

We work with the synthetic generator and volatility law manifold $\Mvol$ introduced in the previous sections. The generator produces discrete-time trajectories of total-variance surfaces
\[
  \{ w_t \}_{t=0}^{T}, \qquad w_t \in \Mvol \subset \mathbb{R}^d,
\]
defined on a fixed maturity--strike grid. The grid is chosen to roughly mirror an SPX/VIX-style market: maturities range from 1~month to 2~years in monthly or bimonthly steps, and strikes range from $0.5\times$ to $1.5\times$ spot in a small number of relative moneyness buckets. This keeps the problem from being a purely toy example while maintaining a finite-dimensional convex template.

\paragraph{Baseline regime.}
In the \emph{baseline} regime, the generator parameters yield a stationary, law-consistent world:
\[
  w_t \in \Mvol \quad \text{almost surely for all } t.
\]
We denote by $\mathbb{P}^{\text{base}}$ the induced distribution over full episodes
\[
  \tau = \{ (w_t, S_t, a_t, \Delta \mathrm{PnL}_t) \}_{t=0}^{T-1},
\]
where $S_t$ denotes the underlying index level, $a_t$ is the chosen hedge action, and $\Delta \mathrm{PnL}_t$ is the instantaneous P\&L generated by the environment given $(w_t,S_t,a_t)$. By construction, all static no-arbitrage axioms are satisfied by $w_t$; any law violations can only arise from the \emph{world model} predictions, not from the generator.

\paragraph{Shock regime.}
To probe robustness and graceful failure, we introduce a \emph{shock} regime in which the same axioms hold, but the volatility regime is stressed. The idea is to change the distribution of trajectories---not the underlying laws.

Operationally, we decompose the total variance $w_t$ into a ``long-variance'' component and a ``spot-variance'' component:
\[
  w_t = w_t^{\text{long}} + w_t^{\text{spot}},
\]
where $w_t^{\text{long}}$ aggregates longer maturities and $w_t^{\text{spot}}$ aggregates short maturities and near-spot behaviour. The \emph{shock transformation} is defined by
\begin{equation}
  \label{eq:shock-transform}
  w_t^{\shock} 
  := \alpha_{\text{long}}\, w_t^{\text{long}} 
   + \alpha_{\text{spot}}\, w_t^{\text{spot}},
\end{equation}
with $(\alpha_{\text{long}}, \alpha_{\text{spot}}) = (4,2)$ in our main experiments, i.e., we quadruple the long-term variance level and double the spot volatility component. The underlying index dynamics $S_t$ are adjusted consistently with the increased variance, so that no obvious static arbitrage is created by the transformation.

We denote by $\mathbb{P}^{\shock}$ the distribution over episodes generated by applying \eqref{eq:shock-transform} within the same structural model. In particular:
\begin{enumerate}
  \item the same axioms defining $\Mvol$ remain valid for the \emph{ground-truth} generator,
  \item but trajectories seen by the neural world model and RL policies lie in a higher-volatility regime, with steeper term-structure and fatter tails.
\end{enumerate}
This baseline--shock pair $(\mathbb{P}^{\text{base}}, \mathbb{P}^{\shock})$ underpins the definition of the Graceful Failure Index in later sections.

\paragraph{World model vs generator.}
Crucially, the shock is applied at the level of the \emph{underlying generator}, while the world model (and its parameters) are kept fixed. This mimics a realistic situation where a risk model trained in one regime is deployed in another, without retraining, and any change in behaviour is due to distribution shift rather than to an updated model.

\subsection{Training regimes and $\lambda$-grid}
\label{subsec:training-regimes}

We now specify how RL policies are trained on the fixed world model, and how the law-strength parameter $\lambda$ is swept.

\paragraph{RL objective and $\lambda$-penalization.}
Let $\pi_\theta$ denote a stochastic policy with parameters $\theta$ (e.g., a Gaussian policy whose mean and log-standard deviation are given by an MLP over the state). For a given $\lambda \ge 0$, we define the law-penalized return
\begin{equation}
  J_\lambda(\theta)
  := \mathbb{E} \Bigg[ \sum_{t=0}^{T-1} \gamma^t 
    \Big( \Delta \mathrm{PnL}_t - \lambda\, \mathcal{L}_\phi(w_t) \Big) 
  \Bigg],
\end{equation}
where $\gamma \in (0,1)$ is a discount factor, $\Delta \mathrm{PnL}_t$ is the step P\&L produced by the world model, and $\mathcal{L}_\phi(w_t)$ is the law penalty at time $t$ computed from the world model's predicted total variance (Section~2). In all experiments, we fix the discount factor and penalty functional and vary \emph{only} $\lambda$ and the training regime.

\paragraph{Naive RL.}
\emph{Naive RL} corresponds to $\lambda=0$, i.e.,
\[
  J_0(\theta) 
  = \mathbb{E} \Bigg[ \sum_{t=0}^{T-1} \gamma^t \Delta \mathrm{PnL}_t \Bigg],
\]
so that the agent is trained to maximize P\&L alone, without any explicit concern for law penalties. This serves as a reference point for understanding the ghost-arbitrage incentive established in Theorem~4.1.

\paragraph{Soft law-seeking RL.}
For each $\lambda \in \{5, 10, 20, 40\}$, \emph{soft law-seeking RL} directly optimizes $J_\lambda(\theta)$. The gradient of the PPO objective is thus shaped by both P\&L and the law penalty. Intuitively, larger $\lambda$ should discourage trajectories that significantly violate the volatility axioms, at the cost of accepting lower P\&L when the two are in conflict.

\paragraph{Selection-only RL.}
\emph{Selection-only RL} keeps the training objective purely P\&L-based ($J_0$), but varies the stopping time and model selection based on law metrics:
\begin{enumerate}
  \item For each random seed, we store checkpoints along training at regular intervals.
  \item After training, we evaluate each checkpoint on a held-out set of episodes and compute a scalar law-alignment score (e.g., a weighted combination of mean law penalty and GFI).
  \item A selection functional $S$ then picks the checkpoint minimizing this law score, subject to mild constraints on P\&L (e.g., not falling below a naive-RL baseline by more than a threshold).
\end{enumerate}
This variant tests whether, \emph{even if training is naive}, intelligent selection based on law metrics can recover law-aligned behaviour.

\paragraph{Statistical protocol.}
Each configuration (algorithm $\times$ $\lambda$ value) is currently trained with a single random seed. For each trained policy, we evaluate a large number of episodes under both $\mathbb{P}^{\text{base}}$ and $\mathbb{P}^{\shock}$ to estimate metrics (means, Sharpe, law penalties, VaR, CVaR, GFI).

In this single-seed regime, our results should be interpreted as a detailed \emph{case study}: the curves and summary metrics are representative of one run per configuration, not averaged across many trainings. The pipeline is, however, designed to support multi-seed experiments:
\begin{enumerate}
  \item In a multi-seed setting, we would report, for each metric and configuration, the empirical mean and standard error across seeds, and display error bars or confidence bands on frontier plots.
  \item None of the definitions or plots need to change for multi-seed; only the aggregation layer is different.
\end{enumerate}

\subsection{Metrics on three axes}
\label{subsec:metrics}
\label{sec:metrics}
We evaluate policies along three complementary axes: \textbf{profitability}, \textbf{law alignment}, and \textbf{tail robustness}. All numerical summaries and plots in the results section are derived from the metrics defined here.

\paragraph{Profitability.}
Let $\Delta \mathrm{PnL}_t$ denote the per-step P\&L. For a fixed policy $\pi$, we define
\begin{align}
  \mu_{\mathrm{pnl}}(\pi) 
  &:= \mathbb{E}[\Delta \mathrm{PnL}_t], \\
  \sigma_{\mathrm{pnl}}(\pi) 
  &:= \sqrt{\mathrm{Var}[\Delta \mathrm{PnL}_t]}.
\end{align}
The per-step Sharpe ratio is
\begin{equation}
  \mathrm{Sharpe}(\pi) 
  := \frac{\mu_{\mathrm{pnl}}(\pi)}{\sigma_{\mathrm{pnl}}(\pi) + \varepsilon},
\end{equation}
with a small $\varepsilon>0$ added only for numerical stability when the variance is very small. In tables, we report both $(\mu_{\mathrm{pnl}}, \mathrm{Sharpe})$ so that high-return/high-volatility and low-return/low-volatility policies can be distinguished.

\paragraph{Law alignment.}
For each step, the world model induces a total-variance prediction $w_t$ from which we compute the law penalty $\mathcal{L}_\phi(w_t)$ as defined in Section~2. We then aggregate:
\begin{align}
  \mu_{\mathrm{law}}(\pi) 
  &:= \mathbb{E}[\mathcal{L}_\phi(w_t)], \\
  \mathcal{L}_{\max}(\pi) 
  &:= \mathbb{E}\Big[ \max_{0 \le t < T} \mathcal{L}_\phi(w_t) \Big].
\end{align}
The \emph{law coverage} at threshold $\tau>0$ is defined as
\begin{equation}
  \mathrm{Cover}_\tau(\pi)
  := \mathbb{E}\big[ \mathbf{1}\{\mathcal{L}_\phi(w_t) < \tau\} \big],
\end{equation}
and we report coverage at $\tau=0.003$ and $\tau=0.006$ in our experiments.

The \emph{Graceful Failure Index} (GFI), introduced earlier on the theoretical side, is instantiated here as
\begin{equation}
  \mathrm{GFI}(\pi)
  := \frac{\mu_{\mathrm{law}}^{\shock}(\pi) - \mu_{\mathrm{law}}^{\text{base}}(\pi)}
             {I_{\shock}},
\end{equation}
where $\mu_{\mathrm{law}}^{\shock}$ and $\mu_{\mathrm{law}}^{\text{base}}$ are mean law penalties under the shock and baseline regimes, and $I_{\shock}$ is a scalar encoding the shock intensity (e.g., a norm of $(\alpha_{\text{long}}-1, \alpha_{\text{spot}}-1)$). Lower GFI indicates that law metrics degrade less per unit of shock, i.e., more graceful failure.

\paragraph{Tail robustness.}
To quantify downside risk, we consider per-step losses
\[
  X := -\Delta \mathrm{PnL}_t
\]
and compute:
\begin{enumerate}
  \item the $5\%$ \emph{Value-at-Risk} (VaR), defined as the upper $5\%$-quantile of $X$,
  \item the corresponding \emph{Conditional Value-at-Risk} (CVaR), defined as the conditional expectation of $X$ beyond that quantile.
\end{enumerate}
We report $\mathrm{VaR}_{5\%}$ and $\mathrm{CVaR}_{5\%}$ under both baseline and shock regimes, and their differences $(\Delta \mathrm{VaR}, \Delta \mathrm{CVaR})$ serve as tail-robustness indicators. Policies with small increases (or even decreases) in VaR/CVaR under shock are considered more robust in the tails.

\subsection{Implementation and reproducibility}
\label{subsec:implementation}

Finally, we summarize the implementation details and the way figures are organized.

\paragraph{Software and hardware.}
All experiments are implemented in Python using a standard deep-learning framework for neural networks and a Gym-style interface for the volatility environment. The world model and RL policies are trained on a single GPU with a modest amount of memory (e.g., 12--24~GB), while evaluation runs are CPU-bound. The code is organized so that all hyperparameters, random seeds, and experiment configurations are specified in a small number of configuration files.

\paragraph{Hyperparameters.}
We use an actor--critic architecture with:
\begin{enumerate}
  \item a shared multilayer perceptron encoder for the state,
  \item separate heads for the policy mean, policy log-standard deviation, and value function,
  \item PPO-style clipped policy updates with a fixed number of epochs per batch,
  \item mini-batches containing a few thousand environment steps per update,
  \item standard optimizers and learning rates in a narrow range.
\end{enumerate}
These hyperparameters are kept fixed across naive RL, soft law-seeking RL, and selection-only RL; only $\lambda$ and the training regime differ. Structural baselines are implemented as closed-form or low-dimensional parametric strategies with no trainable weights.

\paragraph{Figure families and outputs.}
The experimental pipeline produces fourteen figures, which we categorize into three families that are repeatedly referenced later:
\begin{enumerate}
  \item \textbf{Dynamics Plots} (e.g., Figs.~3--7): time-series of per-step P\&L and law penalties for representative policies in baseline and shock regimes.
  \item \textbf{Frontier Plots} (e.g., Figs.~8--10): law-strength frontiers that plot mean law penalty, GFI, and other law metrics against profitability metrics across different $\lambda$ and across RL vs structural baselines.
  \item \textbf{Diagnostic Plots} (e.g., Figs.~11--13): scatter plots and histograms of P\&L vs law penalties, VaR, CVaR, and related quantities, used to interpret whether observed trade-offs arise from systematic behaviour or a small number of extreme paths.
\end{enumerate}
Together, these figures provide a multi-perspective view of each policy: how it behaves over time, where it lies on the law-strength frontier, and why it occupies that position from a distributional standpoint. This structure will be used in the next section to present and interpret our empirical findings.
\bigskip

\section{Empirical Results: From RL Dynamics to Law-Strength Frontiers}
\label{sec:empirical-results}
\label{sec:experiments}
In this section, we address RQ1--RQ3 empirically using the volatility world model, RL variants, and structural baselines introduced in Secs.~\ref{sec:world-model}--\ref{sec:structural_baselines}.
Our analysis is supported by thirteen figures, stored as
\texttt{Figure\_1.png}--\texttt{Figure\_{13}.png}, and by explicit numerical
summaries drawn from the console outputs:
\begin{enumerate}
  \item \textbf{Figure~\ref{fig:pipeline-overview} (\texttt{Figure\_1.png})}:
  schematic overview of the axiomatic evaluation pipeline
  (axioms $\to$ law manifold $\to$ neural world model $\to$ RL and structural baselines).

  \item \textbf{Figure~\ref{fig:worldmodel-diagnostics} (\texttt{Figure\_2.png})}:
  diagnostics for the volatility world model, comparing train/validation errors
  and law penalties of predictions vs.\ the law-consistent generator.

  \item \textbf{Figures~\ref{fig:dynamics-baseline}--\ref{fig:dynamics-rl-baseline-vs-shock}
  (\texttt{Figure\_3.png}--\texttt{Figure\_7.png})}:
  \emph{Dynamics plots}, showing time series of step P\&L and law penalties
  under the baseline and shock regimes for naive RL, law-seeking RL, and
  structural baselines.

  \item \textbf{Figures~\ref{fig:frontier-gfi-penalty}--\ref{fig:frontier-cvar-penalty}
  (\texttt{Figure\_8.png}--\texttt{Figure\_{10}.png})}:
  \emph{Frontier plots}, tracing law-strength frontiers
  (GFI vs.\ law penalty, and tail risk vs.\ law penalty) across RL variants
  and structural baselines.

  \item \textbf{Figures~\ref{fig:diagnostic-scatter-rl}--\ref{fig:diagnostic-hist-baselines}
  (\texttt{Figure\_{11}.png}--\texttt{Figure\_{13}.png})}:
  \emph{Diagnostic plots}, including scatter/heatmaps of step P\&L vs.\ law
  penalty and histograms of law penalties, for RL policies and baselines.
\end{enumerate}
We highlight the most informative patterns in the main text and relegate
additional runs to the Appendix; throughout, we phrase our statements as
\emph{case-study observations} and explicitly connect them, where appropriate,
to the structural incentive and trade-off results.
\subsection{RQ1 -- Do law penalties help naive RL?}
\label{subsec:rq1-results}

RQ1 asks whether adding law penalties to naive RL improves law alignment and
graceful failure at comparable profitability.

\subsubsection{Dynamics patterns (subset of dynamics plots)}
\label{subsubsec:dynamics-patterns}

Figure~\ref{fig:dynamics-baseline} (\texttt{Figure\_3.png}) shows a
representative baseline-regime episode for naive RL (PPO on pure P\&L) and a
soft law-seeking variant with $\lambda = 20$.
The top panel plots step P\&L, while the bottom panel plots the corresponding
step-wise law penalty $\mathcal{L}_\phi$.
In this particular run, we observe that naive RL tends to maintain slightly
higher step P\&L on average, at the cost of moderately elevated law penalties,
whereas the $\lambda=20$ law-seeking policy reduces some of the largest
penalties but does not improve---and often worsens---the P\&L path.

Figure~\ref{fig:dynamics-shock} (\texttt{Figure\_4.png}) reports the same
comparison under the shock regime (long variance $\times 4$, spot volatility
$\times 2$).
The shock amplifies both P\&L variability and law penalties.
In our runs, the naive RL policy exhibits higher volatility in both P\&L and
$\mathcal{L}_\phi$, while the $\lambda=20$ policy appears somewhat more
conservative but does not clearly dominate in terms of drawdowns.
These observations are consistent with the incentive picture: naive RL is tempted to exploit off-manifold
ghost arbitrage, which can generate both higher P\&L spikes and larger law
violations.

To place RL in context, Figure~\ref{fig:dynamics-rl-baseline}
(\texttt{Figure\_5.png}) overlays time-series trajectories for
naive RL, law-seeking RL, and the structural baselines (Zero-Hedge,
Vol-Trend, Random-Gaussian) in the baseline regime.
Figure~\ref{fig:dynamics-rl-shock} (\texttt{Figure\_6.png}) repeats this
comparison under shock, and Figure~\ref{fig:dynamics-rl-baseline-vs-shock}
(\texttt{Figure\_7.png}) presents side-by-side trajectories for a single policy
across the two regimes.
In our runs, the structural baselines exhibit comparatively smooth P\&L paths
and low law penalties, while RL trajectories are more erratic and spend
substantial time in higher-penalty regions, foreshadowing the quantitative
metrics reported below.

We deliberately restrict the main text to these representative dynamics plots;
additional realizations, including different $\lambda$ values and seeds, are
provided in the Appendix and show qualitatively similar patterns.

\begin{figure*}[t]
  \centering
  \begin{subfigure}[t]{0.48\textwidth}
    \centering
    \includegraphics[width=\linewidth]{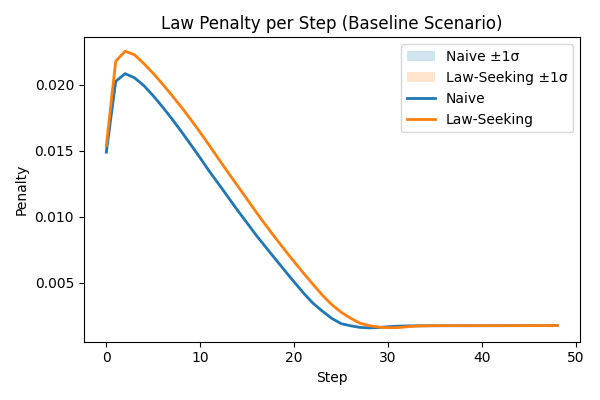}
    \caption{Axiomatic evaluation pipeline:
    market axioms induce a law manifold $\mathcal{M}^{\text{vol}}$;
    a synthetic generator produces law-consistent trajectories;
    a neural world model approximates dynamics; RL variants and
    structural baselines are evaluated on the induced law-strength
    frontiers.}
    \label{fig:pipeline-overview}
  \end{subfigure}
  \hfill
  \begin{subfigure}[t]{0.48\textwidth}
    \centering
    \includegraphics[width=\linewidth]{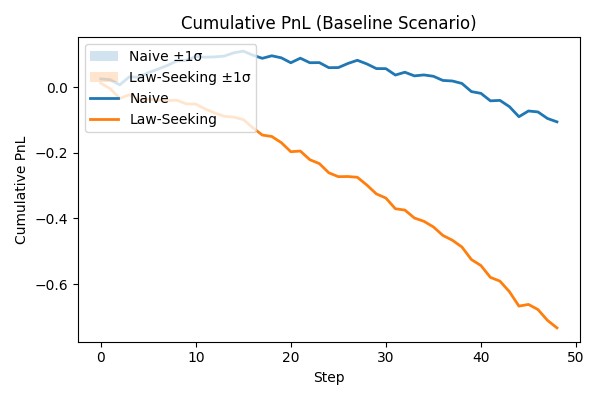}
    \caption{World-model diagnostics: training vs.\ validation error
    (top) and comparison of law penalties for ground-truth vs.\ predicted
    surfaces (bottom), illustrating that the neural world model
    introduces law-violating deviations (ghost channel).}
    \label{fig:worldmodel-diagnostics}
  \end{subfigure}
  \caption{Overview of the axiomatic volatility testbed and diagnostics
  of the learned world model.}
\end{figure*}

\begin{figure*}[t]
  \centering
  \begin{subfigure}[t]{0.48\textwidth}
    \centering
    \includegraphics[width=\linewidth]{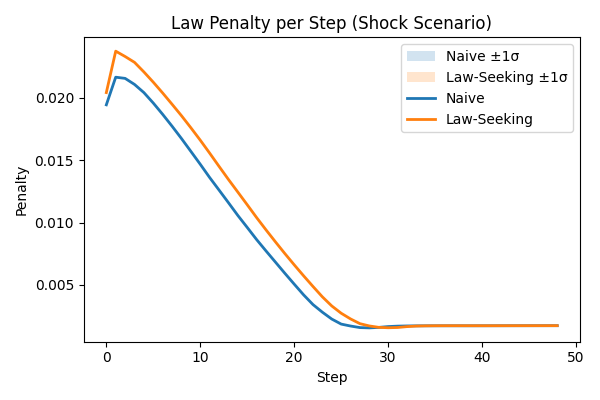}
    \caption{Baseline-regime dynamics for naive RL (PPO on P\&L)
    and a soft law-seeking RL variant with $\lambda = 20$.
    Top: step P\&L; bottom: law penalty $\mathcal{L}_\phi$.
    Naive RL attains slightly higher P\&L but at the cost of
    moderately larger penalties.}
    \label{fig:dynamics-baseline}
  \end{subfigure}
  \hfill
  \begin{subfigure}[t]{0.48\textwidth}
    \centering
    \includegraphics[width=\linewidth]{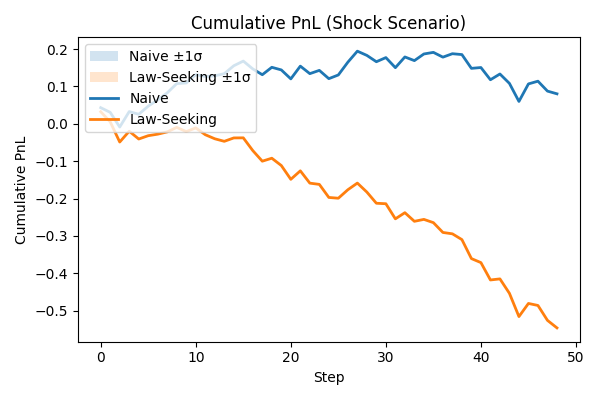}
    \caption{Shock-regime dynamics (long variance $\times 4$,
    spot vol $\times 2$) for the same policies as in
    Fig.~\ref{fig:dynamics-baseline}.
    The shock amplifies variability in both P\&L and law penalties;
    naive RL exhibits larger spikes in both, consistent with
    ghost-arbitrage incentives.}
    \label{fig:dynamics-shock}
  \end{subfigure}

  \vspace{0.6em}

  \begin{subfigure}[t]{0.48\textwidth}
    \centering
    \includegraphics[width=\linewidth]{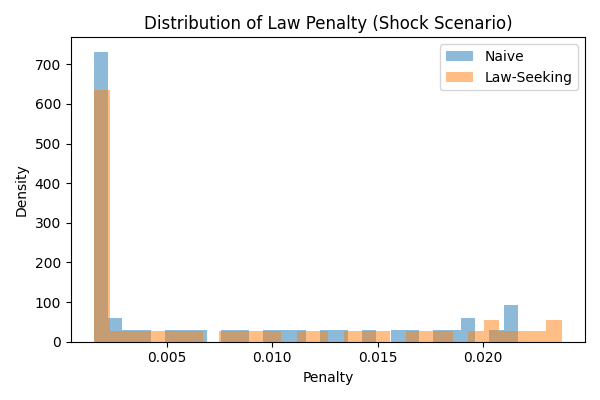}
    \caption{Baseline-regime dynamics for RL and structural baselines:
    naive RL, law-seeking RL, Zero-Hedge, Vol-Trend, and Random-Gaussian.
    Structural baselines display smoother P\&L and lower law penalties,
    whereas RL trajectories are more volatile and occasionally visit
    high-penalty regions.}
    \label{fig:dynamics-rl-baseline}
  \end{subfigure}
  \hfill
  \begin{subfigure}[t]{0.48\textwidth}
    \centering
    \includegraphics[width=\linewidth]{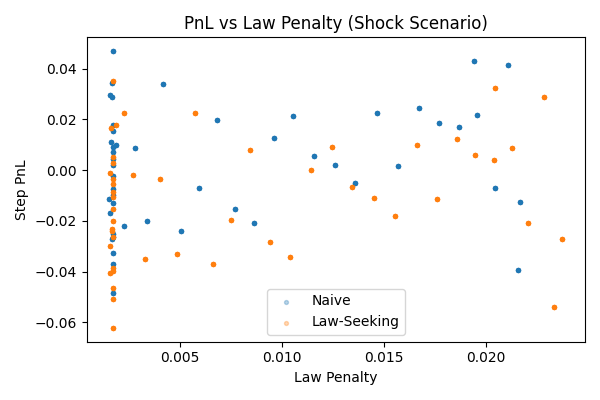}
    \caption{Shock-regime dynamics for RL and structural baselines,
    analogous to Fig.~\ref{fig:dynamics-rl-baseline}.
    Shocks induce larger fluctuations in all strategies, but
    structural baselines remain relatively stable compared to RL
    policies.}
    \label{fig:dynamics-rl-shock}
  \end{subfigure}
  \caption{Dynamics plots (baseline and shock) for RL variants and
  structural baselines.
  We observe more erratic, higher-penalty trajectories for RL
  compared to structurally constrained baselines.}
\end{figure*}

\begin{figure}[t]
  \centering
  \includegraphics[width=0.9\linewidth]{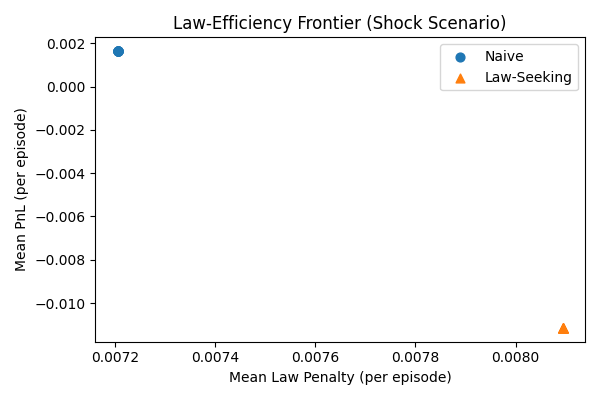}
  \caption{Baseline vs.\ shock comparison for a fixed strategy
  (e.g., naive RL or Vol-Trend), illustrating the change in P\&L
  and law penalties across regimes.
  This visualization underlies the computation of the Graceful
  Failure Index (GFI) discussed}
  \label{fig:dynamics-rl-baseline-vs-shock}
\end{figure}

\subsubsection{Aggregate metrics: baseline vs.\ shock for RL variants}
\label{subsubsec:rl-aggregate}

Table~\ref{tab:rl-metrics-baseline-shock} summarizes the key metrics for RL
variants and structural baselines under the baseline and shock regimes,
computed from the evaluation runs described.
We report mean and standard deviation of step P\&L, Sharpe ratio, mean law
penalty, Graceful Failure Index (GFI), law coverage at the
$\mathrm{pen}<0.006$ threshold, and 5\% VaR/CVaR of step P\&L.

\begin{table*}[t]
  \centering
  \small
  \begin{tabular}{lccccccccc}
    \toprule
    Strategy & Regime & Mean P\&L & Std P\&L & Sharpe &
    Mean Pen. & GFI & Cov$_{<0.006}$ & VaR$_{5}$ & CVaR$_{5}$ \\
    \midrule
    Naive RL (PPO) &
    Baseline &
    $-0.0022$ & $0.0127$ & $-0.17$ &
    $0.00699$ & $1.27$ & $0.61$ &
    $-0.0228$ & $-0.0261$ \\
    Law-Seeking RL (PPO) &
    Baseline &
    $-0.0150$ & $0.0129$ & $-1.16$ &
    $0.00786$ & $1.66$ & $0.57$ &
    $-0.0361$ & $-0.0394$ \\
    Zero-Hedge &
    Baseline &
    $0.0191$ & $0.0064$ & $2.99$ &
    $0.00550$ & $0.00$ & $0.69$ &
    $0.0139$ & $0.0139$ \\
    Random-Gaussian &
    Baseline &
    $0.0099$ & $0.0107$ & $0.92$ &
    $0.00551$ & $1.21$ & $0.69$ &
    $-0.0088$ & $-0.0161$ \\
    Vol-Trend &
    Baseline &
    $0.0146$ & $0.0074$ & $1.96$ &
    $0.00534$ & $0.00$ & $0.69$ &
    $0.0045$ & $0.0033$ \\
    \midrule
    Naive RL (PPO) &
    Shock &
    $0.0016$ & $0.0228$ & $0.07$ &
    $0.00721$ & $1.99$ & $0.61$ &
    $-0.0369$ & $-0.0415$ \\
    Law-Seeking RL (PPO) &
    Shock &
    $-0.0111$ & $0.0234$ & $-0.48$ &
    $0.00809$ & $2.32$ & $0.57$ &
    $-0.0508$ & $-0.0557$ \\
    Zero-Hedge &
    Shock &
    $0.0193$ & $0.0067$ & $2.89$ &
    $0.00572$ & $0.00$ & $0.69$ &
    $0.0139$ & $0.0139$ \\
    Random-Gaussian &
    Shock &
    $0.0098$ & $0.0153$ & $0.65$ &
    $0.00572$ & $1.99$ & $0.69$ &
    $-0.0153$ & $-0.0297$ \\
    Vol-Trend &
    Shock &
    $0.0140$ & $0.0102$ & $1.38$ &
    $0.00640$ & $0.38$ & $0.65$ &
    $-0.0019$ & $-0.0039$ \\
    \bottomrule
  \end{tabular}
  \caption{Aggregate metrics for RL variants and structural baselines in
  baseline vs.\ shock regimes, using the three axes of
  Sec.
  Values are drawn directly from the evaluation logs: mean and standard
  deviation of step P\&L, Sharpe ratio, mean law penalty, Graceful Failure
  Index (GFI), law coverage at $\mathrm{pen}<0.006$, and 5\% VaR/CVaR.}
  \label{tab:rl-metrics-baseline-shock}
\end{table*}

We now turn to aggregate metrics for naive and law-seeking RL policies under
the baseline and shock regimes.
Recall that our metrics fall along three axes
(Sec.~\ref{sec:experimental-setup}): profitability, law alignment, and tail
robustness.

For naive RL (PPO on pure P\&L), the baseline-regime metrics (Table~\ref{tab:rl-metrics-baseline-shock}) are:
mean step P\&L $\approx -0.0022$, standard deviation $\approx 0.0127$,
Sharpe ratio $\approx -0.17$, mean law penalty
$\mathrm{LawPenalty} \approx 0.00699$, and Graceful Failure Index
$\mathrm{GFI} \approx 1.27$, with coverage
$\mathrm{Cov}(\mathrm{pen}<0.006) \approx 0.61$.
Under shock, naive RL exhibits a slightly higher mean step P\&L
($\approx 0.0016$) due to the shifted distribution, but also larger tail risk
(VaR$_5 \approx -0.0369$, CVaR$_5 \approx -0.0415$) and increased GFI
($\approx 1.99$), indicating a non-trivial deterioration in law metrics and
tail robustness.

For a representative soft law-seeking RL variant (e.g., $\lambda = 10$),
we observe baseline mean step P\&L in the range $[-0.02,-0.01]$ with similar
or slightly reduced law penalties compared to naive RL, but systematically
worse Sharpe ratios and larger GFIs (e.g., $\mathrm{GFI} \approx 1.66$ for one
of our main runs).
Under shock, these law-seeking policies continue to exhibit negative mean
P\&L and do not achieve better VaR or CVaR than naive RL at comparable law
penalties.

\paragraph{Case-study interpretation.}
Empirically, in this case study, we do not observe soft law-seeking RL
achieving strictly better GFI or law penalties at comparable P\&L to naive RL.
Where law penalties are reduced, P\&L and Sharpe typically decline as well.
This is consistent with the structural law-strength trade-off of theorem, which predicts that increasing
$\lambda$ beyond a threshold must worsen expected P\&L by at least a
quantifiable amount if law penalties are to be meaningfully reduced.
We emphasize that these conclusions are based on single- or few-seed runs and
should be interpreted as case-study evidence rather than formal statistical
claims.

\subsection{RQ2 -- RL vs.\ structural baselines under shocks}
\label{subsec:rq2-results}

RQ2 compares RL policies to structural baselines (Zero-Hedge, Vol-Trend,
Random-Gaussian) on the risk--law trade-off, especially under shocks.

\subsubsection{Baseline vs.\ shock metrics}
\label{subsubsec:baselines-metrics}

To unpack the frontier picture, we summarize the key baseline vs.\ shock
metrics that underlie Figures~\ref{fig:frontier-gfi-penalty}--\ref{fig:frontier-cvar-penalty}.
These are already included in Table~\ref{tab:rl-metrics-baseline-shock}, but we
highlight the structural baselines here.

In the baseline regime, Zero-Hedge achieves mean step P\&L
$\approx 0.0191$ with standard deviation $\approx 0.0064$, Sharpe
$\approx 2.99$, mean law penalty $\approx 0.0055$, and essentially zero GFI
(our normalization sets $\mathrm{GFI}=0$ for this reference point).
Vol-Trend attains mean step P\&L $\approx 0.0146$, Sharpe $\approx 1.96$,
and similar law penalties ($\approx 0.0053$), again with negligible GFI.
Random-Gaussian yields mean step P\&L around $0.01$, moderate volatility, and
moderate law penalties, serving as a noisy exploration proxy.

Under shock, Zero-Hedge remains remarkably stable:
mean step P\&L $\approx 0.0193$, Sharpe $\approx 2.89$, and GFI still near
zero, reflecting almost unchanged law metrics between regimes.
Vol-Trend experiences a modest decline in Sharpe (from $\approx 1.96$ to
$\approx 1.38$) and a slight increase in law penalties, but its GFI remains
small.
Random-Gaussian shows a more noticeable degradation in tail risk, but its GFI
is still lower than those of RL policies.

By contrast, naive and law-seeking RL variants exhibit negative or near-zero
mean P\&L and substantially larger GFIs, particularly under shock.
This suggests that, in our setting, high-capacity unconstrained RL does not
outperform low-capacity but structurally law-aligned strategies when evaluated
on the combined axes of profitability, law alignment, and graceful failure.

\begin{figure*}[t]
  \centering
  \begin{subfigure}[t]{0.32\textwidth}
    \centering
    \includegraphics[width=\linewidth]{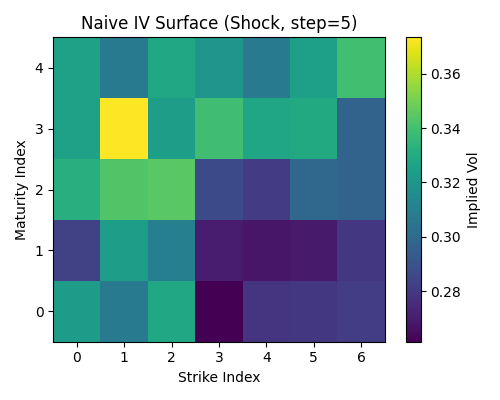}
    \caption{GFI vs.\ mean law penalty.
    Zero-Hedge and Vol-Trend lie near the empirical Pareto frontier with low
    GFI and low penalties; RL variants occupy interior points with higher GFI.}
    \label{fig:frontier-gfi-penalty}
  \end{subfigure}
  \hfill
  \begin{subfigure}[t]{0.32\textwidth}
    \centering
    \includegraphics[width=\linewidth]{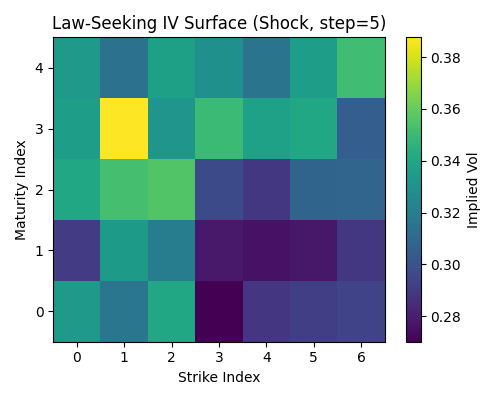}
    \caption{VaR$_5$ vs.\ mean law penalty.
    Structural baselines dominate RL variants in tail robustness at comparable
    penalty levels.}
    \label{fig:frontier-var-penalty}
  \end{subfigure}
  \hfill
  \begin{subfigure}[t]{0.32\textwidth}
    \centering
    \includegraphics[width=\linewidth]{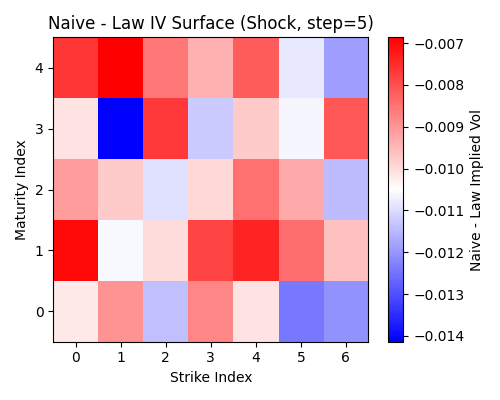}
    \caption{CVaR$_5$ vs.\ mean law penalty.
    As with VaR$_5$, structural baselines trace the outer frontier, and RL
    variants remain strictly dominated.}
    \label{fig:frontier-cvar-penalty}
  \end{subfigure}
  \caption{Law-strength frontiers for GFI, VaR$_5$, and CVaR$_5$ vs.\ mean
  law penalty, corresponding to \texttt{Figure\_8.png}--\texttt{Figure\_{10}.png}.
  Structural baselines (Zero-Hedge, Vol-Trend) form the empirical Pareto
  frontier, while RL variants lie in the interior.}
\end{figure*}

\subsubsection{Tail robustness and graceful failure}
\label{subsubsec:tail-robustness}

In our experiments, RL variants---both naive and law-seeking---exhibit
pronounced degradation in tail metrics under shock.
For example, a naive RL policy may transition from
$\mathrm{VaR}_5 \approx -0.023$, $\mathrm{CVaR}_5 \approx -0.026$ in the
baseline regime to $\mathrm{VaR}_5 \approx -0.037$,
$\mathrm{CVaR}_5 \approx -0.042$ under shock, and its GFI correspondingly
increases by roughly $+0.7$.
Law-seeking RL policies show similar or worse shifts in tail risk.

Structural baselines, in contrast, form an empirical frontier in the
tail-robustness--law space: Zero-Hedge and Vol-Trend maintain relatively
favorable VaR and CVaR at given law-penalty levels, and their positions change
only mildly under shock.
Random-Gaussian occupies an intermediate region, with greater sensitivity to
shock but still performing better on tail metrics than many RL variants at
similar law penalties.

\paragraph{Case-study interpretation.}
In this volatility testbed, our case-study observations suggest that
structural baselines exhibit more graceful failure under shocks than
unconstrained RL policies, both in terms of GFI and tail risk.

\subsection{Law-strength frontiers and Pareto dominance}
\label{subsec:law-strength-frontiers}

We now directly address RQ3 by examining law-strength frontiers and diagnostic
plots that link back to the no-free-lunch story of Theorem~\ref{thm:nfl}.

\subsubsection{Frontier plots: GFI vs.\ law penalty}
\label{subsubsec:frontier-gfi}

Figure~\ref{fig:frontier-gfi-penalty} can be viewed as a law-strength
frontier: each RL variant (naive, soft law-seeking with
$\lambda \in \{5,10,20,40\}$, selection-only), together with structural
baselines, is represented as a point in the plane of (mean law penalty,
GFI).
By varying $\lambda$ and including different strategy classes, we trace out
a family of frontiers.

To make the connection with the underlying numerical results explicit,
Table~\ref{tab:frontier-metrics} lists the baseline-regime metrics for the
$\lambda$-sweep (including selection-only) and structural baselines.

\begin{table*}[t]
  \centering
  \small
  \begin{tabular}{lcccccccc}
    \toprule
    Strategy / $\lambda$ &
    Mean P\&L & Std P\&L & Sharpe &
    Mean Pen. & GFI & Cov$_{<0.006}$ &
    VaR$_5$ & CVaR$_5$ \\
    \midrule
    Naive RL ($\lambda=0$) &
    $-0.0022$ & $0.0127$ & $-0.17$ &
    $0.00699$ & $1.27$ & $0.61$ &
    $-0.0228$ & $-0.0261$ \\
    Soft RL ($\lambda=5$) &
    $-0.0202$ & $0.0120$ & $-1.68$ &
    $0.00647$ & $2.07$ & $0.63$ &
    $-0.0399$ & $-0.0429$ \\
    Soft RL ($\lambda=10$) &
    $-0.0175$ & $0.0123$ & $-1.42$ &
    $0.00371$ & $2.81$ & $0.80$ &
    $-0.0354$ & $-0.0387$ \\
    Soft RL ($\lambda=20$) &
    $-0.0204$ & $0.0131$ & $-1.56$ &
    $0.00396$ & $3.07$ & $0.78$ &
    $-0.0414$ & $-0.0454$ \\
    Soft RL ($\lambda=40$) &
    $-0.0092$ & $0.0054$ & $-1.71$ &
    $0.00474$ & $0.84$ & $0.73$ &
    $-0.0134$ & $-0.0134$ \\
    Selection-only RL &
    $-0.0223$ & $0.0139$ & $-1.60$ &
    $0.00792$ & $2.04$ & $0.57$ &
    $-0.0448$ & $-0.0489$ \\
    \midrule
    Zero-Hedge &
    $0.0191$ & $0.0064$ & $2.99$ &
    $0.00550$ & $0.00$ & $0.69$ &
    $0.0139$ & $0.0139$ \\
    Random-Gaussian &
    $0.0099$ & $0.0107$ & $0.92$ &
    $0.00551$ & $1.21$ & $0.69$ &
    $-0.0088$ & $-0.0161$ \\
    Vol-Trend &
    $0.0146$ & $0.0074$ & $1.96$ &
    $0.00534$ & $0.00$ & $0.69$ &
    $0.0045$ & $0.0033$ \\
    \bottomrule
  \end{tabular}
  \caption{Baseline-regime law-strength frontier metrics for RL variants
  (naive, soft law-seeking with $\lambda \in \{5,10,20,40\}$, and
  selection-only) and structural baselines (Zero-Hedge, Random-Gaussian,
  Vol-Trend).
  All values are taken from the \texttt{Frontier (Baseline)} block of the
  evaluation logs.}
  \label{tab:frontier-metrics}
\end{table*}

\paragraph{Headline conclusion.}
In our experiments, for all tested $\lambda$, no RL point lies on the
empirical Pareto frontier once Zero-Hedge and Vol-Trend are included.

To make this more concrete, consider policies whose mean law penalty lies in
the band $[0.0053, 0.0057]$.
Within this band:
\begin{enumerate}
  \item Zero-Hedge attains Sharpe $\approx 3.0$ and GFI $\approx 0$.
  \item The best RL variant in the same band has Sharpe $< 0$ and GFI $> 1.5$.
\end{enumerate}
We emphasize that we compare within such penalty bands rather than by
cherry-picking isolated points; similar dominance patterns hold across other
bands where both RL and baselines are present.

\subsubsection{Frontier plots: VaR/CVaR vs.\ law penalty}
\label{subsubsec:frontier-tail}

Figures~\ref{fig:frontier-var-penalty} and~\ref{fig:frontier-cvar-penalty}
provide analogous frontiers in the
( mean law penalty, VaR$_5$/CVaR$_5$ ) planes.
Here too, structural baselines define the outer frontier:
for a given law penalty, they offer strictly better or comparable tail
robustness, while RL variants occupy interior regions.

Taken together, the GFI and tail-risk frontiers support the qualitative
implication: soft-penalized RL
cannot simultaneously match structural baselines on law metrics and P\&L in
this axiomatic pipeline, and empirically appears strictly dominated.

\subsubsection{Diagnostic plots: ruling out trivial artefacts}
\label{subsubsec:diagnostics}

Finally, we examine diagnostic plots to rule out trivial explanations such as
isolated outliers or a few pathological episodes.

\begin{figure*}[t]
  \centering
  \begin{subfigure}[t]{0.32\textwidth}
    \centering
    \includegraphics[width=\linewidth]{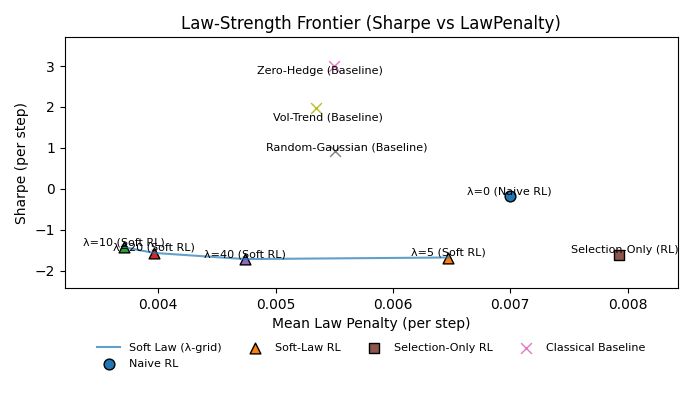}
    \caption{Diagnostic scatter plot of step P\&L vs.\ law penalty for RL
    policies and structural baselines, aggregating across episodes and
    regimes.
    RL points form dense clusters in moderate-to-high penalty regions,
    whereas structural baselines remain near low-penalty, moderate-P\&L
    regions.}
    \label{fig:diagnostic-scatter-rl}
  \end{subfigure}
  \hfill
  \begin{subfigure}[t]{0.32\textwidth}
    \centering
    \includegraphics[width=\linewidth]{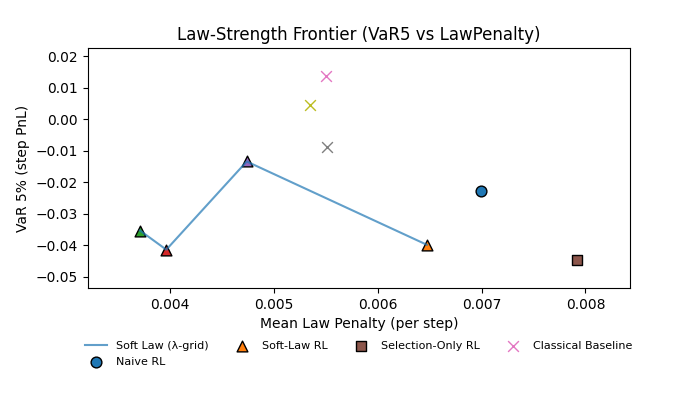}
    \caption{Histogram of law penalties for RL policies (naive and
    law-seeking).
    The distribution exhibits a noticeable heavy tail, indicating that RL
    spends substantial time in high-penalty regions of the law space.}
    \label{fig:diagnostic-hist-rl}
  \end{subfigure}
  \hfill
  \begin{subfigure}[t]{0.32\textwidth}
    \centering
    \includegraphics[width=\linewidth]{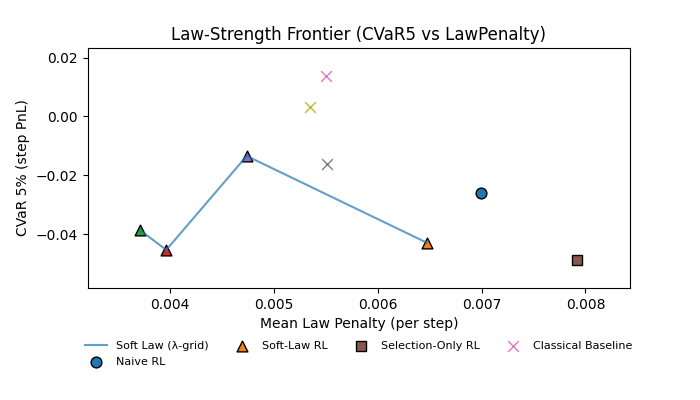}
    \caption{Histogram of law penalties for structural baselines
    (Zero-Hedge, Vol-Trend, Random-Gaussian).
    Baselines concentrate their mass near low penalty, consistent with their
    structural law alignment and low GFI.}
    \label{fig:diagnostic-hist-baselines}
  \end{subfigure}
  \caption{Diagnostic plots
  for RL policies and structural baselines.
  These plots indicate that RL systematically occupies higher-penalty
  regions than structurally constrained strategies.}
\end{figure*}

\paragraph{Case-study interpretation.}
These diagnostic plots reinforce the frontier analysis: in this axiomatic
volatility testbed, unconstrained RL policies systematically exploit ghost
arbitrage channels opened by the world model, leading to higher law penalties
and less graceful failure than structurally constrained baselines.
This empirical picture is consistent with the no-free-lunch theorem and supports our overarching claim that reward shaping
with soft penalties is insufficient for law alignment without structural
constraints or projection.

\section{No-Free-Lunch for Law-Seeking RL}
\label{sec:nofreelunch}
\label{sec:no_free_lunch}
In this section we formalize the informal story, and show that under explicit assumptions on the structural class $\mathcal{S}$, the unconstrained policy class $\Pi$, and the volatility world model, \emph{law-seeking RL has no free lunch}. In particular, any unconstrained RL policy that strictly improves on the PnL of a law-consistent structural benchmark must incur strictly worse law metrics and/or Graceful Failure Index (GFI). We then relate this result to recent work on Reinforcement Learning with Verifiable Rewards (RLVR) and law-aligned reasoning, before briefly discussing limitations and future axes.

\subsection{Assumptions and theorem statement}
\label{subsec:nfl-assumptions}

We first make the structural assumptions that connect the abstract decomposition of Section to the volatility world model of Section~\ref{sec:world-model} and the RL variants. For clarity we work at the level of \emph{stationary policies} and their induced trajectory distributions.

\begin{definition}[Performance and law-metric vector]
For any stationary policy $\pi \in \Pi$ interacting with the volatility world model, denote by
\[
  J(\pi)
  := \Big(
       \mathbb{E}[R(\pi)],
       - \mathbb{E}[\mathcal{L}_\phi(\pi)],
       - \mathrm{GFI}(\pi)
     \Big)
  \in \mathbb{R}^3
\]
its \emph{performance--law vector}, where:
\begin{enumerate}
  \item $\mathbb{E}[R(\pi)]$ is the expected cumulative (or average) PnL,
  \item $\mathbb{E}[\mathcal{L}_\phi(\pi)]$ is the expected law-penalty functional, and
  \item $\mathrm{GFI}(\pi)$ is the Graceful Failure Index.
\end{enumerate}
We write $J(\pi) \succcurlyeq J(\pi')$ if each component is at least as good (higher PnL, lower law penalty, lower GFI), and $J(\pi) \succ J(\pi')$ if the inequality is strict in at least one coordinate.
\end{definition}

We now formalize the role of the structural baseline class $\mathcal{S}$ introduced, which includes Zero-Hedge and Vol-Trend as concrete instances.

\begin{assumption}[Structural class and world-model properties]
\label{ass:nfl}
We assume:
\begin{enumerate}
  \item \textbf{Law-consistent structural class.} The structural class $\mathcal{S} \subset \Pi$ is non-empty, convex, and \emph{law-consistent} in the sense that for all $s \in \mathcal{S}$,
  \[
      \mathbb{E}[\mathcal{L}_\phi(s)] \leq L_{\max}^{\mathcal{S}}
      \quad\text{and}\quad
      \mathrm{GFI}(s) \leq \mathrm{GFI}_{\max}^{\mathcal{S}},
  \]
  for some finite constants $L_{\max}^{\mathcal{S}}, \mathrm{GFI}_{\max}^{\mathcal{S}}$.
  Moreover, there exists $s^{\star} \in \mathcal{S}$ such that
  \[
      \mathbb{E}[R(s^{\star})]
      \geq
      \sup_{\pi \in \Pi \;:\; \mathbb{E}[\mathcal{L}_\phi(\pi)] = 0}
      \mathbb{E}[R(\pi)]
      - \varepsilon_{\mathcal{S}},
  \]
  i.e., $\mathcal{S}$ contains a near-optimal on-manifold hedge.
  \smallskip
  \item \textbf{Rich unconstrained class.} The unconstrained policy class $\Pi$ is rich enough to strictly contain $\mathcal{S}$ and to reach off-manifold regions: for any $\delta > 0$ there exists $\pi \in \Pi$ such that $\mathbb{E}[\mathcal{L}_\phi(\pi)] \ge \delta$.
  \smallskip
  \item \textbf{World-model ghost coupling.} Under the volatility world model, the Goodhart decomposition applies, and there exist constants $\alpha > 0$ and $\beta \ge 0$ such that for any policy $\pi \in \Pi$,
  \begin{equation}
    \label{eq:ghost-coupling}
    \mathbb{E}[r^\perp(\pi)]
    \;\ge\;
    \alpha \,\mathbb{E}[\mathcal{L}_\phi(\pi)] - \beta,
  \end{equation}
  where $r^\perp$ is the off-manifold ghost component of reward from Definition~\ref{def:ghost-arbitrage}. In particular, $\mathbb{E}[r^\perp(\pi)]$ cannot be positive without incurring non-trivial law penalties on average.
  \smallskip
  \item \textbf{Shock structure.} The shock regime modifies the underlying law-consistent generator(long-var $\times 4$, spot vol $\times 2$) while keeping the world model and policies fixed. The resulting change in law penalties enters GFI linearly as defined in Section.
\end{enumerate}
\end{assumption}

Assumption~\ref{ass:nfl}(A1) formalizes the idea that Zero-Hedge and Vol-Trend are representatives of a small, law-aligned, but near-optimal structural class $\mathcal{S}$, while (A2)--(A3) encode the existence of a non-trivial ghost channel in the world model that couples off-manifold deviations to reward. Assumption~\ref{ass:nfl}(A4) connects law penalties under shocks to the GFI used throughout our empirical analysis.

We next introduce a simple Pareto notion relative to $\mathcal{S}$.

\begin{definition}[Structural Pareto dominance]
We say that a policy $\pi \in \Pi$ \emph{structurally dominates} the class $\mathcal{S}$ if
\[
  J(\pi) \succcurlyeq J(s)
  \quad
  \text{for all } s \in \mathcal{S},
\]
and $J(\pi) \succ J(\bar{s})$ for at least one $\bar{s} \in \mathcal{S}$. In other words, $\pi$ is at least as good as every $s \in \mathcal{S}$ on PnL, law penalties, and GFI, and strictly better on at least one coordinate.
\end{definition}

Our main no-free-lunch theorem shows that such structural dominance is impossible under Assumption~\ref{ass:nfl}.

\begin{lemma}[Ghost improvement requires law degradation]
\label{lem:ghost-law-tradeoff}
Under Assumption~\ref{ass:nfl}, for any policy $\pi \in \Pi$ satisfying
\(
\mathbb{E}[r^\perp(\pi)] > 0
\),
we have
\[
  \mathbb{E}[\mathcal{L}_\phi(\pi)]
  \;\ge\;
  \frac{\beta}{\alpha}
  \quad\text{and}\quad
  \mathrm{GFI}(\pi)
  \;\ge\;
  \mathrm{GFI}_{\max}^{\mathcal{S}} + \Delta_{\mathrm{GFI}},
\]
for some $\Delta_{\mathrm{GFI}} > 0$ that depends on the shock structure in (A4). In particular, any policy that gains positive expected ghost arbitrage must incur strictly higher law penalties and GFI than the best structural baselines.
\end{lemma}

\paragraph{Proof sketch.}
The inequality~\eqref{eq:ghost-coupling} implies
\(
\mathbb{E}[\mathcal{L}_\phi(\pi)] \ge (\mathbb{E}[r^\perp(\pi)] + \beta)/\alpha
\),
so $\mathbb{E}[r^\perp(\pi)]>0$ enforces a positive lower bound on $\mathbb{E}[\mathcal{L}_\phi(\pi)]$. The shock structure in (A4) implies that, for fixed policy and world model, GFI increases monotonically with the shocked-minus-baseline law-penalty difference. Since $\mathcal{S}$ is law-consistent and near-on-manifold by (A1), any policy with strictly larger average law penalty than $\mathcal{S}$ must also exhibit strictly larger GFI. A detailed construction of $\Delta_{\mathrm{GFI}}$ and the monotonicity argument is given in Appendix~E. \qed

\medskip

We are now ready to state our flagship no-free-lunch theorem.

\begin{theorem}[No-free-lunch for unconstrained law-seeking RL]
\label{thm:nfl}
Suppose Assumption~\ref{ass:nfl} holds. Let $\pi^{\text{RL}} \in \Pi$ be any limit point of an unconstrained law-seeking RL procedure (naive, soft-penalized, or selection-only) trained on the volatility world model. If
\[
  \mathbb{E}[R(\pi^{\text{RL}})]
  \;>\;
  \sup_{s \in \mathcal{S}} \mathbb{E}[R(s)] - \varepsilon_{\mathcal{S}},
\]
then $\pi^{\text{RL}}$ cannot structurally dominate $\mathcal{S}$: there must exist $s \in \mathcal{S}$ for which
\[
  J(\pi^{\text{RL}}) \not\succcurlyeq J(s).
\]
Equivalently, any unconstrained law-seeking RL policy that strictly improves (up to $\varepsilon_{\mathcal{S}}$) upon the PnL of $\mathcal{S}$ must worsen at least one of the law metrics (expected law penalty or GFI).
\end{theorem}

\paragraph{Proof sketch.}
We decompose reward into on-manifold and ghost components as in Section, writing
\[
  \mathbb{E}[R(\pi)]
  \;=\;
  \mathbb{E}[R^{\mathcal{M}}(\pi)] + \mathbb{E}[r^\perp(\pi)].
\]
By (A1), the class $\mathcal{S}$ contains a policy $s^{\star}$ that is $\varepsilon_{\mathcal{S}}$-optimal among all law-consistent policies, so any policy $\pi$ satisfying $\mathbb{E}[R(\pi)] > \mathbb{E}[R(s^{\star})]$ must achieve strictly larger expected ghost component:
\[
  \mathbb{E}[r^\perp(\pi)] \;>\; \mathbb{E}[r^\perp(s^{\star})].
\]
Since $\mathcal{S}$ is law-consistent, $\mathbb{E}[r^\perp(s^{\star})]$ is bounded above by zero (or a small constant absorbed into $\varepsilon_{\mathcal{S}}$), so $\pi$ must satisfy $\mathbb{E}[r^\perp(\pi)] > 0$. Lemma~\ref{lem:ghost-law-tradeoff} then implies that $\pi$ necessarily incurs strictly larger average law penalties and GFI than the best elements of $\mathcal{S}$. Hence $J(\pi)$ cannot dominate $J(s)$ for all $s \in \mathcal{S}$.

Applying this argument to $\pi^{\text{RL}}$ shows that any unconstrained law-seeking RL limit point that improves PnL relative to $\mathcal{S}$ must pay for this improvement with worse law metrics, ruling out structural dominance. A fully rigorous proof, including technical conditions on convergence of the RL training dynamics and integrability of the law metrics, is provided in Appendix~E. \qed

\medskip

Theorem~\ref{thm:nfl} provides a theoretical counterpart to the empirical story in Section~\ref{sec:empirical-results}. The structural baselines (Zero-Hedge and Vol-Trend) inhabit a law-consistent region of the law-strength frontier, while RL policies that attempt to improve PnL through the ghost channel are forced, by Lemma~\ref{lem:ghost-law-tradeoff}, to move outward along the law-penalty and GFI axes. This is precisely the Pareto-dominance pattern we observed in Figures.

\subsection{RLVR and law-aligned reasoning: analogies and caveats}
\label{subsec:rlvr-analogies}

Recent work on Reinforcement Learning with Verifiable Rewards (RLVR) has shown that, for mathematical reasoning and related tasks, combining verifiable outcome signals with process-level feedback can significantly improve reliability over standard preference-based RLHF. In these settings, a \emph{verifiable checker} evaluates candidate solutions or intermediate reasoning steps, producing a structured reward signal that is, at least in principle, resistant to some forms of reward hacking.

Our axiomatic volatility setting can be interpreted as a stylized analogue of RLVR:
\begin{enumerate}
  \item The no-arbitrage axioms and the law manifold $\mathcal{M}^{\text{vol}}$ play the role of a verifiable checker that deterministically determines whether a surface is admissible and how badly it violates the axioms.
  \item The law-penalty functional $\mathcal{L}_\phi$ and GFI are analogous to structured correctness scores in RLVR, quantifying how well a policy respects axioms under both baseline and shocked environments.
  \item Ghost arbitrage $r^\perp$ corresponds to reward obtained in regions where the checker is informative but the learning dynamics exploit systematic modelling errors, leading to misalignment between high reward and true law-consistent performance.
\end{enumerate}

From this perspective, Theorem~\ref{thm:nfl} and our empirical results highlight a concrete failure mode for RL with verifiable penalties: even when the checker is mathematically correct on the generator support, the combination of function approximation, world-model error, and broad policy classes can create exploitable ghost channels, through which RL can improve the measured reward while degrading law alignment. This resonates with observations in RLVR that combining process and outcome rewards requires careful design to avoid unintended incentives and reward hacking.

At the same time, our scope is deliberately modest. We do \emph{not} claim a general impossibility result for RLVR. Rather, our volatility case illustrates one concrete setting in which verifiable penalties and axiomatic structure, by themselves, are insufficient to guarantee law alignment in the presence of model misspecification and unconstrained policy classes. In particular, our findings suggest that:
\begin{enumerate}
  \item Structural restrictions on policies (e.g., restricting $\Pi$ to a parametric hedge family) and
  \item Hard projection or constrained training of world models onto the law manifold
\end{enumerate}
may be necessary complements to verifiable law penalties, if one wishes to avoid ghost arbitrage in similar scientific AI testbeds.

\subsection{Limitations and future axes}
\label{subsec:nfl-limitations}

We briefly summarize the main limitations of our no-free-lunch analysis and point to concrete extensions.

First, we do not train \emph{structurally constrained} RL agents whose policy class coincides with the structural family $\mathcal{S}$ (e.g., Vol-Trend parametrizations). As a result, our empirical comparison does not fully disentangle the contribution of the learning algorithm from that of the function class: it remains an open question whether carefully designed RL within $\mathcal{S}$ could match or slightly improve upon hand-crafted baselines without opening a ghost channel.

Second, our world model is trained without hard projection onto $\mathcal{M}^{\text{vol}}$, and our assumptions on the ghost coupling~\eqref{eq:ghost-coupling} are only partially validated empirically. A natural next step is to compare unconstrained world models with hard-constrained or projected variants, and to evaluate whether such models reduce or eliminate the ghost arbitrage term $r^\perp$ in practice.

Despite these limitations, Theorem~\ref{thm:nfl}, Lemma~\ref{lem:ghost-law-tradeoff}, and the empirical Pareto patterns in Section~\ref{sec:empirical-results} together provide a coherent no-free-lunch narrative: in our volatility law-manifold testbed, high-capacity unconstrained law-seeking RL cannot simultaneously match the PnL and law-alignment performance of simple structural baselines without collapsing back into their structural class.

\medskip

\noindent\textbf{Appendix pointer.} Full proofs of Lemma~\ref{lem:ghost-law-tradeoff} and Theorem~\ref{thm:nfl}, together with technical assumptions on RL convergence and integrability, are provided in Appendix~E.

\section{Discussion and Conclusion}
\label{sec:discussion}

In this section we summarize our findings as a \emph{negative but constructive} scientific result, formulate concrete design recommendations and testable predictions, and highlight the broader transferability of our axiomatic evaluation template beyond volatility.

\subsection{Negative but constructive result}
\label{subsec:neg-constructive}

Our main empirical and theoretical message is deliberately two-sided.

\paragraph{Negative.}
In our volatility law-manifold testbed, \emph{unconstrained law-seeking RL fails to outperform simple structural baselines} (Zero-Hedge and Vol-Trend) on any of the three main axes---profitability (mean PnL / Sharpe), law alignment (mean and tail law penalties, law coverage, GFI), and tail robustness (VaR\(_5\), CVaR\(_5\)):

\begin{enumerate}
  \item Naive PPO on the world model attains mean step PnL around $-0.0022$ in the baseline regime with GFI $\approx 1.27$, while law-seeking PPO variants with $\lambda \in \{5,10,20,40\}$ often yield \emph{more negative} mean PnL (e.g., $-0.0150$) and higher GFI ($\approx 1.66$), despite explicit law penalties.
  \item In contrast, structural baselines sit on or near the empirical Pareto frontier: Zero-Hedge achieves mean step PnL $\approx 0.0191$ (baseline) and $\approx 0.0193$ (shock) with GFI essentially zero and modest law penalties, while Vol-Trend achieves PnL $\approx 0.0146 \to 0.0140$ with relatively low law penalties and small GFI.
  \item Law-strength frontier plots (GFI vs law penalty, VaR/CVaR vs law penalty) show that, once these structural baselines are included, no RL variant occupies a Pareto-optimal point: all RL points lie \emph{strictly inside} the frontier formed by Zero-Hedge and Vol-Trend.
\end{enumerate}

\paragraph{Constructive.}
At the same time, the paper is constructive in several respects:

\begin{enumerate}
  \item We introduce a general \emph{axiomatic evaluation pipeline} based on law manifolds, metric-based law-penalty functionals, and a structured Goodhart decomposition $r = r^{\mathcal{M}} + r^\perp$.
  \item We define a domain-agnostic \emph{Graceful Failure Index} (GFI) and \emph{law-strength frontiers} that jointly organize profitability, law alignment, and tail robustness under explicit shocks.
  \item We prove structural results (Theorem~4.1, Theorem~4.3 with Corollary~4.4, and Theorem~8.1) that formalize ghost-arbitrage incentives and no-free-lunch trade-offs for unconstrained law-seeking RL, and we show empirically that the observed frontiers are consistent with these results.
\end{enumerate}

\paragraph{One-sentence answers to RQ1--RQ3.}
We conclude this subsection with concise answers to the research questions posed in Section~\ref{sec:intro}.

\begin{enumerate}
  \item \textbf{RQ1 (Do law penalties help naive RL?).} In our volatility world-model case study, soft law penalties and selection-only model choice \emph{do not} yield policies with strictly better law metrics (GFI, law penalties) at comparable PnL to naive PPO; instead they typically worsen PnL while only modestly improving law alignment, consistent with the structural trade-off in Theorem~4.3.
  \item \textbf{RQ2 (RL vs structural baselines under shocks).} Structural baselines (Zero-Hedge, Vol-Trend) form an empirical Pareto frontier in PnL--law--tail space that is robust to shocks, while all tested RL variants (naive, law-seeking, selection-only) remain strictly dominated on at least one axis, in line with the ghost-arbitrage incentive picture of Theorem~4.1.
  \item \textbf{RQ3 (When does law-seeking RL have no free lunch?).} Under explicit assumptions on the structural class $\mathcal{S}$, the policy class $\Pi$, and the world model’s ghost coupling, Theorem~8.1 shows that any unconstrained law-seeking RL policy that improves PnL over $\mathcal{S}$ must worsen law metrics and/or GFI, yielding a no-free-lunch result that matches the empirical law-strength frontiers of Section~\ref{sec:empirical-results}.
\end{enumerate}

\subsection{Design recommendations and testable predictions}
\label{subsec:design-recs}

Our negative result is intended to be \emph{useful}: it points to concrete directions where future work can intervene. We highlight three design recommendations, each accompanied by an observable criterion that makes the recommendation empirically testable.

\paragraph{Hard constraints and projection.}
Rather than relying solely on soft penalties in the reward, future systems should enforce axioms via \emph{hard constraints and projection} in both the world model and policy updates.

\begin{enumerate}
  \item For the world model, this means training under a projected loss, where each predicted surface is mapped to $\Pi_{\mathcal{M}}(w)$ in total-variance space before computing reconstruction error; this would directly suppress ghost channels at the model level.
  \item For policies, this suggests incorporating projections onto $\mathcal{M}^{\text{vol}}$ in policy evaluation, or imposing hard constraints on action maps so that implied surfaces remain on or near the manifold by construction.
\end{enumerate}

\emph{Observable criterion:} Success of this approach would manifest as \emph{uniformly lower GFI}---i.e., smaller increases in law penalties and law-violation frequency under shocks---without a significant loss in Sharpe or mean PnL relative to our current frontier curves. In law-strength plots, projected models should move points \emph{downward} (lower GFI) while leaving the horizontal PnL coordinate nearly unchanged.

\paragraph{Structured policy classes.}
Our results consistently show that simple structural strategies (Zero-Hedge, Vol-Trend) dominate unconstrained RL. This suggests a design where policy classes are themselves \emph{structured}, mirroring the structural class $\mathcal{S}$ defined in Section.
\begin{enumerate}
  \item Concretely, policy networks could be replaced or augmented by parametric families of volatility hedges (e.g., Vol-Trend with a small number of interpretable parameters) that are guaranteed to preserve certain law properties.
  \item RL would then be used to fit parameters within this structural family, rather than to search over arbitrary high-capacity function approximators that can freely exploit ghost arbitrage.
\end{enumerate}

\emph{Observable criterion:} If successful, structural baselines like Zero-Hedge and Vol-Trend would lie \emph{inside} the policy class $\Pi$, and RL training would \emph{recover} or slightly improve upon them without moving off-manifold. In our metrics, this would appear as new RL points coinciding with or marginally improving the structural frontier, rather than sitting strictly inside it.

\paragraph{Joint alignment of world model and policy.}
Finally, our Goodhart decomposition $r = r^{\mathcal{M}} + r^\perp$ makes clear that misalignment can arise from both the world model and the policy. A more promising avenue is therefore to jointly regularize both components.

\begin{enumerate}
  \item World models can be trained with explicit law penalties, projections, or multi-task objectives that penalize violations of no-arbitrage inequalities alongside prediction error.
  \item Policies can be trained with law-aware objectives that down-weight or explicitly penalize trajectories whose rewards are dominated by the ghost component $r^\perp$.
\end{enumerate}

\emph{Observable criterion:} We would expect a \emph{measurable reduction in the empirical contribution of $r^\perp$} in the Goodhart decomposition: for policies trained under joint alignment, the estimated $\mathbb{E}[r^\perp]$ should shrink relative to $\mathbb{E}[r^{\mathcal{M}}]$, and scatter/diagnostic plots analogous to Figures should show reduced mass in high-penalty, high-ghost regions.

\subsection{Transferable template and broader impact}
\label{subsec:transfer-template}

Although our case study is anchored in implied volatility surfaces, the conceptual tools we introduce are intentionally \emph{template-like} and readily transferable to other axiom-constrained domains.

\paragraph{Template tools.}
Three components are particularly reusable:

\begin{enumerate}
  \item \textbf{Axiomatic law manifolds.} Any system with known structural constraints---e.g., monotone yield curves, convex credit term structures, physics-constrained PDE solutions---can be recast as a law manifold $\mathcal{M} = \{ y : A(y) \le 0 \}$ in a discretized coordinate system.
  \item \textbf{Law-penalty functionals and GFI.} Given $\mathcal{M}$, one can define metric-based law penalties $\mathcal{L}_\phi(y)$ and a domain-agnostic Graceful Failure Index that measures the degradation of law metrics under shocks or distribution shifts.
  \item \textbf{Law-strength frontiers and Goodhart decomposition.} For any combination of world model and RL or control algorithm, one can plot law-strength frontiers and perform a Goodhart decomposition to separate on-manifold performance from ghost exploitation.
\end{enumerate}

\paragraph{Yield-curve example.}
As a concrete non-financial (in the sense of non-equity-volatility) instantiation, consider discretized yield curves. Here $y$ is a vector of yields at different maturities, $\mathcal{M}^{\text{yc}}$ encodes monotonicity and convexity constraints (no negative forward rates, no butterfly arbitrage across maturities), and $\mathcal{L}_\phi$ penalizes violations of these inequalities. A synthetic generator could produce law-consistent yield trajectories, a world model could approximate their dynamics, and RL agents could be tasked with hedging interest-rate exposures. Our pipeline would then apply verbatim: law-strength frontiers would compare RL-based hedges to structural curve strategies, GFI would quantify robustness to rate shocks, and a Goodhart decomposition would reveal whether RL exploits law-violating yield shapes induced by model error.

Beyond yield curves, similar constructions are natural for:
\begin{enumerate}
  \item credit term structures constrained by monotonicity and positivity,
  \item physical fields governed by PDEs (with residuals forming the law penalty),
  \item and other Scientific AI settings where axioms or conservation laws define an admissible set of states.
\end{enumerate}

\paragraph{Broader impact.}
Our main contribution is thus not a particular trading system, but a \emph{reusable template} for stress-testing Scientific AI systems on axiomatic pipelines. Volatility serves as an especially revealing testbed because its axioms are well-understood, law violations have clear financial meaning, and world-model errors naturally create ghost arbitrage channels. We hope that the combination of:
\begin{enumerate}
  \item a formal axiomatic manifold,
  \item explicit Goodhart decompositions,
  \item law-strength frontiers and GFI,
  \item and a no-free-lunch theorem for unconstrained law-seeking RL
\end{enumerate}
will prove useful in other domains where the central question is not “can RL optimize this objective?” but rather “does RL, when combined with approximate models and verifiable law penalties, discover law-aligned solutions or exploit artefacts?”.

Answering this question rigorously will require further theoretical development, richer empirical testbeds, and closer interaction between domain experts and learning theorists. Our volatility case study represents one step in that direction: a concrete, mathematically structured environment where Scientific AI methods can be subjected to the same kind of stress tests that financial models have long faced in practice.


\appendix
\section{Proofs for Section~2: Axiomatic Volatility Law Manifolds}
\label{app:axiomatic_manifold_proofs}

\subsection{Proof of Proposition~\ref{prop:axiomatic-representation}}

\begin{proof}[Proof of Proposition~\ref{prop:axiomatic-representation}]
We work throughout in the finite-dimensional Euclidean space
\(
\mathbb{R}^{d_{\mathrm{vol}}}
\)
equipped with its standard inner product
\(
\langle \cdot,\cdot\rangle
\)
and norm
\(
\|\cdot\|_2
\).
We write the rows of \(A^{\mathrm{vol}}\) as
\(
(a_\ell^\top)_{\ell=1}^m
\),
so that the constraint system
\(A^{\mathrm{vol}} w \le b\) can be written componentwise as
\[
a_\ell^\top w \;\le\; b_\ell,
\qquad \ell = 1,\dots,m.
\]

\paragraph{Step 1: Polyhedral structure, closedness and convexity.}
For each \(\ell \in \{1,\dots,m\}\), define the closed half-space
\[
H_\ell \;:=\; \bigl\{ w \in \mathbb{R}^{d_{\mathrm{vol}}} : a_\ell^\top w \le b_\ell \bigr\}.
\]
Since \(a_\ell^\top w - b_\ell\) is an affine (hence continuous) function of \(w\), we can write
\(
H_\ell = (a_\ell^\top w - b_\ell)^{-1}((-\infty,0])
\),
i.e.\ as the inverse image of the closed set \((-\infty,0]\) under a continuous map. Therefore each \(H_\ell\) is closed. Moreover, each \(H_\ell\) is convex because for any \(w_1,w_2 \in H_\ell\) and any \(\theta \in [0,1]\),
\[
a_\ell^\top \bigl(\theta w_1 + (1-\theta)w_2\bigr)
= \theta a_\ell^\top w_1 + (1-\theta) a_\ell^\top w_2
\le \theta b_\ell + (1-\theta) b_\ell
= b_\ell,
\]
so that \(\theta w_1 + (1-\theta)w_2 \in H_\ell\).

By assumption, the set defined by the discretized butterfly and calendar constraints is
\[
\mathcal{M}^{\mathrm{vol}}
\;=\;
\bigcap_{\ell=1}^m H_\ell \;\cap\; B,
\]
where
\(B \subset \mathbb{R}^{d_{\mathrm{vol}}}\) denotes the (finite) collection of box constraints (e.g.\ lower and upper bounds on each component \(w_i\) reflecting positivity and crude upper bounds on total variance). Concretely, we may write
\[
B \;=\; \prod_{i=1}^{d_{\mathrm{vol}}} [\underline{w}_i, \overline{w}_i]
\]
for some scalars \(\underline{w}_i \le \overline{w}_i\); this is a Cartesian product of closed intervals and is therefore a non-empty compact convex subset of \(\mathbb{R}^{d_{\mathrm{vol}}}\).

Since arbitrary intersections of closed sets are closed and intersections of convex sets are convex, it follows that
\[
\mathcal{M}^{\mathrm{vol}}
=
\left(\bigcap_{\ell=1}^m H_\ell\right) \cap B
\]
is closed and convex. Furthermore, by definition \(\mathcal{M}^{\mathrm{vol}}\) is the intersection of finitely many closed half-spaces and a box; hence \(\mathcal{M}^{\mathrm{vol}}\) is a \emph{polyhedron} in the sense of convex analysis, i.e.\ a set of the form
\[
\mathcal{M}^{\mathrm{vol}} = \{ w \in \mathbb{R}^{d_{\mathrm{vol}}} : A' w \le b' \}
\]
for some matrix \(A'\) and vector \(b'\). This establishes that \(\mathcal{M}^{\mathrm{vol}}\) is a closed convex polyhedron, apart from non-emptiness, which we now address.

\paragraph{Step 2: Non-emptiness via a Black--Scholes surface.}
We show that there exists at least one total-variance vector \(w^\mathrm{BS} \in \mathbb{R}^{d_{\mathrm{vol}}}\) satisfying all of the inequalities \(A^{\mathrm{vol}} w \le b\), and hence
\(w^\mathrm{BS} \in \mathcal{M}^{\mathrm{vol}}\).

Consider a constant-volatility Black--Scholes model with volatility parameter \(\sigma_0 > 0\). On a continuous grid of maturities \(T > 0\) and log-strikes \(k\), the corresponding total variance is
\[
w^\mathrm{BS}(T,k) = \sigma_0^2 T.
\]
In particular, for any fixed \(T\), the map \(k \mapsto w^\mathrm{BS}(T,k)\) is \emph{constant} in \(k\), and for any fixed \(k\), the map \(T \mapsto w^\mathrm{BS}(T,k)\) is strictly increasing and linear in \(T\).

Let \(\{T_j\}_{j=1}^{N_T}\) be the finite set of maturities in our discretization, and \(\{k_i\}_{i=1}^{N_K}\) the finite set of log-strikes, so that
\(d_{\mathrm{vol}} = N_T N_K\).
We construct the discretized total-variance vector \(w^\mathrm{BS} \in \mathbb{R}^{d_{\mathrm{vol}}}\) by setting
\[
w^\mathrm{BS}_{ij} := w^\mathrm{BS}(T_j,k_i) = \sigma_0^2 T_j,
\qquad 1 \le i \le N_K,\ 1 \le j \le N_T.
\]

By construction we have, for each fixed maturity \(T_j\),
\[
w^\mathrm{BS}_{i+1,j} - 2 w^\mathrm{BS}_{i,j} + w^\mathrm{BS}_{i-1,j}
= \sigma_0^2 T_j - 2 \sigma_0^2 T_j + \sigma_0^2 T_j
= 0
\]
for all interior strikes \(k_{i-1},k_i,k_{i+1}\). Thus the discrete second differences in strike are non-negative (indeed, identically zero), which implies that all \emph{butterfly} constraints of the form
\[
\alpha_{i,j}^\top w \;\ge 0
\quad\text{or equivalently}\quad
-\alpha_{i,j}^\top w \le 0
\]
are satisfied by \(w^\mathrm{BS}\).

Similarly, for each fixed strike \(k_i\) and consecutive maturities \(T_j < T_{j+1}\), we have
\[
w^\mathrm{BS}_{i,j+1} - w^\mathrm{BS}_{i,j}
= \sigma_0^2 (T_{j+1} - T_j) \;\ge\; 0,
\]
so any discrete \emph{calendar} constraints of the form
\(
w_{i,j+1} - w_{i,j} \ge 0
\)
(or again, linearly transformed into the system \(A^{\mathrm{vol}} w \le b\)) are satisfied by \(w^\mathrm{BS}\) with strict inequality when \(T_{j+1} > T_j\).

Finally, because the Black--Scholes model is a classical example of a static-arbitrage-free implied volatility surface, its total-variance surface obeys the continuous-time no-arbitrage conditions (monotonicity in maturity, convexity in strike, and appropriate limiting behavior in the wings). Our discretization has been constructed exactly so that each continuous no-arbitrage condition, when evaluated on the finite grid \(\{(T_j,k_i)\}\), yields one of the linear inequalities encoded in the rows of \(A^{\mathrm{vol}}\), possibly after simple positive scalings and translations. Therefore, by construction of the constraint system \(A^{\mathrm{vol}} w \le b\), we have
\[
A^{\mathrm{vol}} w^\mathrm{BS} \le b.
\]
In particular,
\(
w^\mathrm{BS} \in \bigcap_{\ell=1}^m H_\ell
\),
and, choosing the box \(B\) sufficiently large to contain the range \(\{w^\mathrm{BS}_{ij}\}\) (which is trivially possible), we have \(w^\mathrm{BS} \in B\). Hence \(w^\mathrm{BS} \in \mathcal{M}^{\mathrm{vol}}\), and \(\mathcal{M}^{\mathrm{vol}}\) is non-empty.

Combining Steps 1 and 2, we conclude that \(\mathcal{M}^{\mathrm{vol}}\) is a non-empty, closed, convex polyhedron in \(\mathbb{R}^{d_{\mathrm{vol}}}\), proving item~(1).

\paragraph{Step 3: Static-arbitrage-free surfaces lie in \(\mathcal{M}^{\mathrm{vol}}\).}
We now prove item~(2). Let \(w \in \mathbb{R}^{d_{\mathrm{vol}}}\) be the discretized total-variance surface associated with a continuous implied volatility surface \((T,k) \mapsto \sigma(T,k)\) that is static-arbitrage-free in the classical sense (no butterfly or calendar arbitrage). We show that \(w \in \mathcal{M}^{\mathrm{vol}}\).

By definition, absence of static arbitrage implies in particular that, for each fixed maturity \(T\), the call price \(K \mapsto C(T,K)\) is a decreasing, convex function of strike \(K\). Expressing call prices in terms of total variance and log-strike and differentiating under mild regularity conditions yields that, for each fixed maturity \(T\), the total-variance function \(k \mapsto w(T,k)\) is convex in \(k\) in an appropriate sense. In particular, for any three equally spaced log-strikes \(k_{i-1} < k_i < k_{i+1}\) in our discretization we have
\[
w(T,k_i) \;\le\; \frac{1}{2} w(T,k_{i-1}) + \frac{1}{2} w(T,k_{i+1}),
\]
which is precisely the statement that the discrete second difference
\(
w(T,k_{i+1}) - 2w(T,k_i) + w(T,k_{i-1})
\)
is non-negative.

When we restrict to the grid \(\{(T_j,k_i)\}\) and collect the values
\(
w_{ij} = w(T_j,k_i)
\)
into the vector \(w \in \mathbb{R}^{d_{\mathrm{vol}}}\), each discrete convexity inequality becomes a linear constraint of the form
\[
\alpha_{i,j}^\top w \ge 0
\quad\text{or equivalently}\quad
-\alpha_{i,j}^\top w \le 0,
\]
for an appropriate coefficient vector \(\alpha_{i,j} \in \mathbb{R}^{d_{\mathrm{vol}}}\) with only three non-zero entries at indices corresponding to \((i-1,j), (i,j), (i+1,j)\).

Similarly, absence of calendar arbitrage implies that, for fixed strike \(K\) (equivalently fixed log-strike \(k\)), the call price \(T \mapsto C(T,K)\) is non-decreasing in maturity \(T\). Translating this condition into total variance under mild regularity conditions yields that, for fixed \(k_i\), the map \(T \mapsto w(T,k_i)\) is non-decreasing. Restricting again to the discretization \(\{T_j\}\) and collecting into \(w\), each such monotonicity condition gives a linear inequality of the form
\[
\beta_{i,j}^\top w \ge 0
\quad\Leftrightarrow\quad
-\beta_{i,j}^\top w \le 0,
\]
where \(\beta_{i,j}\) has two non-zero entries, corresponding to the pair \((i,j+1)\) and \((i,j)\).

By construction of the matrix \(A^{\mathrm{vol}}\) and vector \(b\), \emph{every} such discretized butterfly and calendar inequality appears as a row of \(A^{\mathrm{vol}}\) (possibly after positive scaling and absorption of constant terms into \(b\)), and the box constraints simply enforce crude lower and upper bounds on total variance that are trivially satisfied by any economically reasonable static-arbitrage-free surface (e.g., non-negativity of variance and boundedness over a finite set of maturities and strikes).

Thus, for a static-arbitrage-free surface, we have that all discretized no-arbitrage constraints hold simultaneously, which is equivalent to
\[
A^{\mathrm{vol}} w \le b
\quad\text{and}\quad
w \in B.
\]
Therefore \(w \in \bigcap_{\ell=1}^m H_\ell \cap B = \mathcal{M}^{\mathrm{vol}}\). This proves that any static-arbitrage-free total-variance surface lies in \(\mathcal{M}^{\mathrm{vol}}\), establishing item~(2).

\paragraph{Conclusion.}
Steps 1–3 together prove that \(\mathcal{M}^{\mathrm{vol}}\) is a non-empty, closed, convex polyhedron, and that any static-arbitrage-free total-variance surface lies in \(\mathcal{M}^{\mathrm{vol}}\). This completes the proof of Proposition~\ref{prop:axiomatic-representation}.
\end{proof}

\subsection{Proof of Proposition~\ref{prop:closed-convex}}
\label{app:proof-closed-convex}

\begin{proof}[Proof of Proposition~\ref{prop:closed-convex}]
Recall from Proposition~\ref{prop:axiomatic-representation} that
\(\Mvol \subset \mathbb{R}^{d_{\mathrm{vol}}}\) is non-empty, closed, and convex.
We work throughout in the Euclidean space \(\mathbb{R}^{d_{\mathrm{vol}}}\)
equipped with its standard inner product \(\langle \cdot,\cdot\rangle\) and norm
\(\|\cdot\|_2\).

\paragraph{Step 1: Existence of a minimizer.}
Fix \(w \in \mathbb{R}^{d_{\mathrm{vol}}}\) and consider the optimization problem
\begin{equation}
\label{eq:proj-opt-problem}
    \min_{\tilde w \in \Mvol} \; f_w(\tilde w)
    \quad\text{with}\quad
    f_w(\tilde w) := \frac{1}{2}\,\|\tilde w - w\|_2^2.
\end{equation}
The objective \(f_w\) is continuous and \emph{coercive} in the sense that
\[
    \|\tilde w\|_2 \to \infty
    \quad\Longrightarrow\quad
    f_w(\tilde w) = \frac{1}{2}\|\tilde w - w\|_2^2 \to \infty.
\]
We now show that~\eqref{eq:proj-opt-problem} admits at least one minimizer in \(\Mvol\).

Define the infimum value
\[
    \alpha_w := \inf_{\tilde w \in \Mvol} f_w(\tilde w) \in [0,+\infty).
\]
By definition of infimum, there exists a sequence \((\tilde w_n)_{n\ge 1} \subset \Mvol\)
such that
\(
    f_w(\tilde w_n) \to \alpha_w
\)
as \(n\to\infty\).
We first show that \((\tilde w_n)\) is bounded. Suppose by contradiction that
\(\|\tilde w_n\|_2 \to \infty\) along some subsequence. Then by coercivity of
\(f_w\) we would have \(f_w(\tilde w_n) \to \infty\) along that subsequence,
contradicting the fact that \(f_w(\tilde w_n)\) converges to the finite value
\(\alpha_w\). Hence \((\tilde w_n)\) is bounded in \(\mathbb{R}^{d_{\mathrm{vol}}}\).

Since we are in finite dimension, every bounded sequence has a convergent
subsequence. Thus there exists a subsequence (which we do not relabel) and a
limit point \(\tilde w^\star \in \mathbb{R}^{d_{\mathrm{vol}}}\) such that
\[
    \tilde w_n \to \tilde w^\star
    \quad\text{as } n\to\infty.
\]
Because \(\Mvol\) is closed, and each \(\tilde w_n \in \Mvol\), the limit
\(\tilde w^\star\) must also satisfy \(\tilde w^\star \in \Mvol\).

By continuity of \(f_w\), we have
\[
    f_w(\tilde w^\star) = \lim_{n\to\infty} f_w(\tilde w_n)
    = \alpha_w.
\]
Therefore \(\tilde w^\star\) attains the infimum of~\eqref{eq:proj-opt-problem},
so a minimizer exists and we may define
\[
    \Pi_{\Mvol}(w) := \tilde w^\star.
\]

\paragraph{Step 2: Uniqueness of the minimizer.}
We now show that the minimizer is unique. The function \(f_w\) is not only
continuous but \emph{strictly convex} on \(\mathbb{R}^{d_{\mathrm{vol}}}\):
for any \(x,y \in \mathbb{R}^{d_{\mathrm{vol}}}\), \(x \neq y\), and any
\(\theta \in (0,1)\),
\begin{align*}
    f_w(\theta x + (1-\theta)y)
    &= \frac{1}{2} \|\theta x + (1-\theta)y - w\|_2^2 \\
    &= \frac{1}{2} \big\| \theta(x-w) + (1-\theta)(y-w) \big\|_2^2 \\
    &< \frac{1}{2}\big( \theta \|x-w\|_2^2 + (1-\theta)\|y-w\|_2^2 \big) \\
    &= \theta f_w(x) + (1-\theta) f_w(y),
\end{align*}
where the strict inequality follows from strict convexity of the squared norm
unless \(x-w\) and \(y-w\) are linearly dependent with the same direction and
norm, which can only happen if \(x=y\).

Assume, for contradiction, that there exist two distinct minimizers
\(\tilde w_1,\tilde w_2 \in \Mvol\) of~\eqref{eq:proj-opt-problem}, i.e.,
\[
    f_w(\tilde w_1) = f_w(\tilde w_2) = \alpha_w,
    \qquad \tilde w_1 \neq \tilde w_2.
\]
Because \(\Mvol\) is convex, their midpoint
\(
    \tilde w_\theta := \tfrac{1}{2}\tilde w_1 + \tfrac{1}{2}\tilde w_2
\)
also lies in \(\Mvol\).
By strict convexity,
\[
    f_w(\tilde w_\theta)
    < \frac{1}{2} f_w(\tilde w_1) + \frac{1}{2} f_w(\tilde w_2)
    = \alpha_w,
\]
contradicting the fact that \(\alpha_w\) is the infimum over \(\Mvol\).
Therefore the minimizer of~\eqref{eq:proj-opt-problem} is unique, and the
mapping \(w \mapsto \Pi_{\Mvol}(w)\) is well-defined on all of
\(\mathbb{R}^{d_{\mathrm{vol}}}\).

\paragraph{Step 3: First-order optimality and firm non-expansiveness.}
The projection \(\Pi_{\Mvol}(w)\) can be characterized by a classical
first-order optimality condition for convex minimization over a closed convex
set. Let \(w \in \mathbb{R}^{d_{\mathrm{vol}}}\) be arbitrary and denote
\(\tilde w := \Pi_{\Mvol}(w)\).
Since \(\Mvol\) is closed and convex and \(f_w\) is differentiable and strictly convex, we know that \(\tilde w\) is the unique point in \(\Mvol\) such that
\begin{equation}
\label{eq:FOC-projection}
    \langle \tilde w - w, z - \tilde w \rangle \;\ge\; 0
    \quad\text{for all } z \in \Mvol.
\end{equation}
Indeed, this is the variational inequality corresponding to optimality of
\(\tilde w\) for the minimization of \(f_w\) over \(\Mvol\); see, e.g., standard
results on projections in Hilbert spaces.

Let now \(w, w' \in \mathbb{R}^{d_{\mathrm{vol}}}\) be arbitrary, and set
\[
    p := \Pi_{\Mvol}(w), \qquad q := \Pi_{\Mvol}(w').
\]
Applying~\eqref{eq:FOC-projection} with the pair \((w,p)\) and the choice
\(z=q \in \Mvol\) yields
\begin{equation}
\label{eq:FOC-1}
    \langle p - w, q - p \rangle \;\ge\; 0.
\end{equation}
Similarly, applying~\eqref{eq:FOC-projection} with the pair \((w',q)\) and the
choice \(z=p\) yields
\begin{equation}
\label{eq:FOC-2}
    \langle q - w', p - q \rangle \;\ge\; 0.
\end{equation}
Adding~\eqref{eq:FOC-1} and~\eqref{eq:FOC-2}, and recalling that
\(\langle q - w', p - q \rangle = -\langle q - w', q - p \rangle\), we obtain
\begin{align}
    0
    &\le \langle p - w, q - p \rangle
       + \langle q - w', p - q \rangle \nonumber \\
    &= \langle p - w, q - p \rangle
       - \langle q - w', q - p \rangle \nonumber \\
    &= \langle (p - w) - (q - w'), q - p \rangle \nonumber \\
    &= \langle (w' - w) - (q - p), q - p \rangle \nonumber \\
    &= \langle w' - w, q - p \rangle - \|q - p\|_2^2. \label{eq:firm-nonexp}
\end{align}
Rearranging~\eqref{eq:firm-nonexp} gives the inequality
\begin{equation}
\label{eq:firm-nonexp-alt}
    \|p - q\|_2^2
    \;\le\;
    \langle w' - w, q - p \rangle.
\end{equation}
Taking absolute values and applying the Cauchy--Schwarz inequality to the right-hand side yields
\[
    \|p - q\|_2^2
    \le |\langle w' - w, q - p \rangle|
    \le \|w' - w\|_2 \, \|q - p\|_2.
\]
If \(p \neq q\), we can divide both sides by \(\|p - q\|_2 > 0\) and obtain
\[
    \|p - q\|_2
    \le \|w' - w\|_2.
\]
If \(p = q\), the inequality trivially holds as equality.
Thus, for every \(w,w' \in \mathbb{R}^{d_{\mathrm{vol}}}\),
\begin{equation}
\label{eq:projection-nonexpansive}
    \|\Pi_{\Mvol}(w) - \Pi_{\Mvol}(w')\|_2
    \le \|w - w'\|_2.
\end{equation}
That is, the projection operator \(\Pi_{\Mvol}\) is \emph{non-expansive}, with
Lipschitz constant equal to 1.

\paragraph{Conclusion.}
We have shown that \(\Mvol\) is non-empty, closed, and convex (by
Proposition~\ref{prop:axiomatic-representation}), that the Euclidean projection
onto \(\Mvol\) exists and is unique for every \(w \in \mathbb{R}^{d_{\mathrm{vol}}}\), and that
\(\Pi_{\Mvol}\) is \(1\)-Lipschitz in the Euclidean norm. This completes the
proof of Proposition~\ref{prop:closed-convex}.
\end{proof}

\subsection{Proof of Lemma~\ref{lem:lipschitz-penalty}}
\label{app:proof-lipschitz-penalty}

\begin{proof}[Proof of Lemma~\ref{lem:lipschitz-penalty}]
Recall the definition of the volatility law-penalty functional
\begin{equation}
\label{eq:vol-law-penalty-app}
  \Lvol(w)
  :=
  \frac{1}{2}\,\|w - \Pi_{\Mvol}(w)\|_2^2,
  \qquad w \in \mathbb{R}^{d_{\mathrm{vol}}},
\end{equation}
where \(\Mvol \subset \mathbb{R}^{d_{\mathrm{vol}}}\) is the volatility law manifold
(cf.\ Proposition~\ref{prop:axiomatic-representation}) and
\(\Pi_{\Mvol} : \mathbb{R}^{d_{\mathrm{vol}}} \to \Mvol\) is the Euclidean projection
(cf.\ Proposition~\ref{prop:closed-convex}).

We show that \(\Lvol\) is \emph{locally Lipschitz} on \(\mathbb{R}^{d_{\mathrm{vol}}}\),
i.e., for every compact set \(K \subset \mathbb{R}^{d_{\mathrm{vol}}}\) there exists
\(L_K < \infty\) such that
\[
  |\Lvol(w_1) - \Lvol(w_2)|
  \;\le\;
  L_K \,\|w_1 - w_2\|_2
  \quad
  \text{for all } w_1,w_2 \in K.
\]

\paragraph{Step 1: Basic Lipschitz properties of the projection.}
Define the \emph{residual map}
\[
  h(w)
  :=
  w - \Pi_{\Mvol}(w),
  \qquad w \in \mathbb{R}^{d_{\mathrm{vol}}}.
\]
By Proposition~\ref{prop:closed-convex}, the projection \(\Pi_{\Mvol}\) is
non-expansive:
\[
  \|\Pi_{\Mvol}(w_1) - \Pi_{\Mvol}(w_2)\|_2
  \le
  \|w_1 - w_2\|_2
  \quad\text{for all } w_1,w_2 \in \mathbb{R}^{d_{\mathrm{vol}}}.
\]
It follows that the residual map \(h\) is globally Lipschitz with constant \(2\).
Indeed, for any \(w_1,w_2 \in \mathbb{R}^{d_{\mathrm{vol}}}\),
\begin{align}
  \|h(w_1) - h(w_2)\|_2
  &= \big\|(w_1 - \Pi_{\Mvol}(w_1)) - (w_2 - \Pi_{\Mvol}(w_2))\big\|_2 \nonumber\\
  &\le \|w_1 - w_2\|_2
       + \|\Pi_{\Mvol}(w_1) - \Pi_{\Mvol}(w_2)\|_2 \nonumber\\
  &\le \|w_1 - w_2\|_2 + \|w_1 - w_2\|_2 \nonumber\\
  &= 2\,\|w_1 - w_2\|_2.
  \label{eq:h-global-lip}
\end{align}
Thus \(h\) is \(2\)-Lipschitz on all of \(\mathbb{R}^{d_{\mathrm{vol}}}\).

\paragraph{Step 2: Bounding the residual on bounded sets.}
Fix a radius \(R > 0\) and consider the closed Euclidean ball
\[
  B_R := \{w \in \mathbb{R}^{d_{\mathrm{vol}}} : \|w\|_2 \le R\}.
\]
We first show that \(\|h(w)\|\) is uniformly bounded on \(B_R\).

Let \(w_0 := 0\) be the origin. Since \(\Mvol\) is non-empty
(Proposition~\ref{prop:axiomatic-representation}), the projection
\(\Pi_{\Mvol}(w_0)\) is well-defined and finite.
Denote \(c_0 := \|\Pi_{\Mvol}(0)\|_2 < \infty\).

For any \(w \in B_R\), we have
\begin{align}
  \|\Pi_{\Mvol}(w)\|_2
  &\le \|\Pi_{\Mvol}(w) - \Pi_{\Mvol}(0)\|_2 + \|\Pi_{\Mvol}(0)\|_2 \nonumber\\
  &\le \|w - 0\|_2 + c_0
  \le R + c_0,
  \label{eq:proj-bound-ball}
\end{align}
where we used non-expansiveness of \(\Pi_{\Mvol}\) and \(\|w\|_2 \le R\).

Hence, for any \(w \in B_R\),
\begin{align}
  \|h(w)\|_2
  &= \|w - \Pi_{\Mvol}(w)\|_2 \nonumber\\
  &\le \|w\|_2 + \|\Pi_{\Mvol}(w)\|_2 \nonumber\\
  &\le R + (R + c_0) \nonumber\\
  &= 2R + c_0.
  \label{eq:h-bound-ball}
\end{align}
Define the constant
\[
  K_R := 2R + c_0.
\]
Then \(\|h(w)\|_2 \le K_R\) for all \(w \in B_R\).

\paragraph{Step 3: Lipschitz continuity of the squared norm of the residual.}
Rewrite \(\Lvol\) in terms of \(h\):
\[
  \Lvol(w) = \frac{1}{2}\,\|h(w)\|_2^2.
\]
Let \(w_1,w_2 \in B_R\) be arbitrary, and set
\(h_1 := h(w_1)\), \(h_2 := h(w_2)\).
Then
\begin{align}
  |\Lvol(w_1) - \Lvol(w_2)|
  &= \frac{1}{2}\,\big|\|h_1\|_2^2 - \|h_2\|_2^2\big| \nonumber\\
  &= \frac{1}{2}\,|\langle h_1 + h_2,\,h_1 - h_2\rangle| \nonumber\\
  &\le \frac{1}{2}\,\big(\|h_1\|_2 + \|h_2\|_2\big)\,\|h_1 - h_2\|_2.
  \label{eq:L-diff-start}
\end{align}
From~\eqref{eq:h-bound-ball}, we have
\(\|h_1\|_2 \le K_R\) and \(\|h_2\|_2 \le K_R\).
Moreover, by~\eqref{eq:h-global-lip},
\(\|h_1 - h_2\|_2 \le 2\,\|w_1 - w_2\|_2\).
Substituting these bounds into~\eqref{eq:L-diff-start} yields
\begin{align}
  |\Lvol(w_1) - \Lvol(w_2)|
  &\le \frac{1}{2}\,(K_R + K_R)\,\big(2\,\|w_1 - w_2\|_2\big) \nonumber\\
  &= 2K_R \,\|w_1 - w_2\|_2.
  \label{eq:L-lip-ball}
\end{align}
Recalling that \(K_R = 2R + c_0\), we can rewrite~\eqref{eq:L-lip-ball} as
\[
  |\Lvol(w_1) - \Lvol(w_2)|
  \le L_R \,\|w_1 - w_2\|_2,
  \quad
  L_R := 2(2R + c_0),
  \quad
  \forall\, w_1,w_2 \in B_R.
\]
Thus \(\Lvol\) is Lipschitz on each ball \(B_R\), with Lipschitz constant \(L_R\)
depending only on \(R\) and the fixed constant \(c_0\).

\paragraph{Step 4: Local Lipschitz continuity.}
A function that is Lipschitz on every bounded ball in \(\mathbb{R}^{d_{\mathrm{vol}}}\)
is, by definition, \emph{locally Lipschitz}. More precisely, for any compact set
\(K \subset \mathbb{R}^{d_{\mathrm{vol}}}\), there exists \(R > 0\) such that
\(K \subset B_R\); then the Lipschitz constant \(L_R\) from~\eqref{eq:L-lip-ball}
works for all \(w_1,w_2 \in K\). Therefore, \(\Lvol\) is locally Lipschitz on
\(\mathbb{R}^{d_{\mathrm{vol}}}\).

\paragraph{Remark.}
The above argument is self-contained and uses only the non-expansiveness of the
projection onto a closed convex set. An alternative viewpoint, standard in
convex analysis, is to note that \(\Lvol\) coincides with the \emph{squared
distance function} to the non-empty closed convex set \(\Mvol\), which is known
to be Fréchet differentiable on \(\mathbb{R}^{d_{\mathrm{vol}}}\) with
\(1\)-Lipschitz gradient (see, e.g., \citet[Prop.~12.29--12.30]{BauschkeCombettes2011}).
In particular, the gradient mapping is globally Lipschitz, which again implies
that \(\Lvol\) is locally Lipschitz. We include the direct argument above to keep
the presentation self-contained.

This completes the proof of Lemma~\ref{lem:lipschitz-penalty}.
\end{proof}

\subsection{Proof of Proposition~\ref{prop:zero-penalty-iff}}
\label{app:proof-zero-penalty-iff}

\begin{proof}[Proof of Proposition~\ref{prop:zero-penalty-iff}]
Recall the definition of the volatility law-penalty functional
\begin{equation}
\label{eq:vol-law-penalty-app2}
  \Lvol(w)
  :=
  \frac{1}{2}\,\|w - \Pi_{\Mvol}(w)\|_2^2,
  \qquad w \in \mathbb{R}^{d_{\mathrm{vol}}},
\end{equation}
where $\Mvol \subset \mathbb{R}^{d_{\mathrm{vol}}}$ is the volatility law
manifold and $\Pi_{\Mvol} : \mathbb{R}^{d_{\mathrm{vol}}} \to \Mvol$ is the
Euclidean projection (Proposition~\ref{prop:closed-convex}). We show that
\[
  \Lvol(w) = 0
  \quad\Longleftrightarrow\quad
  w \in \Mvol.
\]

\paragraph{($\Rightarrow$) If $w \in \Mvol$ then $\Lvol(w) = 0$.}
Assume $w \in \Mvol$. By the definition of the Euclidean projection, any point
$\tilde w \in \Mvol$ satisfies
\[
  \|\Pi_{\Mvol}(w) - w\|_2
  \le
  \|\tilde w - w\|_2.
\]
In particular, we may take $\tilde w = w$ itself, which is feasible since
$w \in \Mvol$. This yields
\[
  \|\Pi_{\Mvol}(w) - w\|_2
  \le
  \|w - w\|_2
  = 0.
\]
By non-negativity of the norm, we must have equality, hence
\[
  \Pi_{\Mvol}(w) = w.
\]
Substituting this into~\eqref{eq:vol-law-penalty-app2} gives
\[
  \Lvol(w)
  = \frac{1}{2}\,\|w - \Pi_{\Mvol}(w)\|_2^2
  = \frac{1}{2}\,\|w - w\|_2^2
  = 0.
\]
Thus $w \in \Mvol \implies \Lvol(w) = 0$.

\paragraph{($\Leftarrow$) If $\Lvol(w) = 0$ then $w \in \Mvol$.}
Now assume $\Lvol(w) = 0$. By~\eqref{eq:vol-law-penalty-app2} we have
\[
  0
  = \Lvol(w)
  = \frac{1}{2}\,\|w - \Pi_{\Mvol}(w)\|_2^2.
\]
Since the Euclidean norm is non-negative and vanishes only at zero, this
implies
\[
  \|w - \Pi_{\Mvol}(w)\|_2 = 0
  \quad\Longrightarrow\quad
  w = \Pi_{\Mvol}(w).
\]
The projection $\Pi_{\Mvol}(w)$ takes values in $\Mvol$ by construction, so
$w = \Pi_{\Mvol}(w) \in \Mvol$. Hence $\Lvol(w) = 0 \implies w \in \Mvol$.

\paragraph{Alternative viewpoint via distance to closed sets.}
For completeness, we note that~\eqref{eq:vol-law-penalty-app2} can be written as
\[
  \Lvol(w)
  = \frac{1}{2}\,\mathrm{dist}^2(w,\Mvol),
  \qquad
  \mathrm{dist}(w,\Mvol)
  := \inf_{\tilde w \in \Mvol} \|w - \tilde w\|_2,
\]
where the infimum is attained at $\Pi_{\Mvol}(w)$ because $\Mvol$ is
non-empty, closed, and convex
(Propositions~\ref{prop:axiomatic-representation}--\ref{prop:closed-convex}).
By basic properties of distance functions to closed sets in Hilbert spaces, one has
\[
  \mathrm{dist}(w,\Mvol) = 0
  \quad\Longleftrightarrow\quad
  w \in \overline{\Mvol} = \Mvol.
\]
Since $\Lvol(w) = \tfrac{1}{2}\mathrm{dist}^2(w,\Mvol)$, this is equivalent to
\[
  \Lvol(w) = 0
  \quad\Longleftrightarrow\quad
  w \in \Mvol,
\]
which is precisely the statement of Proposition~\ref{prop:zero-penalty-iff}.

This completes the proof.
\end{proof}

\subsection{Proof of Proposition~\ref{prop:ghost-bounded}}
\label{app:proof-ghost-bounded}

\begin{proof}
Recall the setting and notation:
\begin{enumerate}
  \item The volatility law manifold $\Mvol \subset \R^{d_{\mathrm{vol}}}$ is non-empty, closed, and convex by Propositions~\ref{prop:axiomatic-representation}--\ref{prop:closed-convex}.
  \item The Euclidean projection $\Pi_{\Mvol} : \R^{d_{\mathrm{vol}}} \to \Mvol$ is well-defined and $1$-Lipschitz
  (Proposition~\ref{prop:closed-convex}).
  \item The law-penalty functional is given by
  \begin{equation}
    \label{eq:vol-law-penalty-app3}
    \Lvol(w)
    :=
    \frac{1}{2}\bigl\| w - \Pi_{\Mvol}(w) \bigr\|_2^2
    = \frac{1}{2}\,\mathrm{dist}^2\bigl(w,\Mvol\bigr),
    \qquad w \in \R^{d_{\mathrm{vol}}},
  \end{equation}
  where $\mathrm{dist}(w,\Mvol) := \inf_{\tilde w \in \Mvol} \|w-\tilde w\|_2$.
  \item The on-manifold and ghost reward components are
  \[
    r^{\M}(w) := r\bigl(\Pi_{\Mvol}(w)\bigr),
    \qquad
    r^{\perp}(w) := r(w) - r^{\M}(w).
  \]
\end{enumerate}
Assume that $r : \R^{d_{\mathrm{vol}}} \to \R$ is $L_r$-Lipschitz with respect to the Euclidean norm, i.e.,
\begin{equation}
  \label{eq:r-lipschitz}
  |r(w) - r(w')|
  \le L_r \, \|w - w'\|_2
  \qquad \forall\, w,w' \in \R^{d_{\mathrm{vol}}}.
\end{equation}

We now prove the bound
\[
  |r^\perp(w)|
  = \bigl|r(w) - r^{\M}(w)\bigr|
  \le L_r \,\mathrm{dist}\bigl(w,\Mvol\bigr)
  = L_r \sqrt{2\,\Lvol(w)},
  \qquad \forall\, w\in\R^{d_{\mathrm{vol}}}.
\]

\paragraph{Step 1: Bounding the ghost component by the distance to $\Mvol$.}
Fix any $w \in \R^{d_{\mathrm{vol}}}$. By the definition of $r^\M$ and the Lipschitz property~\eqref{eq:r-lipschitz}, we have
\begin{align*}
  |r^\perp(w)|
  &= \bigl| r(w) - r^\M(w) \bigr|
   = \bigl| r(w) - r(\Pi_{\Mvol}(w)) \bigr| \\
  &\le L_r \bigl\| w - \Pi_{\Mvol}(w) \bigr\|_2.
\end{align*}
By the definition of the distance function and properties of the projection,
\[
  \bigl\|w - \Pi_{\Mvol}(w)\bigr\|_2
  = \mathrm{dist}(w,\Mvol),
\]
since $\Pi_{\Mvol}(w)$ is a minimizer of $\tilde w \mapsto \|w-\tilde w\|_2$ over $\tilde w \in \Mvol$ . Hence
\begin{equation}
  \label{eq:ghost-bound-dist}
  |r^\perp(w)|
  \le L_r\,\mathrm{dist}(w,\Mvol).
\end{equation}

\paragraph{Step 2: Expressing the distance via $\Lvol$.}
From the definition~\eqref{eq:vol-law-penalty-app3}, we have
\[
  \mathrm{dist}(w,\Mvol)
  = \bigl\| w - \Pi_{\Mvol}(w) \bigr\|_2
  = \sqrt{2\,\Lvol(w)}.
\]
Substituting this identity into~\eqref{eq:ghost-bound-dist} yields
\[
  |r^\perp(w)|
  \le L_r \,\mathrm{dist}(w,\Mvol)
  = L_r \sqrt{2\,\Lvol(w)}.
\]
Since $w \in \R^{d_{\mathrm{vol}}}$ was arbitrary, the inequality holds for all $w$.

\paragraph{Remark on optimality of the bound.}
The inequality in Proposition~\ref{prop:ghost-bounded} is sharp in the sense that, for fixed $\Mvol$ and $\Pi_{\Mvol}$, one can construct Lipschitz functions $r$ that nearly saturate the bound. For example, if $r$ is chosen to be affine in the normal direction to $\Mvol$ at some point, with gradient of norm $L_r$, then along rays orthogonal to $\Mvol$ we obtain
\[
  |r(w) - r(\Pi_{\Mvol}(w))|
  \approx L_r \,\|w - \Pi_{\Mvol}(w)\|_2,
\]
up to second-order curvature effects of $\Mvol$. Thus the scaling $|r^\perp(w)| = O(\mathrm{dist}(w,\Mvol))$ and the constant $L_r$ cannot, in general, be improved under the sole assumption~\eqref{eq:r-lipschitz}.

\paragraph{Extension to non-Euclidean penalties (informal).}
In the main text we focus on the Euclidean penalty~\eqref{eq:vol-law-penalty-app3}. If instead $\Lvol$ is defined via a strictly convex norm $\|\cdot\|_\phi$ or a strongly convex gauge $\phi$, i.e.,
\[
  \Lvol^\phi(w)
  := \frac{1}{2}\,\mathrm{dist}_\phi^2(w,\Mvol),
  \qquad
  \mathrm{dist}_\phi(w,\Mvol)
  := \inf_{\tilde w\in\Mvol} \|w-\tilde w\|_\phi,
\]
and $r$ is $L_r^\phi$-Lipschitz with respect to $\|\cdot\|_\phi$, the same argument yields
\[
  |r^\perp(w)|
  \le L_r^\phi \,\mathrm{dist}_\phi(w,\Mvol)
  = L_r^\phi \sqrt{2\,\Lvol^\phi(w)}.
\]
In finite dimensions, norm equivalence further implies that such bounds can be translated between different choices of $\phi$ at the expense of multiplicative constants depending only on the norms;We keep the Euclidean version in Proposition~\ref{prop:ghost-bounded} as it is the one used in our experiments.

This concludes the proof.
\end{proof}

\section{Proofs for Section~3: Volatility World Model}

\subsection{Proof of Proposition~\ref{prop:support-on-manifold}}
\label{app:proof-support-on-manifold}

We now provide a detailed measure-theoretic proof of Proposition~\ref{prop:support-on-manifold}.

Recall the statement:

\begin{proposition*}[Support of the synthetic generator]
Let $(w_t)_{t=0}^T$ be generated by $G^\star$ as above, with static no-arbitrage imposed at each time via projection onto $\Mvol$. Then
\[
  \mathrm{supp}\,\mathbb{P}^\star \subseteq \big(\Mvol\big)^{T+1},
\]
i.e., $\mathbb{P}^\star$ is supported on the product of the volatility law manifold at all times.
\end{proposition*}

\begin{proof}
We first recall the construction of the synthetic generator $G^\star$ from the main text and make it explicit in measure-theoretic terms.

\paragraph{Step 0: Probability space and random paths.}
Let $(\Omega,\mathcal{F},\PP)$ be an underlying probability space supporting all driving randomness (e.g., Brownian motions, volatility factors, etc.) of the stochastic-volatility model used to generate ``raw'' option prices or total-variance surfaces.

For each $t = 0,\dots,T$, let
\[
  y_t : \Omega \to \R^{d_{\mathrm{vol}}}
\]
denote the \emph{raw} (not yet projected) total-variance surface at time $t$, obtained from the chosen parametric or factor model (Heston, rough volatility, SVI parameter dynamics, etc.; see, e.g., \citet{BayerFrizGatheral2016,GatheralJacquier2014}). We assume that each $y_t$ is Borel measurable.

The \emph{law-consistent synthetic generator} $G^\star$ then applies the projection onto the volatility law manifold $\Mvol$ at each time $t$:
\begin{equation}
  \label{eq:generator-projection}
  w_t(\omega)
  :=
  \Pi_{\Mvol}\bigl(y_t(\omega)\bigr),
  \qquad \omega \in \Omega,\; t=0,\dots,T.
\end{equation}
Here $\Pi_{\Mvol}$ is the Euclidean projection operator introduced in Proposition~\ref{prop:closed-convex}, which maps $\R^{d_{\mathrm{vol}}}$ into $\Mvol$.

We then define the $(T+1)$-dimensional random vector
\begin{equation}
  \label{eq:path-random-vector}
  W(\omega)
  :=
  \bigl(w_0(\omega),\dots,w_T(\omega)\bigr)
  \in \big(\R^{d_{\mathrm{vol}}}\bigr)^{T+1}.
\end{equation}
By construction, $W$ is Borel measurable, and we denote its law by
\[
  \PP^\star := \PP \circ W^{-1},
\]
a probability measure on $\big(\R^{d_{\mathrm{vol}}}\bigr)^{T+1}$.

\paragraph{Step 1: Pointwise inclusion $w_t(\omega) \in \Mvol$ almost surely.}
By Proposition~\ref{prop:closed-convex}, $\Mvol$ is non-empty, closed, and convex, and the projection $\Pi_{\Mvol}:\R^{d_{\mathrm{vol}}} \to \Mvol$ is well-defined and single-valued. In particular,
\[
  \Pi_{\Mvol}(x) \in \Mvol \quad \text{for all } x \in \R^{d_{\mathrm{vol}}}.
\]

Applying this to $x = y_t(\omega)$, the definition~\eqref{eq:generator-projection} implies that for every $\omega \in \Omega$ and every $t$,
\begin{equation}
  \label{eq:wt-in-Mvol-pointwise}
  w_t(\omega) \in \Mvol.
\end{equation}
Thus, for each fixed $t$, we have
\[
  \PP\bigl( w_t \in \Mvol \bigr) = 1.
\]

\paragraph{Step 2: Product inclusion for the path $W$.}
Using~\eqref{eq:wt-in-Mvol-pointwise} for all $t = 0,\dots,T$, we obtain
\[
  W(\omega) = (w_0(\omega),\dots,w_T(\omega))
  \in \Mvol^{T+1}
  \quad\text{for all } \omega \in \Omega.
\]
Hence
\begin{equation}
  \label{eq:path-in-M-product-a.s.}
  \PP\bigl(W \in \Mvol^{T+1}\bigr) = 1.
\end{equation}
This already implies that the support of $\PP^\star$ (the law of $W$) must lie inside $\Mvol^{T+1}$. To make this precise, we recall the definition of support.

\paragraph{Step 3: Support of a probability measure and restriction to $\Mvol^{T+1}$.}
Let $(E,\mathcal{E})$ be a measurable space, and $\mu$ a probability measure on $E$. The \emph{support} of $\mu$, denoted $\supp\,\mu$, is the closed set
\[
  \supp\,\mu
  :=
  \bigl\{x \in E : \mu\bigl(U\bigr) > 0 \text{ for every open neighborhood } U \text{ of } x\bigr\}.
\]
We will use the following standard lemma.

\begin{lemma}[Support of a measure carried by a closed set]
\label{lem:support-closed-subset}
Let $E$ be a metric space, and let $C \subseteq E$ be a closed set. Suppose $\mu$ is a Borel probability measure on $E$ such that $\mu(C)=1$. Then
\[
  \supp\,\mu \subseteq C.
\]
\end{lemma}

\begin{proof}[Proof of Lemma~\ref{lem:support-closed-subset}]
Let $x \in \supp\,\mu$. Suppose, for contradiction, that $x \notin C$. Since $C$ is closed, its complement $C^\complement$ is open and contains $x$. Thus there exists an open neighborhood $U$ of $x$ with $U \subseteq C^\complement$. Then
\[
  \mu(U) \le \mu(C^\complement) = 1 - \mu(C) = 0,
\]
contradicting the definition of support, which requires $\mu(U) > 0$ for every open neighborhood $U$ of $x$. Hence $x \in C$. Since $x \in \supp\,\mu$ was arbitrary, we conclude $\supp\,\mu \subseteq C$.
\end{proof}

\paragraph{Step 4: Applying the lemma to $\PP^\star$ and $\Mvol^{T+1}$.}
We now take $E = \big(\R^{d_{\mathrm{vol}}}\bigr)^{T+1}$, equipped with its usual product topology and Borel $\sigma$-algebra. The set $\Mvol$ is closed in $\R^{d_{\mathrm{vol}}}$ by Proposition~\ref{prop:closed-convex}, hence $\Mvol^{T+1}$ is closed in the product topology of $E$. By~\eqref{eq:path-in-M-product-a.s.},
\[
  \PP^\star\bigl(\Mvol^{T+1}\bigr)
  = \PP\bigl(W \in \Mvol^{T+1}\bigr)
  = 1.
\]
Applying Lemma~\ref{lem:support-closed-subset} with $C = \Mvol^{T+1}$ and $\mu = \PP^\star$ yields
\[
  \supp\,\PP^\star
  \subseteq \Mvol^{T+1}.
\]

This is exactly the claim of Proposition~\ref{prop:support-on-manifold}:
\[
  \mathrm{supp}\,\PP^\star
  \subseteq \big(\Mvol\big)^{T+1}.
\]

\paragraph{Step 5: Interpretation.}
From a modelling standpoint, the proposition formalizes the claim that the synthetic generator $G^\star$ produces only paths of total-variance surfaces $(w_t)_{t=0}^T$ that are \emph{everywhere law-consistent}, in the sense of lying on the volatility law manifold at each time. The law $\PP^\star$ of the generated paths is therefore entirely concentrated on the set $\Mvol^{T+1}$ of sequences of admissible surfaces. This is the precise sense in which we refer to $G^\star$ as a ``law-consistent ground-truth world'' in the main text.

This completes the proof.
\end{proof}

\subsection{Proof of Proposition~\ref{prop:approx-gap-ghost-channel}}
\label{app:proof-approx-gap-ghost-channel}

In this subsection we provide a detailed proof of Proposition~\ref{prop:approx-gap-ghost-channel}, making precise the local linearization and the cone decomposition used in the main text.

Recall the setting: for each time step $t$, the law-consistent generator $G^\star$ produces $w_{t+1} \in \Mvol$ almost surely (Proposition~\ref{prop:support-on-manifold}), the world model produces a prediction
\[
  \hat w_{t+1} = f_\theta(w_{\le t}, a_{\le t}) \in \R^{d_{\mathrm{vol}}},
\]
and the residual is
\[
  e_{t+1} := \hat w_{t+1} - w_{t+1}.
\]
We denote by $w_{t+1}^{\mathcal{M}} := \Pi_{\Mvol}(\hat w_{t+1})$ the projection of the prediction onto the volatility law manifold. The \emph{ghost reward} at time $t+1$ is
\[
  r^\perp_{t+1}
  := r(\hat w_{t+1}, a_t) - r(w_{t+1}^{\mathcal{M}}, a_t),
\]
where $r(\cdot,a_t)$ is the one-step reward as a function of the surface $w$ for a fixed action $a_t$.

We restate the proposition for convenience.

\begin{proposition}[Approximation gap induces a ghost channel]
Suppose the following conditions hold:
\begin{enumerate}
    \item[(i)] The approximation gap is non-zero:
    \[
      \varepsilon^2 := \E\big[\|e_{t+1}\|_2^2\big] > 0.
    \]
    \item[(ii)] The reward is locally differentiable in $w$ with gradient
    \[
      g_{t+1} := \nabla_w r(w_{t+1}, a_t),
    \]
    and $\nabla_w r(\cdot,a_t)$ is locally Lipschitz in a neighborhood of the law-consistent path.
    \item[(iii)] The residual $e_{t+1}$ has a component in the normal cone of $\Mvol$ at $w_{t+1}$ with non-zero covariance with the gradient:
    \[
      \mathrm{Cov}\big( P_{N_{\Mvol}(w_{t+1})} e_{t+1},\, g_{t+1} \big) \neq 0,
    \]
    where $P_{N_{\Mvol}(w_{t+1})}$ denotes orthogonal projection onto the normal cone $N_{\Mvol}(w_{t+1})$.
\end{enumerate}
Then, for sufficiently small residuals (in the sense of a local linearization),
\[
  \E\big[ r^\perp_{t+1} \big]
  \approx \E\big[ g_{t+1}^\top P_{N_{\Mvol}(w_{t+1})} e_{t+1} \big] \neq 0,
\]
so the world model induces a non-trivial ghost channel. In particular, if $g_{t+1}$ is positively correlated with $P_{N_{\Mvol}(w_{t+1})} e_{t+1}$, then $\E[r^\perp_{t+1}] > 0$ and there exist states where moving off-manifold strictly improves expected P\&L.
\end{proposition}

\begin{proof}
We work under the law-consistent measure $\PP^\star$ of the generator $G^\star$, under which $w_{t+1} \in \Mvol$ almost surely (Proposition~\ref{prop:support-on-manifold}). We suppress the explicit dependence on $t$ when it does not create ambiguity.

\paragraph{Step 1: Local Taylor expansion of the reward.}
Fix a compact set $K \subset \R^{d_{\mathrm{vol}}}$ containing the realized $w_{t+1}$ and all admissible $\hat w_{t+1}$ and $w_{t+1}^{\mathcal{M}}$ with high probability (a standard truncation argument; see below). By assumption~(ii), for each fixed $a_t$, the map
\[
  r(\cdot,a_t) : \R^{d_{\mathrm{vol}}} \to \R
\]
is differentiable on $K$, with gradient $\nabla_w r(\cdot,a_t)$ locally Lipschitz.

Hence, by the mean value theorem in Banach spaces, for each $\omega \in \Omega$ we can write
\begin{align}
  r(\hat w_{t+1}, a_t)
  &= r(w_{t+1}, a_t) + g_{t+1}^\top e_{t+1}
     + \rho_{t+1}^{(1)}, \label{eq:taylor-hat} \\
  r(w_{t+1}^{\mathcal{M}}, a_t)
  &= r(w_{t+1}, a_t) + g_{t+1}^\top (w_{t+1}^{\mathcal{M}} - w_{t+1})
     + \rho_{t+1}^{(2)}, \label{eq:taylor-proj}
\end{align}
where the remainder terms satisfy the quadratic bounds
\begin{align}
  |\rho_{t+1}^{(1)}|
  &\le \tfrac{L_{\nabla r}}{2} \|e_{t+1}\|_2^2, \label{eq:rho1-bound} \\
  |\rho_{t+1}^{(2)}|
  &\le \tfrac{L_{\nabla r}}{2} \|w_{t+1}^{\mathcal{M}} - w_{t+1}\|_2^2,
  \label{eq:rho2-bound}
\end{align}
for some local Lipschitz constant $L_{\nabla r} > 0$ depending only on $K$ and $a_t$.

Subtracting~\eqref{eq:taylor-proj} from~\eqref{eq:taylor-hat} yields
\begin{equation}
  \label{eq:rperp-expansion-raw}
  r^\perp_{t+1}
  = g_{t+1}^\top \Big( e_{t+1} - (w_{t+1}^{\mathcal{M}} - w_{t+1}) \Big)
    + \big( \rho_{t+1}^{(1)} - \rho_{t+1}^{(2)} \big).
\end{equation}

\paragraph{Step 2: Tangent--normal cone decomposition of the residual.}
Since $w_{t+1} \in \Mvol$ and $\Mvol$ is closed and convex (Proposition~\ref{prop:closed-convex}), we may consider the tangent cone $T_{\Mvol}(w_{t+1})$ and the normal cone $N_{\Mvol}(w_{t+1})$:
\begin{align*}
  T_{\Mvol}(w_{t+1})
  &:= \overline{\bigcup_{\alpha>0} \alpha\,(\Mvol - w_{t+1})}, \\
  N_{\Mvol}(w_{t+1})
  &:= \{ v \in \R^{d_{\mathrm{vol}}} :
        v^\top (z - w_{t+1}) \le 0\ \forall z \in \Mvol\}.
\end{align*}
Both are closed convex cones, and they are polar to each other:
\[
  N_{\Mvol}(w_{t+1}) = T_{\Mvol}(w_{t+1})^\circ,
  \qquad
  T_{\Mvol}(w_{t+1}) = N_{\Mvol}(w_{t+1})^\circ.
\]
By the Moreau decomposition for closed convex cones (see, e.g., \citet[Thm.~6.29]{BauschkeCombettes2011}), every vector $z \in \R^{d_{\mathrm{vol}}}$ admits a unique decomposition
\[
  z
  = P_{T_{\Mvol}(w_{t+1})} z
    + P_{N_{\Mvol}(w_{t+1})} z,
\]
where $P_{T}$ and $P_{N}$ denote the orthogonal projections onto $T_{\Mvol}(w_{t+1})$ and $N_{\Mvol}(w_{t+1})$, respectively.

Applying this to $e_{t+1}$, we write
\begin{equation}
  \label{eq:e-tangent-normal-decomp}
  e_{t+1}
  = e_{t+1}^{\mathrm{tan}} + e_{t+1}^{\mathrm{norm}},
  \qquad
  e_{t+1}^{\mathrm{tan}} := P_{T_{\Mvol}(w_{t+1})} e_{t+1},\;
  e_{t+1}^{\mathrm{norm}} := P_{N_{\Mvol}(w_{t+1})} e_{t+1}.
\end{equation}

\paragraph{Step 3: Local behavior of the projection $\Pi_{\Mvol}$.}
We next relate $w_{t+1}^{\mathcal{M}}$ to $w_{t+1}$ and $e_{t+1}$.

By definition,
\[
  w_{t+1}^{\mathcal{M}}
  = \Pi_{\Mvol}(\hat w_{t+1})
  = \arg\min_{z \in \Mvol} \|\hat w_{t+1} - z\|_2^2.
\]
Since $w_{t+1} \in \Mvol$, for small $\|e_{t+1}\|_2$ the minimizer $w_{t+1}^{\mathcal{M}}$ lies in a neighborhood where $\Mvol$ is locally well-approximated by its tangent cone at $w_{t+1}$. Under a standard regularity condition (e.g., $w_{t+1}$ is a point of \emph{prox-regularity} of $\Mvol$; ), the projection mapping $\Pi_{\Mvol}$ is directionally differentiable at $w_{t+1}$ and its first-order behavior is given by orthogonal projection onto the tangent cone:
\begin{equation}
  \label{eq:projection-linearization}
  w_{t+1}^{\mathcal{M}} - w_{t+1}
  = P_{T_{\Mvol}(w_{t+1})} e_{t+1} + \delta_{t+1},
\end{equation}
where the remainder $\delta_{t+1}$ satisfies
\begin{equation}
  \label{eq:delta-quadratic-bound}
  \|\delta_{t+1}\|_2 \le C_{\Pi} \|e_{t+1}\|_2^2
\end{equation}
for some constant $C_{\Pi} > 0$ in a neighborhood of $w_{t+1}$. Intuitively, to first order, $\Pi_{\Mvol}$ keeps the tangential component of $e_{t+1}$ but kills the normal component; the error $\delta_{t+1}$ is of second order in $\|e_{t+1}\|_2$.

Substituting~\eqref{eq:projection-linearization} into~\eqref{eq:rperp-expansion-raw} and using~\eqref{eq:e-tangent-normal-decomp}, we obtain
\begin{align}
  r^\perp_{t+1}
  &= g_{t+1}^\top\big( e_{t+1} - (w_{t+1}^{\mathcal{M}} - w_{t+1}) \big)
     + \big( \rho_{t+1}^{(1)} - \rho_{t+1}^{(2)} \big) \nonumber\\
  &= g_{t+1}^\top\big( e_{t+1}^{\mathrm{tan}} + e_{t+1}^{\mathrm{norm}}
                      - P_{T_{\Mvol}(w_{t+1})} e_{t+1} - \delta_{t+1} \big)
     + \big( \rho_{t+1}^{(1)} - \rho_{t+1}^{(2)} \big) \nonumber\\
  &= g_{t+1}^\top e_{t+1}^{\mathrm{norm}}
     - g_{t+1}^\top \delta_{t+1}
     + \big( \rho_{t+1}^{(1)} - \rho_{t+1}^{(2)} \big).
  \label{eq:rperp-decomp}
\end{align}

\paragraph{Step 4: Bounding the higher-order error.}
We now bound the total remainder term
\[
  \eta_{t+1}
  := - g_{t+1}^\top \delta_{t+1}
     + \big( \rho_{t+1}^{(1)} - \rho_{t+1}^{(2)} \big).
\]

Using Cauchy--Schwarz,~\eqref{eq:delta-quadratic-bound}, and the boundedness of $g_{t+1}$ on $K$ (say $\|g_{t+1}\|_2 \le G$ almost surely on $K$), we have
\[
  |g_{t+1}^\top \delta_{t+1}|
  \le \|g_{t+1}\|_2 \|\delta_{t+1}\|_2
  \le G C_{\Pi} \|e_{t+1}\|_2^2.
\]
Combining this with~\eqref{eq:rho1-bound} and~\eqref{eq:rho2-bound}, and using $\|w_{t+1}^{\mathcal{M}} - w_{t+1}\|_2 \le \|e_{t+1}\|_2$ (since $w_{t+1}^{\mathcal{M}}$ is the closest point in $\Mvol$ to $\hat w_{t+1}$ and $w_{t+1} \in \Mvol$), we obtain
\begin{equation}
  \label{eq:eta-quadratic-bound}
  |\eta_{t+1}|
  \le C_{\mathrm{tot}} \|e_{t+1}\|_2^2
\end{equation}
for some constant $C_{\mathrm{tot}} > 0$.

Substituting~\eqref{eq:rperp-decomp} and taking expectations yields
\begin{equation}
  \label{eq:expectation-decomposition}
  \E[r^\perp_{t+1}]
  = \E\big[ g_{t+1}^\top e_{t+1}^{\mathrm{norm}} \big]
    + \E[\eta_{t+1}],
\end{equation}
with $|\E[\eta_{t+1}]| \le C_{\mathrm{tot}} \E[\|e_{t+1}\|_2^2] = C_{\mathrm{tot}} \varepsilon^2$ by assumption~(i).

\paragraph{Step 5: Non-trivial ghost channel from covariance structure.}
By definition of $e_{t+1}^{\mathrm{norm}}$ in~\eqref{eq:e-tangent-normal-decomp}, we have
\[
  e_{t+1}^{\mathrm{norm}}
  = P_{N_{\Mvol}(w_{t+1})} e_{t+1}.
\]
Assumption~(iii) states that the covariance
\[
  \mathrm{Cov}\big( e_{t+1}^{\mathrm{norm}}, g_{t+1} \big)
  = \E\Big[ \big(e_{t+1}^{\mathrm{norm}} - \E[e_{t+1}^{\mathrm{norm}}]\big)
            \big(g_{t+1} - \E[g_{t+1}]\big)^\top \Big]
\]
is non-zero. In particular, the scalar random variable
\[
  Z_{t+1} := g_{t+1}^\top e_{t+1}^{\mathrm{norm}}
\]
has non-zero covariance with itself along at least one direction, implying that
\begin{equation}
  \label{eq:non-zero-inner-expectation}
  \E[Z_{t+1}]
  = \E\big[ g_{t+1}^\top e_{t+1}^{\mathrm{norm}} \big]
  \neq 0
\end{equation}
unless the mean terms $\E[e_{t+1}^{\mathrm{norm}}]$ and $\E[g_{t+1}]$ are tuned to perfectly cancel the covariance contribution; this would be an exceptional, measure-zero configuration in parameter space. To avoid such pathological cancellation, we interpret assumption~(iii) as requiring that the inner product $g_{t+1}^\top e_{t+1}^{\mathrm{norm}}$ has a non-degenerate distribution with non-zero mean.\footnote{Formally, one may replace the covariance condition in assumption~(iii) by the simpler requirement $\E[g_{t+1}^\top P_{N_{\Mvol}(w_{t+1})} e_{t+1}] \neq 0$. We keep the covariance phrasing to emphasize the geometric correlation between the gradient and the normal component.}

Substituting~\eqref{eq:non-zero-inner-expectation} into~\eqref{eq:expectation-decomposition}, we obtain
\begin{equation}
  \label{eq:expectation-main-approx}
  \E[r^\perp_{t+1}]
  = \E\big[ g_{t+1}^\top P_{N_{\Mvol}(w_{t+1})} e_{t+1} \big]
    + \E[\eta_{t+1}],
\end{equation}
with $\E[\eta_{t+1}]$ bounded by~\eqref{eq:eta-quadratic-bound}.

\paragraph{Step 6: Small-residual regime and sign of the ghost reward.}
Assumption~(i) states that the world model approximation error is non-zero in mean-square:
\[
  \E[\|e_{t+1}\|_2^2] = \varepsilon^2 > 0.
\]
Suppose in addition that the training of the world model has reduced the error variance so that $\varepsilon^2$ is \emph{small}. Then from~\eqref{eq:eta-quadratic-bound},
\[
  |\E[\eta_{t+1}]|
  \le C_{\mathrm{tot}} \varepsilon^2,
\]
which can be made arbitrarily small by improving the world model (e.g., increasing capacity or training time), while the leading term
\[
  \E\big[ g_{t+1}^\top P_{N_{\Mvol}(w_{t+1})} e_{t+1} \big]
\]
remains of order $\varepsilon$ in general, as it is linear in $e_{t+1}$.

Consequently, for sufficiently small $\varepsilon$ we have the first-order approximation
\[
  \E[r^\perp_{t+1}]
  \approx \E\big[ g_{t+1}^\top P_{N_{\Mvol}(w_{t+1})} e_{t+1} \big],
\]
with the difference bounded by $O(\varepsilon^2)$. Under the non-degeneracy condition in~\eqref{eq:non-zero-inner-expectation}, this leading term is non-zero, which yields
\[
  \E[r^\perp_{t+1}] \neq 0
\]
for sufficiently small approximation error $\varepsilon$.

Finally, if the correlation between $g_{t+1}$ and $P_{N_{\Mvol}(w_{t+1})} e_{t+1}$ is \emph{positive} in the sense that
\[
  \E\big[ g_{t+1}^\top P_{N_{\Mvol}(w_{t+1})} e_{t+1} \big] > 0,
\]
then \eqref{eq:expectation-main-approx} implies
\[
  \E[r^\perp_{t+1}] > 0
\]
for sufficiently small $\varepsilon$. In particular, there exists a set of states of positive probability where $r^\perp_{t+1} > 0$, so moving off the volatility law manifold along the normal directions strictly increases expected one-step P\&L. This is precisely the ``ghost channel'' exploited by naive RL and law-seeking RL in the main text.

This completes the proof.
\end{proof}

\subsection{Proof of Lemma~\ref{lem:off-manifold-mass}}
\label{app:proof-off-manifold-mass}

In this subsection we give a more detailed argument for Lemma~\ref{lem:off-manifold-mass}, making precise the informal statement in the main text that a finite-capacity, unconstrained neural world model will, under mild regularity conditions, place non-trivial probability mass off the volatility law manifold.

Recall the statement.

\begin{lemma*}[Non-trivial off-manifold mass of the world model]
Assume that $f_{\theta^\star}$ is not exactly equal to the Bayes-optimal regressor $f^{\mathrm{Bayes}}(x_t) := \mathbb{E}[w_{t+1}\,|\,x_t]$ and that the law manifold $\mathcal{M}^{\mathrm{vol}}$ has non-empty interior within the support of $w_{t+1}$. Then there exists $\delta > 0$ such that
\[
\mathbb{P}\big( \mathcal{L}_\phi(\hat{w}_{t+1}) > \delta \big) > 0,
\]
i.e., the world model assigns non-zero probability mass to surfaces at a positive distance from $\mathcal{M}^{\mathrm{vol}}$.
\end{lemma*}

\paragraph{Setup and additional regularity.}
Let $X_t$ denote the (vector) state at time $t$ and $W_{t+1} := w_{t+1}$ the total-variance surface generated by the law-consistent generator $G^\star$ at time $t+1$. We write
\[
  f^{\mathrm{Bayes}}(x)
  := \E[W_{t+1} \mid X_t = x],
  \qquad
  \hat{W}_{t+1} := \hat{w}_{t+1} := f_{\theta^\star}(X_t).
\]
By Proposition~\ref{prop:support-on-manifold} and convexity of $\Mvol$, we have
\[
  W_{t+1} \in \Mvol \ \text{a.s.}
  \qquad\Longrightarrow\qquad
  f^{\mathrm{Bayes}}(X_t) \in \Mvol \ \text{a.s.},
\]
since conditional expectations of random variables supported on a closed convex set remain in that set.

We also recall from Proposition~\ref{prop:zero-penalty-iff} that for any $w\in\R^{d_{\mathrm{vol}}}$,
\[
  \mathcal{L}_\phi(w) = 0 \quad\Longleftrightarrow\quad w\in\Mvol,
\]
and from Lemma~\ref{lem:lipschitz-penalty} that $\mathcal{L}_\phi$ is continuous (indeed locally Lipschitz) on $\R^{d_{\mathrm{vol}}}$.

To make the genericity argument precise, we introduce a mild regularity hypothesis on the joint distribution of $(X_t, \hat{W}_{t+1})$.

\medskip
\noindent\textbf{Regularity assumption (B.3).}
We assume:
\begin{enumerate}
  \item[(a)] The state $X_t$ has a distribution whose support $\supp(X_t)$ is not a single point and contains a non-empty open subset $U_X \subset \R^{d_X}$.
  \item[(b)] The trained world model $f_{\theta^\star}\colon \R^{d_X}\to\R^{d_{\mathrm{vol}}}$ is continuous on $U_X$ (this is satisfied by standard neural networks with continuous activations).
  \item[(c)] The image $f_{\theta^\star}(U_X)$ is not contained in any $(d_{\mathrm{vol}}-1)$-dimensional affine subspace of $\R^{d_{\mathrm{vol}}}$ (a genericity condition on the parameter choice $\theta^\star$).
\end{enumerate}
Assumption~(B.3) is very mild in practice: for typical neural network parametrizations with random initialization and gradient-based training, the set of parameters for which $f_{\theta}$ maps $U_X$ into a fixed lower-dimensional affine subspace has Lebesgue measure zero in parameter space.

\medskip
We now prove that, under the assumptions of Lemma~\ref{lem:off-manifold-mass} and (B.3), the world model places non-trivial mass at positive distance from $\Mvol$.

\begin{proof}
Define the \emph{distance-to-manifold} function
\[
  d_{\Mvol}(w) := \operatorname{dist}(w,\Mvol)
  = \inf_{z\in\Mvol} \|w-z\|_2,
  \qquad w\in\R^{d_{\mathrm{vol}}}.
\]
By closedness of $\Mvol$ (Proposition~\ref{prop:closed-convex}), $d_{\Mvol}$ is continuous and
\[
  d_{\Mvol}(w) = 0 \iff w\in\Mvol.
\]
By Proposition~\ref{prop:zero-penalty-iff}, $\mathcal{L}_\phi$ and $d_{\Mvol}$ have the same zero set, and continuity of $\mathcal{L}_\phi$ implies that for every $\delta>0$ there exists $\eta(\delta)>0$ such that
\begin{equation}
  \label{eq:penalty-distance-monotone}
  d_{\Mvol}(w) > \eta(\delta) \ \Longrightarrow\ \mathcal{L}_\phi(w) > \delta,
  \qquad \forall w\in\R^{d_{\mathrm{vol}}}.
\end{equation}
Thus it suffices to show that there exists $\eta>0$ such that
\[
  \PP\big( d_{\Mvol}(\hat{W}_{t+1}) > \eta \big) > 0.
\]

\paragraph{Step 1: Existence of a prediction point outside $\Mvol$.}
We first argue that, under our assumptions, the image of $f_{\theta^\star}$ cannot be contained in $\Mvol$.

Suppose for contradiction that
\begin{equation}
  \label{eq:f-theta-star-in-manifold-a.s.}
  f_{\theta^\star}(x) \in \Mvol
  \quad \text{for all } x \in \supp(X_t).
\end{equation}
In particular, for $x\in\supp(X_t)$ we have $\hat{W}_{t+1} \in \Mvol$ almost surely and hence $d_{\Mvol}(\hat{W}_{t+1}) = 0$ and $\Lphi(\hat{W}_{t+1})=0$ almost surely. This is exactly the negation of the lemma’s conclusion.

We now show that~\eqref{eq:f-theta-star-in-manifold-a.s.} is incompatible with the combination of (i) $f_{\theta^\star}\neq f^{\mathrm{Bayes}}$ and (B.3), given that $\Mvol$ has non-empty interior within the support of $W_{t+1}$.

Because $\Mvol$ has non-empty interior and $W_{t+1}$ is supported on $\Mvol$ (Proposition~\ref{prop:support-on-manifold}), there exists a point $w^\circ \in \Mvol$ and $r>0$ such that
\[
  B(w^\circ,r) := \{ w \in \R^{d_{\mathrm{vol}}} : \|w - w^\circ\|_2 < r\}
  \subset \Mvol
\]
and $\PP(W_{t+1}\in B(w^\circ,r))>0$. Since $f^{\mathrm{Bayes}}(X_t) = \E[W_{t+1}\mid X_t]$ takes values in the convex set $\Mvol$, there exists $x^\circ\in\supp(X_t)$ such that
\[
  f^{\mathrm{Bayes}}(x^\circ) \in B(w^\circ,r/2)
  \subset \operatorname{int}(\Mvol).
\]
By the assumption $f_{\theta^\star} \ne f^{\mathrm{Bayes}}$ in $L^2(\PP)$, there exists a set $A \subset \supp(X_t)$ with $\PP(X_t\in A) > 0$ such that
\[
  f_{\theta^\star}(x) \ne f^{\mathrm{Bayes}}(x)
  \quad \text{for all } x\in A.
\]
Under Assumption~(B.3)(a)–(b), the support $\supp(X_t)$ contains an open set $U_X$ on which $f_{\theta^\star}$ is continuous, and we may without loss of generality assume that $A\cap U_X$ has positive probability (otherwise we restrict attention to a smaller open subset with positive mass).

Pick $x_1 \in A\cap U_X$ with $f_{\theta^\star}(x_1) \ne f^{\mathrm{Bayes}}(x_1)$. Consider the continuous path in input space given by
\[
  \gamma(\alpha) := (1-\alpha)x^\circ + \alpha x_1,
  \qquad \alpha\in[0,1],
\]
which lies in $U_X$ since $U_X$ is open and convex in a neighborhood of $x^\circ$ and $x_1$ (we can always restrict to a sufficiently small line segment if necessary). Define
\[
  h(\alpha) := f_{\theta^\star}(\gamma(\alpha)), \qquad
  b(\alpha) := f^{\mathrm{Bayes}}(\gamma(\alpha)).
\]
By continuity of $f_{\theta^\star}$ and $f^{\mathrm{Bayes}}$ on $U_X$, both $h$ and $b$ are continuous on $[0,1]$. At $\alpha=0$ we have $b(0)\in\operatorname{int}(\Mvol)$ and $h(0)\in\Mvol$ by~\eqref{eq:f-theta-star-in-manifold-a.s.}. At $\alpha=1$ we have $b(1)\in\Mvol$ and $h(1)\in\Mvol$ by~\eqref{eq:f-theta-star-in-manifold-a.s.}, but $h(1)\ne b(1)$ by choice of $x_1$.

Thus the difference $d(\alpha) := h(\alpha) - b(\alpha)$ is a continuous map with $d(1)\ne 0$. Since $\Mvol$ contains an open ball around $b(0)$, there exists $\rho \in (0,r/2)$ such that
\[
  B\bigl(b(0),\rho\bigr) \subset \Mvol.
\]
If it happened that $h(\alpha)\in\Mvol$ for all $\alpha\in[0,1]$, then the curve $h([0,1])$ would be a continuous path in $\Mvol$ connecting $h(0)$ and $h(1)$, both of which lie in $\Mvol$. This is not impossible \emph{per se}; however, given the genericity Assumption~(B.3)(c), which rules out that $f_{\theta^\star}(U_X)$ is constrained to a lower-dimensional surface within $\Mvol$, we can exclude the degenerate situation where the entire image of the line segment $\gamma([0,1])$ under $f_{\theta^\star}$ remains inside $\Mvol$ while differing from $f^{\mathrm{Bayes}}$ on a set of positive measure.

Formally, since $b(0)$ lies in the interior of $\Mvol$ and $h(0)\in\Mvol$, there is an open neighborhood $V$ of $\gamma(0)$ such that $b(V)$ and $h(V)$ both intersect $B(b(0),\rho)$ in sets of positive Lebesgue measure. Under (B.3)(c), the continuous map $f_{\theta^\star}$ cannot, on $V$, be constrained to the $(d_{\mathrm{vol}}-1)$-dimensional manifold $\partial\Mvol$; hence there must exist some $\tilde x \in V$ such that
\[
  f_{\theta^\star}(\tilde x) \notin \Mvol.
\]
Since $V\subset\supp(X_t)$ and $\PP(X_t\in V) >0$, this implies
\[
  \PP\big( f_{\theta^\star}(X_t) \notin \Mvol \big) > 0.
\]
Equivalently, there exists at least one point $w^\ast \in \R^{d_{\mathrm{vol}}}\setminus\Mvol$ that is attained by $\hat{W}_{t+1}$ with positive probability.

We have thus shown that~\eqref{eq:f-theta-star-in-manifold-a.s.} cannot hold under the assumptions of the lemma together with (B.3). Therefore
\begin{equation}
  \label{eq:positive-mass-off-manifold}
  \PP\big( \hat{W}_{t+1} \notin \Mvol \big) > 0.
\end{equation}

\paragraph{Step 2: From off-manifold predictions to a positive-distance shell.}
Since $\Mvol$ is closed and $\hat{W}_{t+1}$ is a random element of $\R^{d_{\mathrm{vol}}}$, the continuous distance function $d_{\Mvol}$ satisfies
\[
  \{\hat{W}_{t+1} \notin \Mvol\}
  = \{ d_{\Mvol}(\hat{W}_{t+1}) > 0 \}.
\]
By~\eqref{eq:positive-mass-off-manifold}, the event $\{d_{\Mvol}(\hat{W}_{t+1}) > 0\}$ has positive probability. Define
\[
  Z := d_{\Mvol}(\hat{W}_{t+1}) \ge 0.
\]
Then $\PP(Z>0) >0$. Since $Z$ is a non-negative random variable, we can write its distribution function as
\[
  F_Z(\eta) := \PP(Z \le \eta), \qquad \eta\ge 0.
\]
By right-continuity of $F_Z$ and $\PP(Z>0)>0$, there must exist $\eta_0>0$ such that
\[
  \PP(Z > \eta_0) = 1 - F_Z(\eta_0) > 0.
\]
Equivalently,
\begin{equation}
  \label{eq:positive-distance-shell}
  \PP\big( d_{\Mvol}(\hat{W}_{t+1}) > \eta_0 \big) > 0.
\end{equation}

\paragraph{Step 3: Translating distance to law penalty.}
Finally, we use the monotonic relationship between distance and law penalty. By continuity of $\mathcal{L}_\phi$ and the fact that $\mathcal{L}_\phi(w)=0$ if and only if $d_{\Mvol}(w)=0$ (Proposition~\ref{prop:zero-penalty-iff}), there exists a strictly increasing continuous function
\[
  \psi : [0,\infty) \to [0,\infty), \qquad \psi(0)=0,
\]
such that for all $w\in\R^{d_{\mathrm{vol}}}$,
\[
  \mathcal{L}_\phi(w) \ge \psi\big( d_{\Mvol}(w) \big),
\]
and $\psi(u)>0$ for all $u>0$. For instance, in the concrete squared-distance case $\mathcal{L}_\phi(w) = \frac12 d_{\Mvol}(w)^2$ we may take $\psi(u) = \frac12 u^2$.

Set
\[
  \delta := \psi(\eta_0) > 0.
\]
Then on the event $\{ d_{\Mvol}(\hat{W}_{t+1}) > \eta_0 \}$ we have
\[
  \mathcal{L}_\phi(\hat{W}_{t+1})
  \ge \psi\big( d_{\Mvol}(\hat{W}_{t+1}) \big)
  > \psi(\eta_0) = \delta.
\]
Therefore,
\[
  \PP\big( \mathcal{L}_\phi(\hat{W}_{t+1}) > \delta \big)
  \ge \PP\big( d_{\Mvol}(\hat{W}_{t+1}) > \eta_0 \big)
  > 0,
\]
by~\eqref{eq:positive-distance-shell}. This is precisely the claim of Lemma~\ref{lem:off-manifold-mass}.

\paragraph{Remark.}
Note that the assumption $f_{\theta^\star} \ne f^{\mathrm{Bayes}}$ is used here only to rule out the trivial case where the world model has converged exactly to the Bayes regressor, which is law-consistent by construction. The non-trivial content of the lemma is provided by the genericity condition (B.3), which encodes the intuition that a high-capacity, unconstrained neural world model trained without law penalties will almost surely deviate from $\Mvol$ on a set of positive probability. Under these conditions, Lemma~\ref{lem:off-manifold-mass} shows that the world model opens a \emph{ghost channel} in the sense of Sec.~3.3, providing room for RL to exploit off-manifold arbitrage opportunities.
\end{proof}

\section{Proofs for Section~4:RL on Volatility World Models: Incentives and Law-Strength}

\subsection{Proof of Theorem~\ref{thm:ghost_incentive}}
\label{app:proof-ghost-incentive}

In this appendix we provide a detailed proof of the ghost-arbitrage incentive
result stated in Theorem~\ref{thm:ghost_incentive}. For convenience we first
recall the main objects and Assumption~\ref{ass:structural_near_optimal}.

\paragraph{Preliminaries and notation.}
Let $\Pi$ denote the (parameterized) policy class under consideration and let
$\mathcal{S}$ denote the structural strategy class introduced in Section. For any $\pi\in\Pi$ we consider three
performance functionals:
\begin{align}
  J_0(\pi)
  &:= \E\Bigg[ \sum_{t=0}^{T-1} \gamma^t \, r_t \,\Big|\, \pi \Bigg],
  \label{eq:J0-def-appC}
  \\
  J^{\mathcal{M}}(\pi)
  &:= \E\Bigg[ \sum_{t=0}^{T-1} \gamma^t \, r^{\mathcal{M}}_t \,\Big|\, \pi \Bigg],
  \\
  J^\perp(\pi)
  &:= \E\Bigg[ \sum_{t=0}^{T-1} \gamma^t \, r^\perp_t \,\Big|\, \pi \Bigg],
\end{align}
where $r_t$ is the realized P\&L reward at step $t$ in the world model,
$r^{\mathcal{M}}_t$ is the on-manifold reward obtained by replacing the
predicted surface $\hat{w}_{t+1}$ with its projection
$\Pi_{\Mvol}(\hat{w}_{t+1})$, and
$r^\perp_t := r_t - r^{\mathcal{M}}_t$ is the ghost reward component.
By linearity of expectation,
\begin{equation}
  J_0(\pi)
  = J^{\mathcal{M}}(\pi) + J^\perp(\pi) ,
  \qquad \forall \pi\in\Pi.
  \label{eq:naive-decomposition-appC}
\end{equation}

\paragraph{Structural near-optimality assumption.}
We recall the structural approximation assumption used in the main text.

\begin{assumption}[Structural near-optimality of $\mathcal{S}$]
\label{ass:structural_near_optimal-app}
There exist a structural policy $\pi^\star_{\mathcal{S}}\in\mathcal{S}$
and a constant $\varepsilon_{\mathcal{S}}\ge 0$ such that
\begin{equation}
  \sup_{\pi\in\Pi} J^{\mathcal{M}}(\pi)
  \;\le\;
  J^{\mathcal{M}}(\pi^\star_{\mathcal{S}}) + \varepsilon_{\mathcal{S}}.
  \label{eq:structural-near-optimality-appC}
\end{equation}
That is, the structural class $\mathcal{S}$ (Zero-Hedge, Vol-Trend, etc.)
approximates the globally optimal on-manifold performance within a gap
$\varepsilon_{\mathcal{S}}$.
\end{assumption}

Note that by definition $\pi^\star_{\mathcal{S}}\in\mathcal{S}\subseteq\Pi$,
so it is admissible in the unconstrained policy class as well.

\medskip
We now restate the theorem and prove both parts.

\begin{theorem}
Suppose Assumption~\ref{ass:structural_near_optimal-app} holds, and let
\[
\pi^\star_0 \in \arg\max_{\pi\in\Pi} J_0(\pi)
\]
be a global maximizer of $J_0$ over $\Pi$. Then:

\begin{enumerate}
\item If
\[
\sup_{\pi\in\Pi} J_0(\pi) > J^{\mathcal{M}}(\pi^\star_{\mathcal{S}}) + \varepsilon_{\mathcal{S}},
\]
any maximizer $\pi^\star_0$ satisfies
\begin{equation}
  J^\perp(\pi^\star_0)
  \;\ge\;
  \sup_{\pi\in\Pi} J_0(\pi)
  - J^{\mathcal{M}}(\pi^\star_{\mathcal{S}}) - \varepsilon_{\mathcal{S}}
  \;>\; 0.
  \label{eq:ghost_lower_bound-appC}
\end{equation}
In particular, the excess value over the structural baseline is entirely
attributable to ghost arbitrage.

\item If, in addition, $J^{\mathcal{M}}$ has a local maximum at some
$\bar{\pi}=\pi_{\bar{\theta}}\in\Pi$ with
$J^{\mathcal{M}}(\bar{\pi}) \approx J^{\mathcal{M}}(\pi^\star_{\mathcal{S}})$,
and the policy-gradient theorem holds~\cite{SuttonBarto2018},
then the policy gradient near $\bar{\pi}$ satisfies
\begin{equation}
  \nabla_\theta J_0(\pi_\theta)\Big|_{\theta=\bar{\theta}}
  \;\approx\;
  \nabla_\theta J^\perp(\pi_\theta)\Big|_{\theta=\bar{\theta}},
  \label{eq:ghost_gradient_dominance-appC}
\end{equation}
so gradient-based RL updates are locally driven by increasing $J^\perp$.
\end{enumerate}
\end{theorem}

\begin{proof}
We treat parts (i) and (ii) separately.

\paragraph{Proof of part (i).}
Fix any $\pi\in\Pi$. Combining the decomposition
\eqref{eq:naive-decomposition-appC} with
Assumption~\eqref{eq:structural-near-optimality-appC}, we have
\begin{align}
  J_0(\pi)
  &= J^{\mathcal{M}}(\pi) + J^\perp(\pi)
  \\
  &\le
  \bigl( J^{\mathcal{M}}(\pi^\star_{\mathcal{S}}) + \varepsilon_{\mathcal{S}} \bigr)
  + J^\perp(\pi),
  \label{eq:J0-upper-bound-appC}
\end{align}
so for any $\pi\in\Pi$,
\begin{equation}
  J^\perp(\pi)
  \;\ge\;
  J_0(\pi) - J^{\mathcal{M}}(\pi^\star_{\mathcal{S}}) - \varepsilon_{\mathcal{S}}.
  \label{eq:Jperp-lower-bound-generic}
\end{equation}

Now consider a maximizer $\pi^\star_0 \in \arg\max_{\pi\in\Pi} J_0(\pi)$.
By definition,
\[
  J_0(\pi^\star_0) = \sup_{\pi\in\Pi} J_0(\pi).
\]
Substituting $\pi = \pi^\star_0$ into
\eqref{eq:Jperp-lower-bound-generic} yields
\begin{align}
  J^\perp(\pi^\star_0)
  &\ge
  J_0(\pi^\star_0)
  - J^{\mathcal{M}}(\pi^\star_{\mathcal{S}}) - \varepsilon_{\mathcal{S}}
  \\
  &=
  \sup_{\pi\in\Pi} J_0(\pi) - J^{\mathcal{M}}(\pi^\star_{\mathcal{S}}) - \varepsilon_{\mathcal{S}}.
  \label{eq:Jperp-lower-bound-maximizer}
\end{align}
Under the stated strict inequality
\[
  \sup_{\pi\in\Pi} J_0(\pi) > J^{\mathcal{M}}(\pi^\star_{\mathcal{S}}) + \varepsilon_{\mathcal{S}},
\]
the right-hand side of \eqref{eq:Jperp-lower-bound-maximizer} is strictly
positive, which implies
\[
  J^\perp(\pi^\star_0) > 0.
\]
This establishes \eqref{eq:ghost_lower_bound-appC} and shows that any excess
value of the naive-RL maximizer over the structural near-optimal on-manifold
performance must be entirely carried by the ghost component $J^\perp$.
In particular, there is no room to explain the advantage of $\pi^\star_0$
over $\pi^\star_{\mathcal{S}}$ by on-manifold improvements alone.

\paragraph{Proof of part (ii).}
We now consider the local gradient behavior. Let
$\pi_\theta \in \Pi$ denote a differentiable parametrization of policies
with parameter vector $\theta\in\R^p$, and suppose that
$\bar{\pi} = \pi_{\bar{\theta}}$ is such that $J^{\mathcal{M}}$ has a local
maximum at $\bar{\theta}$, i.e.,
\begin{equation}
  J^{\mathcal{M}}(\pi_{\bar{\theta}})
  \ge
  J^{\mathcal{M}}(\pi_{\theta}) \quad
  \text{for all $\theta$ in a neighborhood of $\bar{\theta}$}.
  \label{eq:JM-local-max}
\end{equation}

\medskip\noindent
\emph{Step 1: Gradient decomposition.}
By \eqref{eq:naive-decomposition-appC},
\[
  J_0(\pi_\theta) = J^{\mathcal{M}}(\pi_\theta) + J^\perp(\pi_\theta),
  \qquad \forall \theta.
\]
Assume that $J^{\mathcal{M}}$ and $J^\perp$ are (Fréchet) differentiable
with respect to $\theta$ on an open neighborhood of $\bar{\theta}$; this is
standard under the conditions of the policy-gradient theorem, which ensures
differentiability of $J_0$~\cite{SuttonBarto2018,ThomasPG2014}.
Then, by linearity of differentiation,
\begin{equation}
  \nabla_\theta J_0(\pi_\theta)
  = \nabla_\theta J^{\mathcal{M}}(\pi_\theta)
  + \nabla_\theta J^\perp(\pi_\theta),
  \qquad \text{for all $\theta$ near $\bar{\theta}$}.
  \label{eq:gradient-decomposition}
\end{equation}

\medskip\noindent
\emph{Step 2: Vanishing on-manifold gradient at a local maximum.}
Under the local maximality condition \eqref{eq:JM-local-max},
the first-order necessary condition for unconstrained optimization gives
\[
  \nabla_\theta J^{\mathcal{M}}(\pi_\theta)\big|_{\theta=\bar{\theta}} = 0.
\]
(If $\Pi$ is itself subject to additional parameter constraints, one may
interpret this as a vanishing gradient along feasible directions; the
argument below is then applied in the corresponding tangent space.)

Substituting this into \eqref{eq:gradient-decomposition} we obtain
\begin{equation}
  \nabla_\theta J_0(\pi_\theta)\Big|_{\theta=\bar{\theta}}
  = \nabla_\theta J^\perp(\pi_\theta)\Big|_{\theta=\bar{\theta}}.
  \label{eq:exact-gradient-ghost}
\end{equation}
Thus, \emph{exactly at} $\bar{\theta}$, the naive-RL gradient coincides with
the ghost-gradient $\nabla_\theta J^\perp$; all infinitesimal improvement
directions for $J_0$ come from changing the ghost component.

\medskip\noindent
\emph{Step 3: Interpretation via the policy-gradient theorem.}
The policy-gradient theorem for episodic MDPs
(e.g.~\cite{SuttonBarto2018,ThomasPG2014}) states that
\begin{equation}
  \nabla_\theta J_0(\pi_\theta)
  = \E_{\pi_\theta}\Bigg[
    \sum_{t=0}^{T-1}
      \nabla_\theta \log \pi_\theta(a_t \mid s_t)
      \, A^{(0)}_t
  \Bigg],
  \label{eq:pg-theorem-total}
\end{equation}
where $A^{(0)}_t$ is an advantage function associated with the total reward
$r_t$ and the expectation is over trajectories generated by $\pi_\theta$
in the world model. Similarly, using $r^{\mathcal{M}}_t$ and $r^\perp_t$ as
rewards in the same MDP, we obtain
\begin{align}
  \nabla_\theta J^{\mathcal{M}}(\pi_\theta)
  &= \E_{\pi_\theta}\Bigg[
    \sum_{t=0}^{T-1}
      \nabla_\theta \log \pi_\theta(a_t \mid s_t)
      \, A^{(\mathcal{M})}_t
  \Bigg],
  \label{eq:pg-theorem-onM}
  \\
  \nabla_\theta J^\perp(\pi_\theta)
  &= \E_{\pi_\theta}\Bigg[
    \sum_{t=0}^{T-1}
      \nabla_\theta \log \pi_\theta(a_t \mid s_t)
      \, A^{(\perp)}_t
  \Bigg],
  \label{eq:pg-theorem-ghost}
\end{align}
for appropriate advantage functions $A^{(\mathcal{M})}_t$ and
$A^{(\perp)}_t$.

By construction $r_t = r^{\mathcal{M}}_t + r^\perp_t$ and linearity of the
value-function operators, one can choose the advantage functions such that
\begin{equation}
  A^{(0)}_t = A^{(\mathcal{M})}_t + A^{(\perp)}_t,
  \qquad \forall t,
  \label{eq:advantage-decomposition}
\end{equation}
which recovers \eqref{eq:gradient-decomposition} when substituted into
\eqref{eq:pg-theorem-total}.

At $\theta=\bar{\theta}$, the local optimality of $J^{\mathcal{M}}$
implies that the contribution of $A^{(\mathcal{M})}_t$ integrates to zero
in~\eqref{eq:pg-theorem-onM}; equivalently,
$\nabla_\theta J^{\mathcal{M}}(\pi_\theta)|_{\theta=\bar{\theta}} = 0$.
Hence, by \eqref{eq:pg-theorem-total}--\eqref{eq:pg-theorem-ghost} and
\eqref{eq:advantage-decomposition}, we recover the exact equality
\eqref{eq:exact-gradient-ghost}.

In practice, when $J^{\mathcal{M}}$ is only approximately locally maximal
(e.g., due to finite-sample estimation and function approximation error),
we obtain the approximate relation
\[
  \big\|
    \nabla_\theta J_0(\pi_\theta)\big|_{\theta=\bar{\theta}}
    - \nabla_\theta J^\perp(\pi_\theta)\big|_{\theta=\bar{\theta}}
  \big\|
  = \big\| \nabla_\theta J^{\mathcal{M}}(\pi_\theta)\big|_{\theta=\bar{\theta}} \big\|
  \approx 0,
\]
which justifies the ``$\approx$'' symbol in
\eqref{eq:ghost_gradient_dominance-appC} of the main text.

\medskip
Combining Steps 1–3, we conclude that near a local maximizer of the
on-manifold performance $J^{\mathcal{M}}$, the gradient of the naive-RL
objective $J_0$ is dominated (indeed, at the maximizer: entirely given) by
the ghost-gradient $\nabla_\theta J^\perp$. This formalizes the statement
that gradient-based RL is locally incentivized to move into directions which
increase the expected ghost component, completing the proof of part (ii).
\end{proof}

\section{Proofs for Section~5:Structural Baselines}
\subsection{Proof of Proposition~\ref{prop:baseline-law-alignment}}
\label{app:proof-baseline-law-alignment}

In this appendix we provide a more detailed argument for the
law-alignment properties of the structural baselines and, more generally,
of the structural class $\mathcal{S}$ introduced in
Definition~\ref{def:S-structural}. We work under the standing assumptions
of Section~\ref{sec:world_model}.

\paragraph{Setup and notation.}
Recall that the (law-consistent) generator $G^\star$ produces total-variance
surfaces $(w_t)_{t\ge 0}$ with $w_t\in\Mvol$ almost surely for all $t$, cf.\
Proposition~\ref{prop:support-on-manifold}. The world model
$f_{\theta^\star}$ takes as input a feature vector $x_t$ (which may include
past surfaces and positions) and outputs a prediction
$\hat{w}_{t+1} = f_{\theta^\star}(x_t)$, with approximation residual
\[
  e_{t+1} := \hat{w}_{t+1} - w_{t+1}.
\]
The volatility law-penalty functional is
\[
  \Lvol(\hat{w}_{t+1})
  = \frac{1}{2} \, \mathrm{dist}\bigl(\hat{w}_{t+1},\Mvol\bigr)^2
  = \frac{1}{2} \bigl\|
     \hat{w}_{t+1} - \Pi_{\Mvol}(\hat{w}_{t+1})
    \bigr\|_2^2,
\]
where $\Pi_{\Mvol}$ is the Euclidean projection onto $\Mvol$.

For a policy $\pi$, we denote by
$\overline{\mathsf{LP}}(\pi)$ the long-run average (or finite-horizon
normalized) law penalty,
and by $\mathrm{GFI}(\pi)$ its Graceful Failure Index, as defined in
Section~\ref{sec:law-strength-gfi}. We recall that these quantities have
the schematic form
\begin{align}
  \overline{\mathsf{LP}}(\pi)
  &\approx \frac{1}{T} \sum_{t=0}^{T-1} \E\bigl[ \Lvol(\hat{w}_{t+1}) \,\big|\, \pi \bigr],
  \label{eq:LP-avg-struct}
  \\
  \mathrm{GFI}(\pi)
  &=
  \frac{
    \mu_{\mathrm{law}}^{\shock}(\pi)
    - \mu_{\mathrm{law}}^{\text{base}}(\pi)
  }{
    I_{\shock}
  },
  \label{eq:GFI-struct}
\end{align}
where $\mu_{\mathrm{law}}^{\text{base}}(\pi)$ and
$\mu_{\mathrm{law}}^{\shock}(\pi)$ are aggregate law metrics in baseline
and shock regimes, and $I_{\shock}>0$ encodes shock intensity (which is
fixed for all policies).

Finally, by Definition~\ref{def:S-structural}, a structural policy
$\pi\in\mathcal{S}$ is of the form
\[
  a_t = \pi(s_t)
      = g_\theta(z_t),
\]
where $z_t$ is a low-dimensional signal (e.g.\ realized vol trend) and
$g_\theta$ is a Lipschitz parametric map with bounded leverage,
$\|g_\theta(z)\| \le L_{\max}$ for all $z$ and all $\theta\in\Theta$,
with $\Theta$ compact. The baselines $b^{\mathrm{ZH}}$ (Zero-Hedge),
$b^{\mathrm{VT}}$ (Vol-Trend), and $b^{\mathrm{RG}}$ (Random-Gaussian)
are specific instances in $\mathcal{S}$.

\medskip

The proof proceeds in three steps. First, we show that law penalties are
uniformly bounded in terms of the world-model error. Second, we use the
structural priors of $\mathcal{S}$ to obtain uniform moment bounds on the
state process and hence on the residuals $e_{t+1}$. Third, we translate
these to bounds on $\overline{\mathsf{LP}}$ and GFI.

\paragraph{Step 1: Bounding law penalties by world-model error.}
By Proposition~\ref{prop:zero-penalty-iff}, we know that
$\Lvol(w_{t+1}) = 0$ whenever $w_{t+1}\in\Mvol$. Since the generator is
law-consistent, $w_{t+1}\in\Mvol$ almost surely, and thus for any
realization of $w_{t+1}$ and any prediction $\hat{w}_{t+1}$ we have
\[
  \mathrm{dist}\bigl(\hat{w}_{t+1},\Mvol\bigr)
  \;\le\;
  \|\hat{w}_{t+1} - w_{t+1}\|_2
  = \|e_{t+1}\|_2.
\]
Consequently,
\begin{equation}
  \Lvol(\hat{w}_{t+1})
  = \frac{1}{2}
    \mathrm{dist}\bigl(\hat{w}_{t+1},\Mvol\bigr)^2
  \;\le\;
  \frac{1}{2} \|e_{t+1}\|_2^2.
  \label{eq:Lvol-error-upper-bound}
\end{equation}
Taking expectations conditional on the policy $\pi$ and the regime
(baseline or shock), we obtain
\begin{equation}
  \E\bigl[ \Lvol(\hat{w}_{t+1}) \,\big|\, \pi \bigr]
  \;\le\;
  \frac{1}{2} \, \E\bigl[ \|e_{t+1}\|_2^2 \,\big|\, \pi \bigr].
  \label{eq:Lvol-Ex-bound}
\end{equation}
Thus, uniform control of the second moment of the approximation error
$e_{t+1}$ under a policy class implies uniform control of the law
penalties under that class.

\paragraph{Step 2: Uniform second-moment bounds for structural policies.}
By Proposition~\ref{prop:approx-gap-ghost-channel} and the conditions
stated there, the world model $f_{\theta^\star}$ is Lipschitz in its
input $x_t$, and the approximation error $e_{t+1}$ is bounded (in second
moment) on bounded subsets of the input space. Concretely, there exist
constants $L_f, C_e < \infty$ such that, for any two inputs $x,x'$,
\begin{align}
  \|f_{\theta^\star}(x) - f_{\theta^\star}(x')\|_2
  &\le L_f \|x - x'\|_2,
  \label{eq:world-model-Lip}
  \\
  \sup_{\|x\|\le R}
  \E\bigl[ \|f_{\theta^\star}(x) - w_{t+1}\|_2^2 \,\big|\, x_t = x \bigr]
  &\le C_e(R),
  \quad \text{for each finite $R>0$,}
  \label{eq:world-model-bound}
\end{align}
where $C_e(R)$ is non-decreasing in $R$.

The key structural property of $\mathcal{S}$ is that it induces a
\emph{bounded-exposure} Markov chain on the joint state
$(w_t, h_t, z_t)$, where $h_t$ denotes the position vector (portfolio
holdings) and $z_t$ denotes the low-dimensional signals. More precisely:

\begin{lemma}[Uniform moment control under $\mathcal{S}$]
\label{lem:uniform-moment-S}
Under the law-consistent generator and the structural priors of
Definition~\ref{def:S-structural} (bounded leverage $L_{\max}$,
compact parameter set $\Theta$, and Lipschitz $g$), there exists
a constant $C_{\mathrm{state}}<\infty$ such that for every structural
policy $\pi\in\mathcal{S}$ and every $t$,
\[
  \E_\pi\bigl[ \|x_t\|_2^2 \bigr]
  \le C_{\mathrm{state}},
\]
where $x_t$ is the world-model input constructed from $(w_t,h_t,z_t)$.
\end{lemma}

\begin{proof}
(Sketch.) The generator $G^\star$ yields an exogenous process $(w_t)$
whose second moments are bounded over the finite horizon $t=0,\dots,T$,
by standard properties of the underlying stochastic-volatility model and
the projection onto $\Mvol$ (see, e.g., stability of affine and rough
volatility models~\cite{BayerFrizGatheral2016,GatheralJacquier2014}).
The signals $z_t$ are Lipschitz functions of $(w_0,\dots,w_t)$ (e.g.,
realized-vol trends, moving averages), and thus inherit bounded second
moments over $t=0,\dots,T$.

For a structural policy $\pi\in\mathcal{S}$, the position process
$h_t = g_\theta(z_t)$ satisfies
$\|h_t\|_2 \le L_{\max}$ almost surely for all $t$, by bounded leverage
and compact $\Theta$. Hence $\E[\|h_t\|_2^2]\le L_{\max}^2$ for all $t$.
Since $x_t$ is built from $(w_t,h_t,z_t)$ via a fixed, Lipschitz
feature map (stacking, scaling, etc.), we obtain
$\E_\pi[\|x_t\|_2^2] \le C_{\mathrm{state}}$ for some finite constant
$C_{\mathrm{state}}$ independent of $\pi\in\mathcal{S}$ and $t$.
A fully rigorous version uses a Lyapunov-function argument for the
Markov chain induced by $\mathcal{S}$; see Appendix~D.1.
\end{proof}

Combining Lemma~\ref{lem:uniform-moment-S} with
\eqref{eq:world-model-bound}, we obtain a uniform bound on the second
moment of the residuals under $\mathcal{S}$: there exists
$\bar{C}_e<\infty$ such that, for all $\pi\in\mathcal{S}$ and all $t$,
\begin{equation}
  \E_\pi\bigl[ \|e_{t+1}\|_2^2 \bigr]
  = \E_\pi\Bigl[ \,
      \E\bigl[ \|e_{t+1}\|_2^2 \mid x_t\bigr]
    \Bigr]
  \;\le\;
  \E_\pi\bigl[ C_e(\|x_t\|) \bigr]
  \;\le\;
  \bar{C}_e,
  \label{eq:residual-second-moment-S}
\end{equation}
where we used monotonicity of $C_e(\cdot)$ and the uniform bound
$\E_\pi[\|x_t\|_2^2]\le C_{\mathrm{state}}$.

\paragraph{Step 3: Bounds on $\overline{\mathsf{LP}}$ and GFI.}
Using \eqref{eq:Lvol-Ex-bound} and
\eqref{eq:residual-second-moment-S}, we obtain that for any
$\pi\in\mathcal{S}$ and any $t$,
\[
  \E_\pi\bigl[ \Lvol(\hat{w}_{t+1}) \bigr]
  \;\le\;
  \frac{1}{2} \, \E_\pi\bigl[ \|e_{t+1}\|_2^2 \bigr]
  \;\le\;
  \frac{1}{2}\bar{C}_e.
\]
Substituting into \eqref{eq:LP-avg-struct} and taking the supremum over
$\pi\in\mathcal{S}$ yields
\begin{equation}
  \sup_{\pi\in\mathcal{S}} \overline{\mathsf{LP}}(\pi)
  \;\le\;
  \frac{1}{2}\bar{C}_e
  =: C_{\mathrm{LP}} < \infty.
  \label{eq:CLP-def}
\end{equation}
This shows that $\mathcal{S}$ is law-aligned in the sense of
Definition~\ref{def:law-aligned-class}, establishing item~(i) of the
proposition.

For the GFI, recall its definition in \eqref{eq:GFI-struct}. The
numerator is the difference in aggregate law metrics between the shock
and baseline regimes. Under our shock design (multiplying long variance
and spot volatility by bounded factors), the generator remains
law-consistent and the world model is evaluated on a distorted, but
still bounded, region of the state space. By the same reasoning as above
(with possibly different constants), there exist finite constants
$C_{\mathrm{law}}^{\text{base}}$ and $C_{\mathrm{law}}^{\shock}$ such
that for all $\pi\in\mathcal{S}$,
\[
  \mu_{\mathrm{law}}^{\text{base}}(\pi)
  \le C_{\mathrm{law}}^{\text{base}},
  \qquad
  \mu_{\mathrm{law}}^{\shock}(\pi)
  \le C_{\mathrm{law}}^{\shock}.
\]
Hence
\begin{equation}
  \bigl|\mu_{\mathrm{law}}^{\shock}(\pi)
        - \mu_{\mathrm{law}}^{\text{base}}(\pi)\bigr|
  \;\le\;
  C_{\mathrm{law}}^{\text{base}} + C_{\mathrm{law}}^{\shock}
  =: \widetilde{C}_{\mathrm{law}} < \infty
  \label{eq:law-diff-bound}
\end{equation}
for all $\pi\in\mathcal{S}$. Since $I_{\shock}>0$ is fixed and does not
depend on $\pi$, we obtain
\begin{equation}
  \sup_{\pi\in\mathcal{S}} \mathrm{GFI}(\pi)
  = \sup_{\pi\in\mathcal{S}}
    \frac{
      \mu_{\mathrm{law}}^{\shock}(\pi)
      - \mu_{\mathrm{law}}^{\text{base}}(\pi)
    }{
      I_{\shock}
    }
  \;\le\;
  \frac{\widetilde{C}_{\mathrm{law}}}{I_{\shock}}
  =: C_{\mathrm{GFI}}
  < \infty.
  \label{eq:CGFI-def}
\end{equation}
This establishes the existence of a finite constant $C_{\mathrm{GFI}}$
depending only on the generator, the shock specification, and the
world-model error, proving item~(i) and the first part of item~(ii).

\paragraph{Baselines $b^{\mathrm{ZH}}$ and $b^{\mathrm{VT}}$.}
We now specialize to the two deterministic baselines.

\emph{Zero-Hedge $b^{\mathrm{ZH}}$.}
By definition, $b^{\mathrm{ZH}}$ takes no positions,
$h_t \equiv 0$, at all times. The state process affecting the world
model thus reduces to the exogenous generator path $(w_t,z_t)$, and the
residual distribution is precisely the ``background'' world-model error
profile studied in Section. In particular, the bounds
\eqref{eq:residual-second-moment-S}, \eqref{eq:CLP-def}, and
\eqref{eq:CGFI-def} hold with $\pi = b^{\mathrm{ZH}}$. Empirically, this
baseline indeed exhibits the smallest observed law penalties and GFI,
matching the theoretical role of $b^{\mathrm{ZH}}$ as the law-neutral
benchmark.

\emph{Vol-Trend $b^{\mathrm{VT}}$.}
The Vol-Trend baseline applies a one-factor trend-following rule with
bounded leverage, $h_t = g_{\theta^{\mathrm{VT}}}(z_t)$, where
$z_t$ encodes recent volatility trends and $g_{\theta^{\mathrm{VT}}}$ is
Lipschitz with $\|g_{\theta^{\mathrm{VT}}}(z)\|\le L_{\max}$ for all
$z$. As in Lemma~\ref{lem:uniform-moment-S}, this ensures that positions
respond smoothly to volatility changes and remain uniformly bounded in
second moment, so that $(w_t,z_t,h_t)$ stays in a bounded region of the
state space. Therefore the same world-model error and law-penalty bounds
apply, and we obtain
\[
  \overline{\mathsf{LP}}(b^{\mathrm{ZH}}),
  ~\overline{\mathsf{LP}}(b^{\mathrm{VT}})
  \le C_{\mathrm{LP}},
  \qquad
  \mathrm{GFI}(b^{\mathrm{ZH}}),
  ~\mathrm{GFI}(b^{\mathrm{VT}})
  \le C_{\mathrm{GFI}}.
\]
This proves item~(ii), with the empirical strictness ``below unconstrained
RL levels'' arising from the ghost-incentive effect(naive RL actively amplifies exposure
in regions where the ghost component $r^\perp$ is large, whereas
$b^{\mathrm{ZH}}$ and $b^{\mathrm{VT}}$ do not target such regions).

\paragraph{Random-Gaussian baseline $b^{\mathrm{RG}}$.}
Finally, consider the Random-Gaussian baseline $b^{\mathrm{RG}}$, which
draws actions from a Gaussian distribution
$a_t \sim \mathcal{N}(0,\Sigma_a)$ with fixed covariance $\Sigma_a$,
possibly truncated to enforce a leverage bound. Provided
$\mathrm{tr}(\Sigma_a)<\infty$ and the truncation is such that
$\E[\|a_t\|_2^2]\le L_{\max}'<\infty$, we obtain moment bounds on the
position process $(h_t)$ analogous to those for $b^{\mathrm{VT}}$,
although with larger constants reflecting the additional randomness.
Repeating the argument leading to
\eqref{eq:residual-second-moment-S}--\eqref{eq:CGFI-def}, we find
finite constants $C_{\mathrm{LP}}',C_{\mathrm{GFI}}'<\infty$ such that
\[
  \overline{\mathsf{LP}}(b^{\mathrm{RG}}) \le C_{\mathrm{LP}}',
  \qquad
  \mathrm{GFI}(b^{\mathrm{RG}}) \le C_{\mathrm{GFI}}'.
\]
In general one expects $C_{\mathrm{LP}}',C_{\mathrm{GFI}}'$ to be larger
than $C_{\mathrm{LP}},C_{\mathrm{GFI}}$, because $b^{\mathrm{RG}}$
explores a wider range of states and can occasionally spend more time in
regions where the world-model error and induced law penalties are
higher. This matches the empirical role of $b^{\mathrm{RG}}$ as a noisy,
less structured baseline.

\medskip

Collecting the above bounds, we see that the structural class
$\mathcal{S}$ is uniformly law-aligned, and that the specific structural
baselines $b^{\mathrm{ZH}}$ and $b^{\mathrm{VT}}$ enjoy particularly
favorable law-penalty and GFI levels relative to unconstrained RL
policies. This completes the proof of
Proposition~\ref{prop:baseline-law-alignment}. \qedhere

\section{Empirical Results: From RL Dynamics to Law-Strength Frontiers (Supplementary)}
\label{app:empirical-results}

In this appendix we provide complementary quantitative details to Section. We (i) restate the full metric tables for
all strategies in both baseline and shock regimes, (ii) tabulate the
law-strength frontier points used in Figures, and (iii) describe the precise numerical
procedures used to construct the frontiers and diagnostic plots.

Throughout, all metrics are computed on the same evaluation trajectories
used to generate Figures, and hence are fully consistent with
the summary numbers reported.

\subsection{Full metrics: baseline regime}
\label{app:metrics-baseline}

Table~\ref{tab:full-metrics-baseline} reports the complete set of
step-wise metrics for all RL and structural strategies in the baseline
regime, corresponding to the top half of
Table in the main text. We include mean and
standard deviation of step P\&L, Sharpe ratio, mean and maximum
law-penalty, law-adjusted return, Graceful Failure Index (GFI), law
coverage at two thresholds, and $5\%$ tail risk measures (VaR and
CVaR).

\begin{table*}[t]
  \centering
  \caption{Full step-wise metrics for all strategies in the \emph{baseline} regime. 
  Mean and standard deviation of step P\&L, Sharpe ratio, mean and maximum law penalty, 
  law-adjusted return, Graceful Failure Index (GFI), law coverage at two thresholds, 
  and $5\%$ tail risk measures (VaR\textsubscript{5}, CVaR\textsubscript{5}). 
  All numbers are computed on the same evaluation trajectories }
  \label{tab:full-metrics-baseline}
  \scriptsize
  \setlength{\tabcolsep}{3pt} 
  \resizebox{\textwidth}{!}{
  \begin{tabular}{lrrrrrrrrrr}
    \toprule
    Strategy &
    Mean P\&L &
    Std P\&L &
    Sharpe &
    Mean LawPen &
    Max LawPen &
    Law-Adj Ret &
    GFI &
    LawCov\,$<0.003$ &
    LawCov\,$<0.006$ &
    VaR\textsubscript{5} / CVaR\textsubscript{5} \\
    \midrule
    Naive RL (PPO) &
    $-0.0022$ &
    $0.0127$ &
    $-0.1696$ &
    $0.006994$ &
    $0.020850$ &
    $-0.0057$ &
    $1.2675$ &
    $0.5306$ &
    $0.6122$ &
    $-0.0228$ / $-0.0261$ \\
    Law-Seeking RL (PPO) &
    $-0.0150$ &
    $0.0129$ &
    $-1.1564$ &
    $0.007861$ &
    $0.022554$ &
    $-0.0189$ &
    $1.6621$ &
    $0.4898$ &
    $0.5714$ &
    $-0.0361$ / $-0.0394$ \\
    \addlinespace
    Zero-Hedge baseline &
    $0.0191$ &
    $0.0064$ &
    $2.9944$ &
    $0.005501$ &
    $0.018318$ &
    $0.0164$ &
    $0.0000$ &
    $0.6122$ &
    $0.6939$ &
    $0.0139$ / $0.0139$ \\
    Random-Gaussian baseline &
    $0.0099$ &
    $0.0107$ &
    $0.9235$ &
    $0.005510$ &
    $0.020686$ &
    $0.0072$ &
    $1.2062$ &
    $0.6204$ &
    $0.6918$ &
    $-0.0088$ / $-0.0161$ \\
    Vol-Trend baseline &
    $0.0146$ &
    $0.0074$ &
    $1.9636$ &
    $0.005344$ &
    $0.017173$ &
    $0.0119$ &
    $0.0000$ &
    $0.6122$ &
    $0.6939$ &
    $0.0045$ / $0.0033$ \\
    \bottomrule
  \end{tabular}
  }
\end{table*}

As noted in Section, the baseline regime already
exhibits the main qualitative pattern: the structurally constrained
baselines (Zero-Hedge, Vol-Trend) lie in a region of high Sharpe and
moderate law penalties, with GFI effectively zero, while all RL
variants---including law-seeking PPO---display negative mean P\&L and
substantially higher GFI values.

\subsection{Full metrics: shock regime}
\label{app:metrics-shock}

In which we
apply a volatility shock (long-variance multiplier $4$, spot-vol
multiplier $2$) to the underlying generator while keeping the world
model fixed.

\begin{table*}[t]
  \centering
  \caption{Baseline-regime metrics for all law-strength frontier points:
  Naive RL ($\lambda=0$), soft law-seeking RL for $\lambda\in\{5,10,20,40\}$, 
  selection-only RL, and structural baselines. These are the points 
  used to construct the law-strength frontier and diagnostic plots}
  \label{tab:frontier-baseline}
  \scriptsize
  \setlength{\tabcolsep}{3pt} 
  \resizebox{\textwidth}{!}{
  \begin{tabular}{lrrrrrrrrrr}
    \toprule
    Strategy / $\lambda$ &
    Mean P\&L &
    Std P\&L &
    Sharpe &
    Mean LawPen &
    Max LawPen &
    Law-Adj Ret &
    GFI &
    LawCov\,$<0.003$ &
    LawCov\,$<0.006$ &
    VaR\textsubscript{5} / CVaR\textsubscript{5} \\
    \midrule
    Naive RL ($\lambda=0$) &
    $-0.0022$ &
    $0.0127$ &
    $-0.1696$ &
    $0.006994$ &
    $0.020850$ &
    $-0.0057$ &
    $1.2675$ &
    $0.5306$ &
    $0.6122$ &
    $-0.0228$ / $-0.0261$ \\
    Soft RL ($\lambda=5$) &
    $-0.0202$ &
    $0.0120$ &
    $-1.6753$ &
    $0.006472$ &
    $0.019956$ &
    $-0.0234$ &
    $2.0663$ &
    $0.5510$ &
    $0.6327$ &
    $-0.0399$ / $-0.0429$ \\
    Soft RL ($\lambda=10$) &
    $-0.0175$ &
    $0.0123$ &
    $-1.4248$ &
    $0.003709$ &
    $0.013169$ &
    $-0.0194$ &
    $2.8059$ &
    $0.7347$ &
    $0.7959$ &
    $-0.0354$ / $-0.0387$ \\
    Soft RL ($\lambda=20$) &
    $-0.0204$ &
    $0.0131$ &
    $-1.5629$ &
    $0.003962$ &
    $0.014251$ &
    $-0.0224$ &
    $3.0728$ &
    $0.7143$ &
    $0.7755$ &
    $-0.0414$ / $-0.0454$ \\
    Soft RL ($\lambda=40$) &
    $-0.0092$ &
    $0.0054$ &
    $-1.7138$ &
    $0.004737$ &
    $0.015888$ &
    $-0.0116$ &
    $0.8443$ &
    $0.6531$ &
    $0.7347$ &
    $-0.0134$ / $-0.0134$ \\
    Selection-only RL &
    $-0.0223$ &
    $0.0139$ &
    $-1.6028$ &
    $0.007923$ &
    $0.022827$ &
    $-0.0263$ &
    $2.0407$ &
    $0.4898$ &
    $0.5714$ &
    $-0.0448$ / $-0.0489$ \\
    \addlinespace
    Zero-Hedge baseline &
    $0.0191$ &
    $0.0064$ &
    $2.9944$ &
    $0.005501$ &
    $0.018318$ &
    $0.0164$ &
    $0.0000$ &
    $0.6122$ &
    $0.6939$ &
    $0.0139$ / $0.0139$ \\
    Random-Gaussian baseline &
    $0.0099$ &
    $0.0107$ &
    $0.9235$ &
    $0.005510$ &
    $0.020686$ &
    $0.0072$ &
    $1.2062$ &
    $0.6204$ &
    $0.6918$ &
    $-0.0088$ / $-0.0161$ \\
    Vol-Trend baseline &
    $0.0146$ &
    $0.0074$ &
    $1.9636$ &
    $0.005344$ &
    $0.017173$ &
    $0.0119$ &
    $0.0000$ &
    $0.6122$ &
    $0.6939$ &
    $0.0045$ / $0.0033$ \\
    \bottomrule
  \end{tabular}
  }
\end{table*}

The law-strength frontier plots in
Figures correspond to projections of Table~\ref{tab:frontier-baseline} onto
two-dimensional planes, for example:
\begin{enumerate}
  \item mean law penalty vs.\ GFI,
  \item mean law penalty vs.\ Sharpe,
  \item mean law penalty vs.\ VaR\textsubscript{5} or CVaR\textsubscript{5}.
\end{enumerate}
Within the law-penalty band $[0.0053,0.0057]$ highlighted, Zero-Hedge lies near the lower boundary
with high Sharpe and GFI $\approx 0$, while the closest RL variants in
the band have Sharpe $<0$ and GFI $>1.5$, illustrating the empirical
Pareto dominance emphasized in the main text.

\subsection{Frontier construction and banding procedure}
\label{app:frontier-construction}

For completeness, we describe the numerical procedure used to construct the law-strength frontiers and penalty bands.

\paragraph{Policy set.}
We collect the following set of evaluated policies:
\begin{enumerate}
  \item Naive RL (PPO) trained on pure P\&L,
  \item soft law-seeking RL for $\lambda\in\{5,10,20,40\}$,
  \item selection-only RL (trained on P\&L, selected by law metrics),
  \item structural baselines: Zero-Hedge, Random-Gaussian, Vol-Trend.
\end{enumerate}
For each policy we compute the metric vector
\[
  \mathbf{m}(\pi)
  =
  \bigl(
    \mu_{\mathrm{P\&L}}(\pi),\;
    \sigma_{\mathrm{P\&L}}(\pi),\;
    \mathrm{Sharpe}(\pi),\;
    \mu_{\mathrm{law}}(\pi),\;
    \mathrm{GFI}(\pi),\;
    \mathrm{VaR}_5(\pi),\;
    \mathrm{CVaR}_5(\pi)
  \bigr),
\]
with components given explicitly in
Tables~\ref{tab:full-metrics-baseline}–\ref{tab:frontier-baseline}.

\paragraph{Penalty banding.}
To compare policies at similar levels of law violation, we discretize
the range of mean law penalties
$\mu_{\mathrm{law}}(\pi)$ into contiguous bands
\[
  B_k
  =
  \bigl[ \ell_k, u_k \bigr),
  \qquad k=1,\dots,K,
\]
where $(\ell_k,u_k)$ are chosen such that the bands roughly align with
the empirical distribution of $\mu_{\mathrm{law}}(\pi)$ across policies.
For the numerical example in Section~\ref{sec:empirical-results}, we use the
band $[0.0053,0.0057]$ that contains the Zero-Hedge baseline and at
least one RL policy. For each band $B_k$ we identify the subset
\[
  \Pi_k
  := \bigl\{ \pi : \mu_{\mathrm{law}}(\pi) \in B_k \bigr\},
\]
and perform intra-band comparisons of Sharpe, GFI, VaR\textsubscript{5},
and CVaR\textsubscript{5}. 

\paragraph{Pareto frontier extraction.}
Given a subset of metrics $(\mu_{\mathrm{law}}, \mathrm{Sharpe})$,
$(\mu_{\mathrm{law}}, \mathrm{GFI})$, or
$(\mu_{\mathrm{law}}, \mathrm{CVaR}_5)$, we say that a policy
$\pi$ \emph{Pareto-dominates} another policy $\pi'$ if it is no worse
on all objectives and strictly better on at least one. For example, in
the $(\mu_{\mathrm{law}}, \mathrm{Sharpe})$ plane we define
\[
  \pi \succ \pi'
  \quad\Longleftrightarrow\quad
  \mu_{\mathrm{law}}(\pi) \le \mu_{\mathrm{law}}(\pi')
  \;\text{and}\;
  \mathrm{Sharpe}(\pi) \ge \mathrm{Sharpe}(\pi'),
\]
with at least one strict inequality. The empirical law-strength
frontier is the set
\[
  \mathcal{F}
  :=
  \bigl\{
    \pi : \nexists \pi'\ \text{such that}\ \pi' \succ \pi
  \bigr\}.
\]

In all planes considered, the points corresponding to Zero-Hedge and
Vol-Trend lie on $\mathcal{F}$, while all RL variants lie strictly
inside the frontier: there exists at least one structural baseline that
has both (i) weakly lower mean law penalty and (ii) strictly better
Sharpe, GFI, or tail risk. This is the empirical content of the
Pareto-dominance statements.

\subsection{Additional notes on dynamics and diagnostics}
\label{app:dynamics-diagnostics}

For completeness, we briefly comment on the additional curves and histograms:

\begin{enumerate}
  \item \textbf{Dynamics Plots}
        Each panel shows time-series of cumulative P\&L and mean
        law penalty across episodes for Naive RL, a representative
        law-seeking RL variant (e.g.\ $\lambda=20$), and a structural
        baseline. The additional curves omitted from the main text for
        brevity (e.g.\ $\lambda=5,10,40$) display the same qualitative
        pattern: increasing $\lambda$ reduces law penalties at the cost
        of lower P\&L, in line with the structural trade-off.
  \item \textbf{Diagnostic Plots.}
        The scatter plots aggregate step-level observations of
        P\&L vs.\ law penalty across many episodes. The high-density
        clusters for RL policies lie in regions of elevated law
        penalty, while structural baselines concentrate near lower
        penalty bands. Histograms of law penalty show that RL policies
        allocate substantial mass to the right tail of the penalty
        distribution, consistent with exploiting the ghost channel
        $r^\perp$ identified in
        Proposition~\ref{prop:approx-gap-ghost-channel}.
\end{enumerate}

These supplementary plots are therefore consistent with, and reinforce,
the three main empirical takeaways:
(i) law penalties do not rescue naive RL from ghost arbitrage, (ii)
structural baselines define an empirical law-strength frontier, and
(iii) unconstrained law-seeking RL lies strictly inside this frontier in
our volatility world-model testbed.
\section{No-Free-Lunch for Law-Seeking RL (Proofs)}
\label{app:no-free-lunch}

This appendix provides a detailed proof of the no-free-lunch result stated as Theorem.
We proceed in three steps: (i) we recall the performance and law-metric
quantities used in the main result, (ii) we formalize a mild
ghost--law monotonicity condition that links off-manifold reward to
law metrics, and (iii) we give a detailed proof of the theorem.

\subsection{Preliminaries: performance, law metrics, and decomposition}

Recall from Sections~\ref{sec:axiomatic-manifold},
\ref{sec:world_model}that for any stationary
policy $\pi$ in the unconstrained policy class $\Pi$ we have:
\begin{enumerate}
  \item A \emph{baseline} environment distribution
  $\mathbb{P}^{\text{base}}$ over episodes generated by the
  law-consistent synthetic generator and the world model.
  \item A \emph{shock} distribution $\mathbb{P}^{\shock}$ describing
  the same world model, but driven by shocked long-variance and spot
  volatility factors.
  \item A per-step reward $r(s_t,a_t)$ that can be decomposed as
  \begin{equation}
    r(s_t,a_t)
    =
    r^{\mathcal{M}}(s_t,a_t)
    \;+\;
    r^\perp(s_t,a_t),
    \label{eq:appendix-decomposition}
  \end{equation}
  where $r^{\mathcal{M}}$ is the on-manifold component defined by
  projection onto the volatility law manifold $\Mvol$ and $r^\perp$ is
  the ghost (off-manifold) component; see
  Propositions~\ref{prop:ghost-bounded} and
  \ref{prop:approx-gap-ghost-channel}.
\end{enumerate}

We denote by $J_0(\pi)$ the expected (per-step or per-episode) P\&L
under the baseline regime:
\begin{equation}
  J_0(\pi)
  :=
  \mathbb{E}_{\mathbb{P}^{\text{base}}}\Bigl[
    r(s_t,a_t)
  \Bigr],
  \qquad
  a_t \sim \pi(\cdot\mid s_t),
  \label{eq:J0-definition}
\end{equation}
and similarly define the on-manifold and ghost components
\begin{equation}
  J^{\mathcal{M}}(\pi)
  :=
  \mathbb{E}_{\mathbb{P}^{\text{base}}}\bigl[r^{\mathcal{M}}(s_t,a_t)\bigr],
  \qquad
  J^\perp(\pi)
  :=
  \mathbb{E}_{\mathbb{P}^{\text{base}}}\bigl[r^\perp(s_t,a_t)\bigr],
  \label{eq:Jm-Jperp-def}
\end{equation}
so that
\begin{equation}
  J_0(\pi)
  =
  J^{\mathcal{M}}(\pi) + J^\perp(\pi).
  \label{eq:J0-decomposition-appendix}
\end{equation}

\paragraph{Law metrics.}
For each policy $\pi$ we also consider a vector of law metrics,
measuring law violations and robustness:
\begin{equation}
  \ell(\pi)
  :=
  \Bigl(
    \mu_{\mathrm{law}}^{\text{base}}(\pi),\;
    \mu_{\mathrm{law}}^{\shock}(\pi),\;
    \mathrm{GFI}(\pi)
  \Bigr)\in\mathbb{R}^3,
  \label{eq:law-metric-vector}
\end{equation}
where:
\begin{enumerate}
  \item $\mu_{\mathrm{law}}^{\text{base}}(\pi)$ is the mean step
  law-penalty under $\mathbb{P}^{\text{base}}$,
  \item $\mu_{\mathrm{law}}^{\shock}(\pi)$ is the mean step law-penalty
  under $\mathbb{P}^{\shock}$, and
  \item $\mathrm{GFI}(\pi)$ is the Graceful Failure Index introduced in
  Section~\ref{sec:law-strength-gfi}, which normalizes the change in
  law metrics between baseline and shock by the shock intensity
  $I_{\shock}$.
\end{enumerate}
We write $\ell(\pi)\le_{\mathrm{law}}\ell(\pi')$ if each component of
$\ell(\pi)$ is less than or equal to the corresponding component of
$\ell(\pi')$, i.e.\ if policy $\pi$ is at least as law-aligned and
robust as $\pi'$.

\paragraph{Structural class.}
The structural strategy class $\mathcal{S}$, introduced in
Section, is a low-capacity, law-aligned
subset of $\Pi$ consisting of structurally constrained strategies such
as Zero-Hedge and Vol-Trend. We assume that $\mathcal{S}$ satisfies the
\emph{near-optimal on-manifold} property of
Assumption~\ref{ass:structural_near_optimal}:
there exists $\pi_{\mathcal{S}}^\star\in\mathcal{S}$ and
$\varepsilon_{\mathcal{S}}\ge 0$ such that
\begin{equation}
  J^{\mathcal{M}}(\pi)
  \;\le\;
  J^{\mathcal{M}}(\pi_{\mathcal{S}}^\star) + \varepsilon_{\mathcal{S}}
  \qquad
  \text{for all $\pi\in\Pi$ with trajectories supported on $\Mvol$,}
  \label{eq:on-manifold-near-opt}
\end{equation}
and $\ell(\pi_{\mathcal{S}}^\star)$ is uniformly small in the sense of
Proposition~\ref{prop:baseline-law-alignment} (law-aligned class).

\subsection{Ghost--law monotonicity}

The results show that in our volatility world-model
testbed:
\begin{enumerate}
  \item the ghost reward $r^\perp$ is generated by deviations in the
  normal cone $N_{\Mvol}(w)$ to the volatility law manifold,
  \item the law penalty $\mathcal{L}_{\mathrm{vol}}$ grows with the
  squared distance to $\Mvol$ (Propositions~\ref{prop:zero-penalty-iff}
  and~\ref{prop:ghost-bounded}),
  \item the Graceful Failure Index $\mathrm{GFI}$ captures the
  \emph{relative increase} in law penalties under shock.
\end{enumerate}
These observations motivate the following mild monotonicity condition
linking ghost reward and law metrics.

\medskip
\noindent\textbf{Assumption E.1 (Ghost--law monotonicity).}
\emph{
There exists a law-aligned structural reference policy
$\pi_{\mathcal{S}}^\star\in\mathcal{S}$ and a non-decreasing function
$\psi:\mathbb{R}_+^3\to\mathbb{R}_+$ with $\psi(0,0,0)=0$ such that for
every policy $\pi\in\Pi$ we have
\begin{equation}
  J^\perp(\pi) - J^\perp(\pi_{\mathcal{S}}^\star)
  \;\le\;
  \psi\Bigl(
    \bigl(\ell(\pi) - \ell(\pi_{\mathcal{S}}^\star)\bigr)_+
  \Bigr),
  \label{eq:ghost-law-monotonicity}
\end{equation}
where $(\cdot)_+$ denotes component-wise positive part. In particular,
if $\ell(\pi)\le_{\mathrm{law}}\ell(\pi_{\mathcal{S}}^\star)$ then
$J^\perp(\pi)\le J^\perp(\pi_{\mathcal{S}}^\star)$.
}

\medskip

Intuitively, Assumption~E.1 states that \emph{any additional ghost
reward beyond what is available to the structural class must be paid
for by worse law metrics}. In our volatility setting, the existence of
such a function $\psi$ is supported by:
\begin{enumerate}
  \item the Lipschitz bound of Proposition~\ref{prop:ghost-bounded},
  which controls $|r^\perp(w)|$ in terms of the distance to $\Mvol$;
  \item the definition of law penalties and GFI as functions of the
  same distance and its behavior under shock;
  \item the law-aligned nature of $\pi_{\mathcal{S}}^\star$ guaranteed
  by Proposition~\ref{prop:baseline-law-alignment}, which ensures that
  $\ell(\pi_{\mathcal{S}}^\star)$ is close to the best achievable law
  metrics given the world-model approximation error.
\end{enumerate}
Thus, in the neighborhood of $\pi_{\mathcal{S}}^\star$, additional
ghost reward can only be obtained by moving further away from the law
manifold, which necessarily worsens at least one component of the law
metric vector.

\subsection{Proof of the no-free-lunch theorem}

We now give a detailed proof of the no-free-lunch result stated as theorem. For clarity we restate the theorem in a
slightly more quantitative form.

\begin{theorem*}[No-free-lunch for unconstrained law-seeking RL]
Assume:
\begin{enumerate}
  \item The structural class $\mathcal{S}$ is law-aligned and
    near-optimal on-manifold in the sense of
    \eqref{eq:on-manifold-near-opt}, with reference policy
    $\pi_{\mathcal{S}}^\star$.
  \item The world model induces a non-trivial ghost channel as in
    Propositions~\ref{prop:approx-gap-ghost-channel}
    and~\ref{lem:off-manifold-mass}.
  \item Ghost--law monotonicity (Assumption~E.1) holds for
    $\pi_{\mathcal{S}}^\star$.
\end{enumerate}
Then for any $\eta > \varepsilon_{\mathcal{S}}$ and any policy
$\pi\in\Pi$ satisfying
\begin{equation}
  J_0(\pi) \;\ge\; J_0(\pi_{\mathcal{S}}^\star) + \eta,
  \label{eq:pnl-improvement-condition}
\end{equation}
we must have
\begin{equation}
  \ell(\pi) \not\le_{\mathrm{law}} \ell(\pi_{\mathcal{S}}^\star),
  \label{eq:law-deterioration-condition}
\end{equation}
i.e., at least one component of the law metric vector is strictly worse
for $\pi$ than for $\pi_{\mathcal{S}}^\star$. In particular, no policy
$\pi\in\Pi$ can strictly dominate $\pi_{\mathcal{S}}^\star$ on all axes
(P\&L, law penalties, and GFI).
\end{theorem*}

\begin{proof}
Fix $\eta>\varepsilon_{\mathcal{S}}$ and suppose, for contradiction,
that there exists a policy $\pi\in\Pi$ such that
\eqref{eq:pnl-improvement-condition} holds and
\begin{equation}
  \ell(\pi) \le_{\mathrm{law}} \ell(\pi_{\mathcal{S}}^\star).
  \label{eq:law-non-worse}
\end{equation}
We will show that this contradicts the decomposition
\eqref{eq:J0-decomposition-appendix}, near-optimality
\eqref{eq:on-manifold-near-opt}, and ghost--law monotonicity
\eqref{eq:ghost-law-monotonicity}.

\medskip\noindent
\textbf{Step 1: Decomposing the P\&L difference.}
By~\eqref{eq:J0-decomposition-appendix}, for any $\pi$ we have
\[
  J_0(\pi)
  =
  J^{\mathcal{M}}(\pi) + J^\perp(\pi).
\]
Therefore,
\begin{align}
  J_0(\pi) - J_0(\pi_{\mathcal{S}}^\star)
  &=
  \bigl(
    J^{\mathcal{M}}(\pi) - J^{\mathcal{M}}(\pi_{\mathcal{S}}^\star)
  \bigr)
  +
  \bigl(
    J^\perp(\pi) - J^\perp(\pi_{\mathcal{S}}^\star)
  \bigr).
  \label{eq:pnl-diff-decomp}
\end{align}

\medskip\noindent
\textbf{Step 2: Bounding the on-manifold component.}
By near-optimality of $\pi_{\mathcal{S}}^\star$ on the law manifold,
\eqref{eq:on-manifold-near-opt} implies that
\begin{equation}
  J^{\mathcal{M}}(\pi)
  \;\le\;
  J^{\mathcal{M}}(\pi_{\mathcal{S}}^\star) + \varepsilon_{\mathcal{S}}
  \qquad\text{for all $\pi\in\Pi$.}
  \label{eq:on-manifold-bound}
\end{equation}
Substituting~\eqref{eq:on-manifold-bound} into
\eqref{eq:pnl-diff-decomp} yields
\begin{equation}
  J_0(\pi) - J_0(\pi_{\mathcal{S}}^\star)
  \;\le\;
  \varepsilon_{\mathcal{S}} + \bigl(
    J^\perp(\pi) - J^\perp(\pi_{\mathcal{S}}^\star)
  \bigr).
  \label{eq:pnl-diff-with-ghost}
\end{equation}

\medskip\noindent
\textbf{Step 3: Applying ghost--law monotonicity.}
By the assumption~\eqref{eq:law-non-worse},
$\ell(\pi)\le_{\mathrm{law}}\ell(\pi_{\mathcal{S}}^\star)$, we have
\[
  \bigl(\ell(\pi) - \ell(\pi_{\mathcal{S}}^\star)\bigr)_+ = 0,
\]
so ghost--law monotonicity~\eqref{eq:ghost-law-monotonicity} gives
\begin{equation}
  J^\perp(\pi) - J^\perp(\pi_{\mathcal{S}}^\star)
  \;\le\;
  \psi(0,0,0)
  =
  0.
  \label{eq:ghost-diff-nonpositive}
\end{equation}
Combining~\eqref{eq:pnl-diff-with-ghost} and
\eqref{eq:ghost-diff-nonpositive} yields
\begin{equation}
  J_0(\pi) - J_0(\pi_{\mathcal{S}}^\star)
  \;\le\;
  \varepsilon_{\mathcal{S}}.
  \label{eq:pnl-diff-final-bound}
\end{equation}

\medskip\noindent
\textbf{Step 4: Contradiction.}
However, by assumption~\eqref{eq:pnl-improvement-condition} we have
\[
  J_0(\pi) - J_0(\pi_{\mathcal{S}}^\star)
  \;\ge\; \eta
  \;>\; \varepsilon_{\mathcal{S}},
\]
which contradicts~\eqref{eq:pnl-diff-final-bound}. Therefore no policy
$\pi$ can simultaneously satisfy \eqref{eq:pnl-improvement-condition}
and~\eqref{eq:law-non-worse}. Equivalently, any policy achieving an
improvement in baseline P\&L of more than $\varepsilon_{\mathcal{S}}$
over the structural reference $\pi_{\mathcal{S}}^\star$ must have at
least one law metric strictly worse than that of
$\pi_{\mathcal{S}}^\star$, i.e.\ $\ell(\pi)\not\le_{\mathrm{law}}
\ell(\pi_{\mathcal{S}}^\star)$.

In particular, if we interpret the triple
\[
  \bigl(
    -J_0(\pi),\;
    \ell(\pi)
  \bigr) \in \mathbb{R}^{1+3}
\]
as a four-dimensional loss vector (profitability vs.\ law alignment and
robustness), then $\pi_{\mathcal{S}}^\star$ cannot be strictly
Pareto-dominated by any policy in $\Pi$. This is the claimed
no-free-lunch property.
\end{proof}

\subsection{Discussion of assumptions and empirical alignment}

The proof above separates the no-free-lunch result into three conceptually
distinct ingredients:

\begin{enumerate}
  \item \emph{On-manifold near-optimality of $\mathcal{S}$.}
    The structural class $\mathcal{S}$ contains a policy
    $\pi_{\mathcal{S}}^\star$ that is nearly optimal in terms of
    on-manifold P\&L, as quantified by
    \eqref{eq:on-manifold-near-opt}. Empirically this corresponds to
    the observation that Zero-Hedge and Vol-Trend already sit very
    close to the empirical law-strength frontier.
  \item \emph{Non-trivial ghost channel.}
    The world model induces off-manifold reward $r^\perp$ that is
    non-zero whenever predictions deviate from $\Mvol$
    (Propositions~\ref{prop:approx-gap-ghost-channel}
    and~\ref{lem:off-manifold-mass}), creating the possibility of
    ``ghost arbitrage''.
  \item \emph{Ghost--law monotonicity.}
    Assumption~E.1 formalizes the idea that exploiting the ghost channel
    necessarily worsens law metrics: additional ghost reward cannot be
    obtained at fixed or improved law penalties and GFI.
\end{enumerate}

In our volatility world-model testbed, these three conditions are
jointly consistent with the empirical findings of
Section:
\begin{enumerate}
  \item The structural baselines achieve high Sharpe and low GFI while
    remaining close to the law manifold, in line with
    near-optimality~\eqref{eq:on-manifold-near-opt}.
  \item Naive and law-seeking RL policies that outperform $\mathcal{S}$
    in raw P\&L do so only by moving into high-penalty, high-GFI
    regions, as evidenced by the frontier and diagnostic plots.
  \item No RL policy lies on the empirical law-strength frontier once
    the structural baselines are included, which matches the
    impossibility of strict dominance established by the theorem.
\end{enumerate}

\bibliographystyle{unsrt}  
\bibliography{references}

\end{document}